\documentclass[onecolumn]{IEEEtran}
\usepackage{amsmath,amsfonts,amssymb}
\usepackage{amsthm}
\usepackage{bm}
\usepackage{mathtools}
\usepackage{hyperref}
\usepackage{tikz}
\usetikzlibrary{calc}
\usepackage{caption}
\allowdisplaybreaks
\newtheorem{theorem}{Theorem}
\newtheorem{proposition}{Proposition}
\newtheorem{assumption}{Assumption}
\newtheorem{definition}{Definition}
\newtheorem{lemma}{Lemma}
\newtheorem{corollary}{Corollary}
\newtheorem{remark}{Remark}
\DeclareMathOperator*{\argmin}{arg\,min}
\DeclarePairedDelimiter\floor{\lfloor}{\rfloor}
\DeclareCaptionType{step}[Step]
\newcommand{\nll}{\mathrel{\not{\mkern -7mu\ll}}}

\title{Hypothesis Testing of Mixture Distributions using Compressed Data}
\author{Minh Thanh Vu}
\begin{document}
\maketitle
\begin{abstract}
  In this paper we revisit the binary hypothesis testing problem with one-sided compression. Specifically we assume that the distribution in the null hypothesis is a mixture distribution of iid components. The distribution under the alternative hypothesis is a mixture of products of either iid distributions or finite order Markov distributions with stationary transition kernels. The problem is studied under the Neyman-Pearson framework in which our main interest is the maximum error exponent of the second type of error. We derive the optimal achievable error exponent and under a further sufficient condition establish the maximum $\epsilon$-achievable error exponent. It is shown that to obtain the latter, the study of the exponentially strong converse is needed. Using a simple code transfer argument we also establish new results for the Wyner-Ahlswede-K{\"o}rner problem in which the source distribution is a mixture of iid components.
\end{abstract}

\section{Introduction}
\subsection{Motivations \& Related Works}
Hypothesis testing with communication constraints is a classic problem in information theory. The problem was initiated by Ahswede and Csisz{\'a}r in \cite{ahlswede1986hypothesis} as well as by Berger in \cite{berger1979hypothesis}. It is assumed that a pair of data sequences $(x^n,y^n)$ is observed at separate locations. The sequence $x^n$ is compressed and sent to the location of $y^n$ via a noiseless channel. The decision center there decides whether $(x^n,y^n)$ is iid generated from the null hypothesis with a distribution $P_{XY}$ or from the alternative hypothesis with a distribution $Q_{XY}$. The Neyman-Pearson framework was used to study the trade-off between the probabilities of errors. The main interest was to establish the maximum error exponent of the second type of error when the probability of the first type of error is bounded as the number of samples tends to infinity. A single letter formulation was given in \cite{ahlswede1986hypothesis} for the testing against independence scenario and strong converse was proven for the general setting under a special condition. Various lower bounds for the general setting have been proposed in \cite{han1987hypothesis} and \cite{shimokawa1994error}. The work \cite{han1987hypothesis} also established the optimal error exponent in the zero-compression rate regime when $\epsilon$ is small enough. Using the same condition as in \cite{ahlswede1986hypothesis}, the work \cite{shalaby1992multiterminal} established the optimal $\epsilon$-error exponent also in the zero-compression rate regime. The optimality of the random binning scheme proposed in \cite{shimokawa1994error} was shown for the conditional independence setting in \cite{rahman2012optimality}. In \cite{tian2008successive} the authors extended the testing against independence study to the successive refinement setting. The work \cite{watanabe2017neyman} studied the non-asymptotic regime of the general setting under the two-sided zero-compression rate constraint. Yet, in all of the above studies, the distributions in both hypotheses are assumed to be iid.

Mixture distributions are prevalent in practice, cf. \cite{mclachlanpeel} for a comprehensive list of applications. However, hypothesis testing for mixture distributions is an under-examined direction in information theory. Notable studies are given in \cite{chen1996general,hanspectrum,han2018first}. Communication constraints are not included in the above works. In this work we study the following binary hypothesis problem in which the hypotheses are given by
\begin{align}
  H_0\colon P_{Y^nX^n} &= \sum_{i}\nu_i P_{Y_iX_i}^{\otimes n},\nonumber\\
  H_1\colon Q_{Y^nX^n} &= \sum_{jt}\tau_{jt} Q_{Y_j^n}\times Q_{X_t^n}.
\end{align}
Mixture of iid components is also known as mixture of repeated measurements or mixture of grouped observations in the literature. By no means exhaustive, we list a few works with applications in topic modeling \cite{ritchie2020consistent}, in cognitive psychology \cite{elmore2004estimating}, and in developmental psychology \cite{cruz2004semiparametric}, cf. also \cite{wei2020convergence}. In machine learning applications, to allow flexible modeling it is often assumed that within each component the joint distribution has a product form without the requirement of having the same marginal distribution, cf. \cite{pal2002noise} and \cite{anandkumar2012method}. We keep the iid condition in the null hypothesis for tractable analysis. A motivating example is provided in the following.

Assume that two statisticians observe two bags of words $x^n$ and $y^n$ taken as excerpts from some documents where $n$ denote the number of words in each bag. For simplicity we assume that in each bag words have an identical marginal distribution and also the order of words is not important. Then it is likely that the bag of words $x^n$ ($y^n$) is not generated iid, since for example knowing the first word $x_1$ ($y_1$) to be in Latin (German) gives us a guess that the whole $x^n$ ($y^n$) is in Latin (German) \cite{jordan2010lec}. Since the order of words does not matter, $x^n$ and $y^n$ are sequences of exchangeable observations. It is natural to approximate the distributions of $x^n$ and $y^n$ by finite mixtures of iid distributions, due to de Finetti's theorem \cite{kirsch2019elementary}, in which each underlying state represents a topic. The two statisticians then form a hypothesis testing problem. In the null hypothesis, they assume that the two bags of words are generated jointly iid according to $P_{Y_iX_i}^{\otimes n}$ from an unknown topic $i\in [1:m]$ with probability $\nu_i$, for example $x_{l}$ is a direct translation or a synonym of $y_l$ for all $l\in [1:n]$. In the alternative hypothesis they assume that the two bags of words are generated independently and iid according to $Q_{Y_j}^{\otimes n}\times Q_{X_t}^{\otimes n}$ from an unknown topic $(j,t)$ with probability $\tau_{jt}$.

We similarly are interested in the optimal error exponent of the second type of error when the probability of the first type is restricted by some $\epsilon\in [0,1)$. The iid assumption of each component in the null hypothesis as well as the factorization assumption in the alternative hypothesis are used to facilitate the derivation. In contrast to the classic strong converse result in \cite{ahlswede1986hypothesis}, the maximum $\epsilon$-achievable error exponent generally depends on not only the prior distribution $(\nu_i)$ but also relations between different constituent components in the mixture distributions. When no compression is involved, the weak law of large numbers can be used to establish the maximum $\epsilon$-achievable error exponent, cf. \cite{hanspectrum,han2018first}. However, this argument is no longer applicable in the presence of compression. We need to use an argument established through proving exponential strong converses of constituent components.
\subsection{Contributions}
We summarize contents of our work in the following.
\begin{itemize}
\item We use code transformation arguments to unveil a new property of the classic Ahlswede-Csisz{\'a}r problem in a general setting, cf. Theorem \ref{theorem_acs_extended} in Appendix \ref{ap_common_seq}. The arguments are also useful in simplifying proofs of different results in our work.
\item We provide the optimal achievable error exponent in Theorem \ref{thm_3} and show in Corollary \ref{coroll_1} that for $\epsilon>0$ small enough the obtained error exponent is also $\epsilon$-optimal. Our proof is based on studying a compound hypothesis testing problem of differentiating between $\{P_{Y_iX_i}^{\otimes n}\}$ and $\{Q_{Y_j^{n}}\times Q_{X_t^{n}}\}$ under communication constraints. The results are established when $\{(Q_{Y_j^{n}})\}$ and $\{(Q_{X_t^n})\}$ are stationary memoryless processes or finite Markov processes with stationary transition probabilities.
\item Under a further sufficient condition on the sets of distributions $\{P_{Y_iX_i}^{\otimes n}\}$ and $\{Q_{Y_j^n}\}$, Assumption \ref{assump_2},  we provide a complete characterization of the maximum $\epsilon$-error exponent in Theorem \ref{thm_4}. It is shown that even if the strong converse is available for the compound problem of testing $\{P_{Y_iX_i}^{\otimes n}\}$ against $\{Q_{Y_j^n}\times Q_{X_t^n}\}$, it is not sufficient for establishing the maximum $\epsilon$-error exponent in our mixture setting.  Our derivation is based on the exponentially strong converse of the false alarm probabilities under the assumed condition.
  \item We refine a recently established connection in \cite{vu2021hypothesis} between the Wyner-Ahlswede-K{\"o}rner (WAK) problem \cite{wyner1975source,ahlswede1975source} and the hypothesis testing against independence problem. While the previous result holds only under the stationary ergodic assumption, our new connection in Theorem \ref{thm_5} is valid under a more general assumption. We then use the refined connection to establish the corresponding minimum achievable compression rate and the minimum $\epsilon$-achievable compression rate for the WAK problem with mixture distributions.
  \end{itemize}
 \subsection{Organization}
  Our paper is organized as follows. In Section \ref{sec_2} we review previous results and define quantities which are essential to characterize optimal error exponents in our study. We then establish various results for the compound setting, some might be of independent interest in Section \ref{sec_3}. In Section \ref{sec_4}, we provide the maximum achievable error exponent in the mixture setting. We state the sufficient condition and use it to establish the maximum $\epsilon$-achievable error exponent in the mixture setting in Section \ref{sec_5}. Then the refined connection between the WAK problem and the hypothesis testing against independence problem is given in Section \ref{sec_6}. Consequences of the new connection are also given therein.
  \subsection{Notations}
 We focus on finite alphabets in this paper. Before we begin we make the following conventions. Given a probability measure $\mu$, $\mu^{\otimes n}$ denotes its $n$-fold product measure extension. $\log(·)$ denotes the natural logarithm. For any two distributions $P$ and $Q$ on an alphabet $\mathcal{U}$, assume that $Q(u)=0$, if $P$ is absolutely continuous w.r.t $Q$, denoted by $P\ll Q$, holds then we define $\iota_{P\Vert Q}(u) = 0$, otherwise we define $\iota_{P\Vert Q}(u) = +\infty$, irrespective of whether $P(u)$ is equal to 0 or not. We also define\footnote{For simplicity we use the convention $\log 0 = -\infty$ in the following.} $\iota_{P\Vert Q}(u) = \log P_{Y}(u)/Q_{Y}(u)$ when $Q(u)>0$. The relative entropy between two distributions $P$ and $Q$ is defined as $D(P\Vert Q) = \mathbb{E}[\iota_{P\Vert Q}(U)]$ where $U\sim P$. For a joint distribution $P_{UV}$ on $\mathcal{U}\times\mathcal{V}$ if $P_{U}(u)=0$ or $P_{V}(v)=0$ holds then we define $\iota_{P_{UV}}(u;v)=0$ as $P_{UV}\ll P_{U}\times P_{V}$ is valid. Otherwise we define $\iota_{P_{UV}}(u;v) = \log P_{UV}(u,v)/(P_{U}(u)\times P_{V}(v))$. The mutual information between $U$ and $V$ that is jointly distributed according to $P_{UV}$ is defined as $I(U;V) = \mathbb{E}[\iota_{P_{UV}}(U;V)]$. For a finite set $\mathcal{A}$, we use $\vert  \mathcal{A}\vert  $ and $\mathcal{A}^c$  to denote its cardinality and its complement. For a mapping $\phi\colon\mathcal{X}\to\mathcal{M}$ we define $\vert\phi\vert$ to be the cardinality of its range, $\vert\phi\vert \triangleq \vert \mathcal{M} \vert$. For a given distribution $P$ on $\mathcal{X}$ and a stochastic mapping $f\colon \mathcal{X}\to \{0,1\}$ with the corresponding transition kernel $W$ we define 
 \begin{align}
 P(f) \triangleq \sum_xP(x)W(0\vert  x),\; \text{and}\; P(1-f) \triangleq 1-P(f).
 \end{align}
\section{Preliminaries}\label{sec_2}
We review previous results on the hypothesis testing problem with one-sided compression. Then we present an important assumption and essential quantities that are needed to establish results of our study in later sections.
\begin{figure}[htb]
\centering
\includegraphics{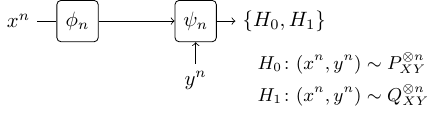}
\caption{Illustration of the hypothesis with communication constraint setting.}\label{as_fig}
\end{figure}

In \cite{ahlswede1986hypothesis}, the authors studied the following binary hypothesis testing problem. Deciding whether $(x^n,y^n)$ is iid generated from $H_0\colon P_{XY}$ or $H_1\colon Q_{XY}$ using a testing scheme $(\phi_n,\psi_n)$. Herein $\phi_n$ is a compression mapping,
\begin{equation}
  \phi_n\colon\mathcal{X}^n\to\mathcal{M},\label{testing_sch_a}
\end{equation}
and $\psi_n$ is a decision mapping
\begin{equation}
  \psi_n\colon\mathcal{Y}^n\times \mathcal{M}\to \{0,1\}.\label{testing_sch_b}
\end{equation}
The setting is depicted in Fig. \ref{as_fig}.
The corresponding type I and type II error (also known as false alarm and miss detection) probabilities are given by  
\begin{align}
  \alpha_n &= P_{Y^n\phi_n(X^n)}(1-\psi_n) = \sum P_{Y^n\phi_n(X^n)}(y^n,\phi_n(x^n))P_{1\vert  y^n,\phi_n(x^n)}\nonumber\\
  \beta_n &= Q_{Y^n\phi_n(X^n)}(\psi_n) = \sum Q_{Y^n\phi_n(X^n)}(y^n,\phi_n(x^n))P_{0\vert  y^n,\phi_n(x^n)},\label{error_probs_def}
\end{align}
where $P_{k\vert  y^n,\phi_n(x^n)}$, $k\in \{0,1\}$, is the probability that $\psi_n$ outputs $k$ given the pair $(y^n,\phi_n(x_n))$. The following achievability definition is repeatedly used in the subsequent analysis.
\begin{definition}\label{def_1}
For a given $R_c$ and an $\epsilon \in [0,1)$, $E$ is an $\epsilon$-achievable error exponent of the second type for the binary hypothesis testing problem if there exists a sequence of testing schemes $(\phi_n,\psi_n)$ such that

  \begin{align}
    \limsup_{n\to\infty}\alpha_n\leq \epsilon,\; &\liminf_{n\to\infty}\frac{1}{n}\log\frac{1}{\beta_n}\geq E,\nonumber\\
    \limsup_{n\to\infty}\frac{1}{n}&\log\vert  \phi_n\vert  \leq R_c.\label{conds_def_1}
  \end{align}
We define $E^{\star}_{\epsilon}(R_c)$ as the supremum of all $\epsilon$-achievable error exponent at $R_c$.
\end{definition}
When $Q(y\vert  x)>0$ for all $(x,y)\in\mathcal{X}\times\mathcal{Y}$ holds, Ahlswede and Csisz{\'a}r proved the strong converse result that for a given $R_c$, $E_{\epsilon}^{\star}(R_c)$ does not depend on $\epsilon\in[0,1)$. They also provided a multi-letter formula for $E_{\epsilon}^{\star}(R_c)$ under the stated condition. In the case that $Q_{XY} = P_{X}\times P_{Y}$ holds, the formula reduces to\footnote{For notation simplicity, we suppress the cardinality bound for $\mathcal{U}$ in the subsequent.}
\begin{align}
  E_{\epsilon}^{\star}(R_c) = \max_{P_{U\vert  X}\colon I(X;U)\leq R_c}I(Y;U),\;\forall \epsilon \in [0,1).
\end{align}
A closer examination also reveals that in the case $Q_{XY} = Q_Y\times Q_{X}$ holds, where $Q_Y$ and $Q_X$ are distributions on $\mathcal{Y}$ and $\mathcal{X}$ satisfying $D(P_Y\Vert Q_Y)<\infty$ and $D(P_X\Vert Q_X)<\infty$, we also have\footnote{In \cite{shalaby1992multiterminal}, the authors claimed in the introduction that Ahswede and Csisz{\'a}r obtained a single-letter characterization when $Q_{XY}=Q_X\times Q_Y$ using the entropy characterization methods. However, in \cite{ahlswede1986hypothesis} only characterizations for $Q_{XY} = P_X\times P_Y$ and an arbitrary $Q_{XY}$  with $R_c\geq H(X)$ were provided. Further details to validate this claim were not provided in \cite{shalaby1992multiterminal}.} for all $\epsilon\in [0,1)$
\begin{equation}
  E_{\epsilon}^{\star}(R_c) =  \max_{P_{U\vert  X}\colon I(X;U)\leq R_c}I(Y;U) + D(P_Y\Vert Q_Y) + D(P_X\Vert Q_X).
\end{equation}
In other words, Theorem 5 in \cite{ahlswede1986hypothesis} is tight in this case. We give a general relation using code transformations which leads to this observation in Appendix \ref{ap_common_seq}.\\
\noindent This observation motivates us to study in this paper generalized problems of compound hypothesis testing and hypothesis testing for mixture distributions. These settings are based on finite collections of distributions $\{P_{X_iY_i}^{\otimes n}\}_{i=1}^m$ and $\{Q_{Y_j^n}\times Q_{X_t^n}\}_{j\in [1:k],t\in [1:r]}$. Let 
\begin{align}
\mathcal{P}_{\mathcal{X}}=\{\tilde{P}\mid \tilde{P}\; &\text{is the marginal distributions on}\; \mathcal{X}\nonumber\\
& \text{ of } P_{X_iY_i} \;\text{for some } i\in [1:m]\}.
\end{align}
We assume that $\mathcal{P}_{\mathcal{X}}$ can be enumerated as $\{P_{\mathcal{X},1},\dots, P_{\mathcal{X},\vert \mathcal{P}_{\mathcal{X}}\vert}\}$. For each $s\in [1:\vert \mathcal{P}_{\mathcal{X}}\vert ]$, let $\mathfrak{F}_s$ be the set of indices $i$ such that $P_{\mathcal{X},s}$ is a the marginal distribution of $P_{Y_iX_i}$ on $\mathcal{X}$, i.e.,
\begin{equation}
\mathfrak{F}_s = \{i\mid P_{\mathcal{X},s} \text{ is the marginal on }\mathcal{X} \text{ of } P_{X_iY_i}\}.\label{fs_def}
\end{equation}
To characterize the optimal error exponents in the compound setting and the mixture setting, we need to make the following assumption.
\begin{assumption}\label{assump_1}
The processes $\{Q_{X_t^n}\}$ and $\{Q_{Y_j^n}\}$ are assumed to be in one of the following categories.  
\begin{itemize}
\item $Q_{X_t^n} = Q_{X_t}^{\otimes n}$ and $Q_{Y_j^n} = Q_{Y_j}^{\otimes n}$ hold for all $j\in [1:k]$ and $t\in [1:r]$, respectively. We assume further that for all $i\in [1:m]$, $\min_{j}D(P_{Y_i}\Vert Q_{Y_j})<\infty$ and for all $s\in [1:\vert \mathcal{P}_{\mathcal{X}}\vert ]$, $\min_{t}D(P_{\mathcal{X},s}\Vert Q_{X_t})<\infty$ hold. Accordingly, we define the minimum distances as\footnote{In case that $\argmin$ does not return a unique value, we pick the first one.}

\begin{align}
  \forall i\in [1:m],\; j_i^{\star} &= \argmin_{j\in[1:k]}D(P_{Y_i}\Vert Q_{Y_j}),\; d_i^{\mathrm{y}} = D(P_{Y_i}\Vert Q_{Y_{j_i^{\star}}}),\; \nonumber\\
  \forall s\in [1:\vert \mathcal{P}_{\mathcal{X}}\vert ],\; t_s^{\star} &= \argmin_{t\in[1:r]}D(P_{\mathcal{X},s}\Vert Q_{X_t}),\;d_s^{\mathrm{x}} = D(P_{\mathcal{X},s}\Vert Q_{X_{t_s^{\star}}}).
\end{align}
\item $(Q_{X_t^n})$ and $(Q_{Y_j^n})$ are Markov processes of finite orders with stationary transition probabilities for all $t\in [1:r]$ and $j\in [1:k]$. Compared to the first category we need an additional assumption which is for all $n$, we have $P_{\mathcal{X},s}^{\otimes n}\ll Q_{X_t^n}$ for all $(s,t)$ and similarly $P_{Y_i}^{\otimes n}\ll Q_{Y_j^n}$ for all $(i,j)$. We define the following limits, called relative entropy rates,

  \begin{align}
\forall (s,t),\;  A_{X}^{(st)} &= \lim_{n\to\infty}[D(P_{\mathcal{X},s}^{\otimes (n+1)}\Vert Q_{X_t^{n+1}}) - D(P_{\mathcal{X},s}^{\otimes n}\Vert Q_{X_t^n})],\nonumber\\
 \forall (i,j),\; A_{Y}^{(ij)} &= \lim_{n\to\infty}[D(P_{Y_i}^{\otimes (n+1)}\Vert Q_{Y_j^{n+1}}) - D(P_{Y_i}^{\otimes n}\Vert Q_{Y_j^n})].
  \end{align}
  The existence of these (relative divergence rate) limits are guaranteed by \cite[Theorem 1]{barron1985strong}. Similarly, we assume that for all $i\in [1:m]$, $\min_jA_{Y}^{(ij)}<\infty$ and for all $s\in [1:\vert \mathcal{P}_{\mathcal{X}}\vert ]$, $\min_tA_X^{(st)}<\infty$ hold. We further define
  \begin{align}
\forall i\in [1:m],\;    j_i^{\star} &= \argmin_{j\in[1:k]}A_{Y}^{(ij)},\; d_i^{\mathrm{y}} = A_{Y}^{(ij_i^{\star})},\nonumber\\
\forall s\in [1:\vert \mathcal{P}_{\mathcal{X}}\vert ],\;    t_s^{\star} &= \argmin_{t\in[1:r]} A_{X}^{(st)},\; d_s^{\mathrm{x}} = A_{X}^{(st_s^{\star)}}. 
   \end{align}
 \end{itemize}
\end{assumption}

\noindent For notation simplicity we also define total minimum distances for all $i\in [1:m]$, $i\in\mathfrak{F}_s$,
\begin{align}
  d_{is}^{\mathrm{yx}} = d_i^{\mathrm{y}} + d_s^{\mathrm{x}}.
\end{align}
The following quantities are used to characterize optimal error exponents in later sections\footnote{\noindent For clarity, we reserve pairs of random variables $(Y_i,X_i)$ for distributions $P_{Y_iX_i}$ and pairs of random variables $(\bar{Y}_i,\bar{X}_s)$ for distributions $P_{Y_i\vert  X_i}\times P_{\mathcal{X},s}$.}
\begin{align}
 \forall i\in [1:m],\; i\in \mathfrak{F}_s,\; \xi_{i}(R_c) &= \max_{P_{U\vert  \bar{X}_s}\colon I(\bar{X}_s;U)\leq R_c}I(\bar{Y}_i;U) + d_{is}^{\mathrm{yx}},\nonumber\\
 \forall s\in [1:\vert \mathcal{P}_{\mathcal{X}}\vert ],\;  \theta_s(R_c) &= \max_{P_{U\vert  \bar{X}_s}\colon I(\bar{X}_s;U)\leq R_c}\min_{i\in\mathfrak{F}_s} [I(\bar{Y}_i;U) + d_{is}^{\mathrm{yx}}],
\end{align}
where $P_{\bar{X}_s} = P_{\mathcal{X},s}$ and for all $i\in\mathfrak{F}_s$, $P_{\bar{Y}_i\vert  \bar{X}_s} = P_{Y_i\vert  X_i} $ hold.

\section{Compound Hypothesis Testing}\label{sec_3}
In this section we study the compound hypothesis testing problem of differentiating between two collections of distributions $\{P_{Y_iX_i}^{\otimes n}\}_{i\in [1:m]}$ and $\{Q_{Y_j^n}\times Q_{X_t^{n}}\}_{j\in [1:k], t\in [1:r]}$ satisfying Assumption \ref{assump_1}. We derive the optimal achievable error exponent and provide partial strong converse results for this setting.\\
A testing scheme for this problem is similarly characterized by a pair of compression-decision mappings $(\phi_n,\psi_n)$. For a given testing scheme $(\phi_n,\psi_n)$, we define the following quantities which characterize the maximum type I and type II error probabilities
\begin{align}
  \alpha_n &= \max_{i\in [1:m]}\alpha_n^{(i)} = \max_{i\in [1:m]}P_{Y_i^n\phi_n(X_i^n)}(1-\psi_n),\nonumber\\
  \beta_n &= \max_{j\in [1:k],t\in [1:r]}\beta_n^{(jt)}= \max_{j\in [1:k],t\in[1:r]}Q_{Y_j^n}\times Q_{\phi_n(X_t^n)}(\psi_n).
\end{align}
Similarly we also use Definition \ref{def_1} as the $\epsilon$-achievability definition. For a given $R_c$, the maximum $\epsilon$-achievable error exponent is denoted by $E_{\mathrm{comp},\epsilon}^{\star}(R_c)$.\\
As our achievability result will be used in Sections \ref{sec_4} and \ref{sec_5}, for notation compactness, we define the following auxiliary sets. For a given mapping $\phi_n$ and a positive number $E$, define for each $i\in [1:m]$, an intersected decision set $\mathcal{I}_{n}^{(i)}(E)$ as follows
\begin{align}
  \mathcal{I}_{n}^{(i)}(E) = &\{(y^n,\phi_n(x^n))\mid P_{Y_i^n\phi_n(X_i^n)}(y^n,\phi_n(x^n))\nonumber\\
  &\hspace{1cm}>e^{nE} \max_{j\in [1:k],t\in[1:r]}Q_{Y_j^n}\times Q_{\phi_n(X_t^n)}(y^n,\phi_n(x^n))\}.\label{ie_expl}
\end{align}
In the case $k=1$ and $r=1$ hold, $\mathcal{I}_{n}^{(i)}(E)$ is a decision region based on likelihood ratio for testing $P_{Y_i^n\phi_n(X_i^n)}$ against $Q_{Y_1^n}\times Q_{\phi_n(X_1^n)}$.
\subsection{Characterization of $E_{\mathrm{comp},0}^{\star}(R_c)$}
 The following result characterizes the optimal achievable error exponent in the compound hypothesis testing problem.
\begin{theorem}\label{thm_1}
  For a given compression threshold $R_c$, we have $E_{\mathrm{comp},0}^{\star}(R_c) = \min_{s\in [1:\vert \mathcal{P}_{\mathcal{X}}\vert ]}\theta_s(R_c)$. Furthermore,  
  for a given positive number $\gamma$ and any sequence of testing schemes $(\phi_n,\psi_n)$ such that $\bar{E}=\min_{s\in [1:\vert \mathcal{P}_{\mathcal{X}}\vert ]}\theta_s(R_c)-\gamma$ is achievable we also have with $E = \bar{E}-\gamma$
  \begin{align}
    \lim_{n\to\infty}P_{Y_i^n\phi_n(X_i^n)}\big[(\mathcal{I}_{n}^{(i)}(E))^c\big] = 0,\;\forall i \in [1:m].
  \end{align} 
\end{theorem}
The achievability proof of Theorem \ref{thm_1} is given in Appendix \ref{proof_thm_1}. The converse proof of Theorem \ref{thm_1} follows from the one of Theorem \ref{thm_3}. The proof of Theorem \ref{thm_1} uses a combination of techniques from \cite{hanspectrum}, \cite{verdu2012non} and our new mixing idea. In the following we provide an overview of steps in the proof of Theorem \ref{thm_1}.

Assume that the set of marginal distributions on $\mathcal{X}$, $\mathcal{P}_{\mathcal{X}}$, consists of a single element. Assume further that the number of components in the null hypothesis is two, i.e. $m=2$ and $\mathfrak{F}_1= \{1,2\}$. First we check whether the sequence $x^n$ is ``typical" in the sense that 
\begin{align}
\min_t \iota_{P_X^{\otimes n}\Vert Q_{X_t^n}}(x^n)>n(d_1^{\mathrm{x}}-\gamma).\label{change_qxto_px}
\end{align}
This helps us to perform the change of measure step from $Q_{X_t^n}$ to  $P_{X}^{\otimes n}$ in the analysis of the type-II (or miss detection) probability. The above condition is violated with vanishing probability in the analysis of the false alarm probability. We then select a test channel $P_{U\vert  X}$ and generate a codebook from the marginal distribution $P_{U}^{\otimes n}$. 
In our proof we do not estimate $i$ from the sequence $y^n$  to avoid potential complications in the analysis of the miss detection probability. We \textit{artificially create} instead the following joint distribution
\begin{align}
P_{Y^nU^n} = \sum_{i=1}^2 \nu_i P_{Y_iU}^{\otimes n},
\end{align}
where $\nu_i$ are positive probability weights. $P_{Y^nU^n}$ shifts the burden from calculating the miss detection probabilities to bounding the false alarm probabilities, which is less complex. We then consider the following score function which is helpful in defining a deterministic decision mapping $\psi_n$
\begin{align}
\zeta(y^n,u^n) = \iota_{P_{Y^nU^n}}(y^n;u^n) + \min_{j}\iota_{P_{Y^n\Vert Q_{Y_j^n}}}(y^n).\label{change_qto_pyu}
\end{align}
 In the score function $\zeta(\cdot,\cdot)$ the first term resolves the uncertainty within the set of marginal distributions $\{P_{Y_i}^{\otimes n}\}$, while the second term resolves the mismatch between two sets of distributions $\{P_{Y_i}^{\otimes n}\}$ and $\{Q_{Y_t^n}\}$. The second term also \textit{indirectly} checks whether $y^n$ is ``typical''.\\
Given a chosen codeword, which we explain how to obtain later, we decide that the null hypothesis is true if $\zeta(y^n,u^n)>n(E-d_1^{\mathrm{x}})$ is fulfilled. Given this decision the miss detection probabilities can be deduced based on the following chain of measure changing steps
\begin{align}
Q_{Y_j^n}(y^n)\stackrel{\eqref{change_qto_pyu}}{\to} P_{Y^n\vert  U^n}(y^n\vert  u^n),\; Q_{\phi_n(X_t^n)}(u^n) \stackrel{\eqref{change_qxto_px}}{\to} P_{\phi(X^n)}(u^n),
\end{align}
as well as the fact that $\sum_{y^n,u^n}P_{Y^n\vert  U^n}(y^n\vert  u^n)P_{\phi(X^n)}(u^n)\leq 1$ holds, where the summation is performed over the decision region.\\
Now to obtain a transmission message index, we search for a codeword such that it yields the lowest maximum false alarm conditional probabilities by looking at
\begin{align}
\max_{i\in\mathfrak{F}_s}\mathrm{Pr}\{\zeta(Y_i^n,u^n)<n(E-d_1^{\mathrm{x}})\vert  X_i^n=x^n\}.
\end{align}
In the analysis of the maximum false alarm probabilities changing measure from $P_{Y^nU^n}$ to $P_{Y_i^nU_i^n}$ in the expression of $\zeta(\cdot,\cdot)$ is relatively standard, cf. \cite{hanspectrum}. We then can use the standard typical arguments to conclude the existence of a good codebook. 
In the general case where the marginal set $\mathcal{P}_{\mathcal{X}}$ has multiple elements, we need to estimate $s$. Because of the way that we design $\zeta(\cdot,\cdot)$ this extra step does not affect the exponent of the miss detection probability.

\subsection{A partial characterization of $E_{\mathrm{comp},\epsilon}^{\star}(R_c)$}
\noindent We have the following result which provides a partial characterization of $E_{\mathrm{comp},\epsilon}^{\star}(R_c)$.
\begin{theorem}\label{thm_2}
  Given a positive number $R_c$, define the inactive set $\mathcal{S} = \{s\mid \theta_s(R_c)<\min_{i\in\mathfrak{F}_s}\xi_i(R_c)\}$.
  \begin{itemize}
  \item If\footnote{We use the following convention: if $\mathcal{S}=\varnothing$ then $\min_{s\in\mathcal{S}}(\cdot) = +\infty$ and $\max_{s\in\mathcal{S}}(\cdot) = -\infty$.} $\epsilon<\min\{\min_{s\in \mathcal{S}}\frac{1}{\vert  \mathfrak{F}_s\vert},1\}$ holds, then we have $$E_{\mathrm{comp},\epsilon}^{\star}(R_c) = \min_{s\in [1:\vert \mathcal{P}_{\mathcal{X}}\vert ]}\theta_s(R_c).$$
\item If $\epsilon>\max\{\max_{s\in\mathcal{S}}\frac{\vert  \mathfrak{F}_s\vert  -1}{\vert  \mathfrak{F}_s\vert}, 0\}$, then we have $$E_{\mathrm{comp},\epsilon}^{\star}(R_c) = \min_{i\in [1:m]}\xi_i(R_c).$$
\item Let $s^{\star}$ be an optimality achieving index, i.e., $\theta_{s^{\star}}(R_c) = \min_{s^{\prime}}\theta_{s^{\prime}}(R_c)$. Assume that $s^{\star}$ is active, i.e., $\theta_{s^{\star}}(R_c) = \min_{i\in \mathfrak{F}_{s^{\star}}}\xi_i(R_c)$. For an arbitrarily given $\gamma>0$, for any sequence of testing schemes $(\phi_n,\psi_n)$ such that the following inequalities are satisfied
\begin{align}
  \limsup_{n\to\infty}\frac{1}{n}\log\vert  \phi_n\vert  \leq R_c,\;\liminf_{n\to\infty}\frac{1}{n}\log\frac{1}{\beta_n}\geq \min_s\theta_s(R_c)+\gamma,
\end{align}
we have then $\lim_{n\to\infty}\alpha_n^{(i^{\star})} = 1$, where $i^{\star}\in\mathfrak{F}_{s^{\star}}$ is an index such that $\xi_{i^{\star}}(R_c) =\theta_{s^{\star}}(R_c)$. It also implies that under this assumption we have $E_{\mathrm{comp},\epsilon}^{\star}(R_c) = \min_s\theta_s(R_c)$ for all $\epsilon\in [0,1)$.
\end{itemize}
\end{theorem}
A simple case in which the third statement in Theorem \ref{thm_2} holds, is when $\vert  \mathfrak{F}_s\vert  =1$ for all $s\in [1:\vert \mathcal{P}_{\mathcal{X}}\vert ]$. The proof of Theorem \ref{thm_2} is given in Appendix \ref{proof_thm_2}. The second part of Theorem \ref{thm_4} will be employed in proving the converse of Theorem \ref{thm_4} for mixture models in Section \ref{sec_4}. 

Recall that $\mathfrak{F}_s$ represents the set of distributions $P_{Y_iX_i}$ which have the same marginal distribution $P_{\mathcal{X},s}$ on $\mathcal{X}$, cf. Equation \eqref{fs_def}. In the following we provide an outline of the proofs of the first two points in the statement Theorem \ref{thm_2}. 
 
We discuss in this paragraph the first item in the first part of Theorem \ref{thm_2}. When $s$ is active, i.e., $\theta_s(R_c) = \min_{i\in\mathfrak{F}_s}\xi_i(R_c)$ holds, then it follows from the strong converse bound for testing $P_{Y_iX_i}^{\otimes n}$ against $Q_{Y_{j_i^{\star}}^n}\times Q_{X_{t_s^{\star}}^n}$ for each $i$ inside the class $\mathfrak{F}_s$, that $E_{\mathrm{comp},\epsilon}^{\star}(R_c)\leq \theta_s(R_c)$ holds for all $\epsilon\in [0,1)$, cf. Theorem \ref{theorem_acs_extended}. Therefore we only need to focus on the inactive set $\mathcal{S}$. For simplicity, in this discussion we can assume that $\mathcal{S} = \{1\}$ holds and there are two components inside $\mathfrak{F}_1$. We assume further that there is no mismatch, i.e., $\{Q_{Y_j^n}\} = \{P_{Y_i}^{\otimes n}\}$ and $\{Q_{X_t^n}\} = \{P_{X_i}^{\otimes n}\}$ hold. The formulation of $\theta_1(R_c)$ requires the selection of a test channel $P_{U\vert  \bar{X}_1}$. To show the strong converse bound, the general idea is hence to identify a common test channel $P_{U\vert  \bar{X}_1}$. This can be done by considering relevant sets $\mathcal{V}_i$, $i=1,2$, of $x^n$ such that for each $i=1,2,$ $$\mathrm{Pr}(\psi_n(Y_i^n,\phi_n(x^n))=0\vert  X_i^n=x^n)>\eta$$
holds, where $\eta\in (0,1-(\epsilon+\gamma))$ and $\gamma\in (0,1-\epsilon)$ are arbitrary. By setting $\eta=1/2-(\epsilon+\gamma)$ and using the reverse Markov inequality, we obtain the following inequalities
  \begin{equation}
    P_{X_i}^{\otimes n}(\mathcal{V}_i)\geq 1/(1+2(\epsilon+\gamma)),\; \forall i=1,2.
  \end{equation}
We require that the intersection $\mathcal{V} = \cap_{i=1}^2\mathcal{V}_i$ should be non-empty.  This allows us to define a joint distribution $P_{\tilde{Y}_1^n\tilde{Y}_2^n\tilde{X}^n}$ where $P_{\tilde{X}^n}$ is supported on $\mathcal{V}$. Note also that $\eta$-restriction in the definition of $\mathcal{V}_i$ allows us to obtain the following inequality
 \begin{equation*}
 n(E-\dots) \leq I(\tilde{Y}_1^n;\phi_n(\tilde{X}^n)) + (\text{a bounded function of $\eta$, $\epsilon$ and $\gamma$}).
 \end{equation*}
 With this we can identify $\tilde{U}_l = (\phi_n(\tilde{X}^n),(\tilde{Y}_i^{l-1})_{i=1}^2)$ for all $l\in [1:n]$. The variational arguments in \cite{tyagi2019strong} can be used to indeed show that $E_{\mathrm{comp},\epsilon}^{\star}(R_c)\leq \theta_1(R_c)$ also holds. To make $\mathcal{V}$ non-empty, we must have $\epsilon<1/2 = \frac{1}{\vert  \mathfrak{F}_1\vert}$. When the inactive set $\mathcal{S}$ contains more than one element we obtain the corresponding threshold $\min_{s\in\mathcal{S}}\frac{1}{\vert  \mathfrak{F}_s\vert}$. 

 Now we discuss about the second item in the first part of Theorem \ref{thm_2}. Since $E_{\mathrm{comp},\epsilon}^{\star}(R_c)\leq \xi_i(R_c)$ holds for all $i\in [1:m]$ and $\epsilon\in [0,1)$, cf. Theorem \ref{theorem_acs_extended}, we only explain the achievability direction of the second item. For simplicity we assume that the set of marginal distributions $\mathcal{P}_{\mathcal{X}}$ has a single element, the element is inactive $\mathcal{S}=\{1\}$, and $\mathfrak{F}_1= \{1,2\}$. Our achievability idea is to build two sequences of testing schemes separately and then mix them together. For this we need to divide the space $\mathcal{X}^n$ into $\vert  \mathfrak{F}_1\vert   = 2$ partitions $\mathcal{C}_1$ and $\mathcal{C}_2$ such that $P_{\bar{X}_1}^{\otimes n}(\mathcal{C}_l)> 1-\epsilon$ for all $l=1,2$ for all sufficiently large $n$. Since $\epsilon>\frac{\vert  \mathfrak{F}_1\vert  -1}{\vert  \mathfrak{F}_1\vert  }=1/2$, such partitioning can be done. For each $i\in \{1,2\}$, we design a sequences of testing schemes $(\phi_n^{1i},\psi_n^{1i})$ to differentiate between $P_{Y_iX_i}^{\otimes n}$ and $Q_{Y_{j_i^{\star}}^n}\times Q_{X_{t_1^{\star}}^n}$ such that $\xi_i(R_c)-\gamma$ is achievable.
As in the proof of Theorem \ref{thm_1} we also define an auxiliary mixture distribution
\begin{align}
P_{Y^nX^n} = \sum_{i}\nu_i P_{Y_iX_i}^{\otimes n}.
\end{align}
The mixture distribution helps to alleviate the estimation of the distribution of $y^n$. Once the preparation is complete, we perform the compression as follows. \\
We first check if $x^n$ is a typical sequence in the sense that whether the following condition is fulfilled or not
\begin{align}
\min_{t}\iota_{P_{X^n}\Vert Q_{X_t^n}}(x^n)>n(d_{1}^{\mathrm{x}}-\gamma).
\end{align}
Again the above inequality also helps resolving the mismatch between two sets of distributions $\{P_{X_i}^{\otimes n}\}$ and $\{Q_{X_t^n}\}$.
Suppose that $x^n$ is typical. If $x^n\in \mathcal{C}_1$ then we use $\phi_n^{11}$ to compress it, and similarly when $x^n\in\mathcal{C}_2$ we use $\phi_n^{12}$ to compress it. The joint compression mapping is then
  \begin{align}
    \phi_n^1(x^n) = \phi_{n}^{11}(x^n)\mathbf{1}\{x^n\in\mathcal{C}_1\} + \phi_{n}^{12}(x^n)\mathbf{1}\{x^n\in\mathcal{C}_2\}.
    \end{align}
The joint compression mapping induces the following distribution from $P_{Y^nX^n}$
  \begin{align}
    P_{Y^nU} = \nu_1P_{Y_1^n\phi_n^{11}(X_1^n)} + \nu_2 P_{Y_2^n\phi_n^{12}(X_2^n)}.
  \end{align}
  We also define the following score function
  \begin{align}
\zeta(y^n,\phi_n^{1}(x^n)) = \iota_{P_{Y^nU}}(y^n;\phi_n^{1}(x^n)) + \min_{j}\iota_{P_{Y^n\Vert Q_{Y_j^n}}}(y^n).
\end{align}
We say $H_0$ is true if $\zeta(y^n,\phi_n^{1}(x^n))>n(E-d_1^{\mathrm{x}})$ holds. Let us look at $\alpha_n^{(1)}$ which can be upper-bounded as
\begin{align}
\alpha_n^{(1)}&\leq \mathrm{Pr}\{\zeta(Y_1^n,\phi_n^1(X_1^n))>n(E-d_1^{\mathrm{x}}), X_1^n \text{ is typical}\}\nonumber\\
& + \mathrm{Pr}\{X_1^n \text{ is atypical}\}\nonumber\\
&\stackrel{(a)}{\leq} \mathrm{Pr}\{\zeta(Y_1^n,\phi_n^{11}(X_1^n))>n(E-d_1^{\mathrm{x}}), X_1^n \in\mathcal{C}_1\}\nonumber\\
& + \mathrm{Pr}\{X_1^n\notin\mathcal{C}_1\} + \mathrm{Pr}\{X_1^n \text{ is atypical}\}\nonumber\\
&\leq \mathrm{Pr}\{\zeta(Y_1^n,\phi_n^{11}(X_1^n))>n(E-d_1^{\mathrm{x}})\}\nonumber\\
& + \mathrm{Pr}\{X_1^n\notin\mathcal{C}_1\} + \mathrm{Pr}\{X_1^n \text{ is atypical}\},
\end{align}
where $(a)$ holds since when $x^n$ is in $\mathcal{C}_1$ and is typical, we use $\phi_n^{11}$ to compress it. We have $\mathrm{Pr}\{X_1^n\notin\mathcal{C}_1\}<\epsilon$ by construction. The first term can be shown to be vanishing using similar steps as in the proof of Theorem \ref{thm_1}.\\
 We note that Theorem \ref{thm_1} only guarantees that $\theta_1(R_c)-\gamma$ is achievable which is below $\min_{i\in\mathfrak{F}_1}\xi_i(R_c)-\gamma$, our desired error exponent in this part of Theorem \ref{thm_2}. The collection of sets $\{\mathcal{C}_l\}_{l=1}^2$ is used to resolve the confusion about which compression mappings we should use when $s=1$ is not active. We are willing to pay an additional $\epsilon$ error probability price for using this collection. The general case is a little bit more complicated but follows the same principles as we discuss herein. Note that when $s$ is active, we do not need to divide $\mathcal{X}^n$ into a collection of subsets as above. This is because by Theorem \ref{thm_1} we can design a sequence of testing schemes to differentiate between $\{P_{Y_iX_i}^{\otimes n}\}_{i\in\mathfrak{F}_s}$ and $\{Q_{Y_j^n}\times Q_{X_t^n}\}$ such that $\min_{i\in\mathfrak{F}_s}\xi_i(R_c)-\gamma$ is achievable.

\begin{remark}
We have $\theta_s(R_c)\leq \min_{i\in\mathfrak{F}_s}\xi_i(R_c)$. The inequality can be strict which can be shown numerically. This means that in general strong converse does not hold for the compound testing problem. In other words, $E_{\mathrm{comp},\epsilon}^{\star}(R_c)$ depends on $\epsilon$.
\end{remark}

\section{Testing Against Generalized Independence}\label{sec_4}
In this section we consider the hypothesis testing problem involving mixture distributions. We use results and techniques from Section \ref{sec_3} to establish the optimal achievable error exponent in this section. We begin with our model's definition.\\
Assume that we have two sets of distributions $\{P_{Y_iX_i}^{\otimes n}\}_{i=1}^m$ and $\{Q_{Y_j^n}\times Q_{X_t^n}\}_{j\in [1:k],t\in [1:r]}$, which fulfill the conditions given in Assumption \ref{assump_1}.
For given $\{P_{Y_iX_i}^{\otimes n}\}_{i=1}^m$, let the distribution under the null hypothesis be defined as 
\begin{align}
  P_{Y^nX^n}= \sum_{i}\nu_iP_{Y_iX_i}^{\otimes n},\;\text{where} \; \forall i,\;\nu_i>0,\;\text{and}\;\sum_{i}\nu_i = 1.
\end{align}
Similarly the distribution under the alternative hypothesis is given by
\begin{align}
Q_{Y^nX^n}= \sum_{j,t}\tau_{jt}Q_{Y_j^n}\times Q_{X_t^n},
\end{align}
where $\tau_{jt}\geq 0$ for all $(j,t)$, $\sum_{j,t}\tau_{jt} = 1$, and $\forall i\in [1:m],\; i\in\mathfrak{F}_s,\; \tau_{j_i^{\star}t_s^{\star}}>0$. For notation simplicity in the subsequent analysis we define $\gamma_q = \min_{i\in[1:m],i\in\mathfrak{F_s}}\tau_{j_i^{\star}t_s^{\star}}$. We name this problem testing against generalized independence.\\
The model of $Q_{Y^nX^n}$ subsumes the following two cases:
\begin{itemize}
\item testing against independence in which $Q_{Y^nX^n} = (\sum_{i}\nu_i P_{Y_i}^{\otimes n})\times (\sum_{i^{\prime}}\nu_{i^{\prime}}P_{X_{i^{\prime}}}^{\otimes n})$ hold,
\item and testing against (unobserved) conditional independence in which $Q_{Y^nX^n} = \sum_{i}\nu_i P_{Y_i}^{\otimes n}\times P_{X_i}^{\otimes n}$ hold.
\end{itemize} 
For a given pair of compression-decision mappings $(\phi_n,\psi_n)$, we define the corresponding type-I and type-II (false alarm and miss detection) probabilities as
\begin{align}
  \alpha_n &= P_{Y^n\phi_n(X^n)}(1-\psi_n),\nonumber\\
  \beta_n &= Q_{Y^n\phi_n(X^n)}(\psi_n).\label{alphabeta_defs}
\end{align}
Similarly as in Definition \ref{def_1} we say that $E$ is an $\epsilon$-achievable error exponent at a compression rate $R_c$ for testing $P_{Y^nX^n}$ against $Q_{Y^nX^n}$ if there exists a sequence of testing schemes $(\phi_n,\psi_n)$ such that all the conditions in  \eqref{conds_def_1} are satisfied. We denote the maximum $\epsilon$-achievable error exponent at the given rate $R_c$ by $E_{\mathrm{mix},\epsilon}^{\star}(R_c)$.
We first characterize the optimal achievable error exponent in the testing against generalized independence problem $E_{\mathrm{mix},0}^{\star}(R_c)$.
\begin{theorem}\label{thm_3}
  For a given compression rate $R_c$, in testing $P_{Y^nX^n}$ against $Q_{Y^nX^n}$ using one-sided compression, we have
  \begin{align}
    E_{\mathrm{mix},0}^{\star}(R_c) = E_{\mathrm{comp},0}^{\star}(R_c) = \min_{s\in[1:\vert \mathcal{P}_{\mathcal{X}}\vert ]}\theta_s(R_c).
  \end{align}  
\end{theorem}
We first provide a remark about the first equality in the statement of Theorem \ref{thm_3}. In this we highlight the difference between our model and a previous study.

Let us consider the case that $\tau_{jt}>0$ for all pairs $(j,t)$ holds. Assume that $E$ is achievable in the mixture setting via a sequence of testing schemes $(\phi_n,\psi_n)$. Since $P_{Y^n\phi_n(X^n)}(1-\psi_n)\to 0$ and $\nu_i>0$ for all $i$, we have $P_{Y_i^n\phi_n(X_i^n)}(1-\psi_n)\to 0$. Similarly, for an arbitrarily given $\gamma>0$ and for all sufficiently large $n$ we have $Q_{Y^n\phi_n(X^n)}(\psi_n)\leq e^{-n(E-\gamma)}$. Since for all $(j,t)$, $\tau_{jt}>0$ holds, we have
$Q_{Y_j^n\phi_n(X_t^n)}(\psi_n)\leq \frac{1}{\tau_{jt}}e^{-n(E-\gamma)}$. This implies $E$ is an achievable error exponent in the compound hypothesis testing problem with the corresponding sequence of testing schemes $(\phi_n,\psi_n)$. Hence $E_{\mathrm{mix},0}^{\star}(R_c)\leq E_{\mathrm{comp},0}^{\star}(R_c)$ holds. The arguments discussed herein are similar to the ones given in \cite{han2018first} when data are not compressed. In our proof of Theorem \ref{thm_3}, we only need the restriction that $\tau_{j_i^{\star}t_s^{\star}}>0$ for all $i\in [1:m]$ and $i\in\mathfrak{F}_s$.

We explain the idea of showing $E_{\mathrm{mix},0}^{\star}(R_c)\leq \min_{s\in[1:\vert \mathcal{P}_{\mathcal{X}}\vert ]}\theta_s(R_c)$ in the following. For simplicity assume that there is no mismatch, i.e., $\{Q_{X_t^n}\} = \{P_{X_s}^{\otimes n}\}$ as well as $\{Q_{Y_j^n}\} = \{P_{Y_i}^{\otimes n}\}$ hold. Assume that $E$ is an achievable error exponent via a sequence of testing schemes $(\phi_n,\psi_n)$. A central idea of the Neyman-Pearson framework is to consider a decision region based on the likelihood ratio. An advantage of working with a likelihood-based decision region is that elementary set operations such as intersection, contraction, etc. can be performed through simple change of measure steps either in the numerator or denominator of the likelihood ratio. We want to show that if $E$ is achievable then roughly
\begin{align}
\mathrm{Pr}\{\iota_{P_{Y_i^n\phi_n(X_i^n)}}(Y_i^n;\phi_n(X_i^n))<nE\}\to 0,\;\text{as}\; n\to\infty. \label{target_thm_3}
\end{align}
The term inside the bracket is a rejection region based on the likelihood ratio for testing $P_{Y_i^n\phi_n(X_i^n)}$ against $P_{Y_i^n}\times P_{\phi_n(X_i^n)}$. Then based on the definition of the spectral-inf mutual information rate as well as the fact that the spectral-inf mutual information rate is bounded by the inf-mutual information rate, we can arrive at a conclusion that for an arbitrarily given $\gamma>0$, and for all $i\in [1:m]$,
\begin{align}
E\leq \frac{1}{n}I(Y_i^n;\phi_n(X_i^n)) + \gamma
\end{align}
for all sufficiently large $n$. Then we can use the standard single-letterization method to obtain that $E\leq \min_{s\in[1:\vert \mathcal{P}_{\mathcal{X}}\vert ]}\theta_s(R_c)$. In order to obtain the conclusion in \eqref{target_thm_3}, we need to perform several change of measure steps. First we form a decision region $\mathcal{A}_n$ based on the likelihood ratio of $P_{Y^n\phi_n(X^n)}$ and $P_{Y^n}\times P_{\phi_n(X^n)}$ as well as the achievable error exponent $E$. Then it can be shown that $P_{Y^n\phi_n(X^n)}(\mathcal{A}_n^c)\to 0$. We then do the first change of measure step from $P_{Y^n\phi_n(X^n)}$ to $P_{Y_i^n\phi_n(X_i^n)}$ to obtain 
\begin{align}
P_{Y_i^n\phi_n(X_i^n)}(\mathcal{A}_n^c)\to 0.
\end{align}
This can be seen from the definition of $P_{Y^nX^n}$, as $\nu_i>0$. Next we need to change the measure inside the definition of $\mathcal{A}_n$. Roughly we want to show that the following inequality holds
\begin{align}
&\log\frac{P_{Y_i^n\phi_n(X_i^n)}}{P_{Y_i^n}\times P_{\phi_n(X_i^n)}}(\cdot,\cdot)+\text{extra  penality}\nonumber\\
&\geq \log\frac{P_{Y^n\phi_n(X^n)}}{P_{Y^n}\times P_{\phi_n(X^n)}}(\cdot,\cdot)\geq nE.
\end{align}
Changing the measures in the denominator from $P_{Y^n}$ to $P_{Y_i^n}$ and $P_{\phi_n(X^n)}$ to $P_{\phi_n(X_i^n)}$ can be done based on inequalities $P_{Y^n}(\cdot)\geq \nu_i P_{Y_i^n}(\cdot)$ and $P_{\phi_n(X^n)}\geq \nu_i P_{\phi_n(X_i^n)}(\cdot)$. These inequalities follow from the definition of $P_{Y^nX^n}$ and hold for all $y^n$ and $\phi_n(x^n)$. The change from $P_{Y^n\phi_n(X^n)}$ to $P_{Y_i^n\phi_n(X_i^n)}$ is not based on the definition of $P_{Y^nX^n}$. However we can show that it holds with high probability. The proof of the general case involves another change of measure step from $Q_{Y_{j_i^{\star}}^n}\times Q_{\phi_n(X_{t_s^{\star}}^n)}$ to $P_{Y_i^n}\times P_{\phi_n(X_i^n)}$ using our code transformation arguments.

\begin{proof}
Assume that $E$ is an achievable error exponent at a compression rate $R_c$ in the compound hypothesis testing problem via a sequence of testing schemes $(\phi_n,\psi_n)$. We have by definition 
\begin{align}
\lim_{n\to\infty}\max_{i\in [1:m]}\alpha_n^{(i)} &=  0,\nonumber\\
\liminf_{n\to\infty}\frac{1}{n}\log\frac{1}{\max_{(j,t)}\beta_n^{(jt)}}&\geq E.
\end{align}
Applying this sequence to the current testing against generalized independence setting we obtain
\begin{align}
P_{Y^n\phi_n(X^n)}(\psi_n) &= \sum_{i=1}^m \nu_i P_{Y_i^n\phi_n(X_i^n)}(\psi_n)\nonumber\\
 &\leq \max_{i\in [1:m]}P_{Y_i^n\phi_n(X_i^n)}(\psi_n)= \max_{i\in [1:m]}\alpha_n^{(i)},\nonumber\\
Q_{Y^n\phi_n(X^n)}(\psi_n) &= \sum_{j,t} \tau_{jt}Q_{Y_j^{n}}\times Q_{\phi_n(X_t^{n})}(\psi_n)\nonumber\\
&\leq \max_{(j,t)}Q_{Y_j^{n}}\times Q_{\phi_n(X_t^{n})}(\psi_n) = \max_{(j,t)}\beta_n^{(jt)}.
\end{align}
We obtain that 
\begin{equation}
\lim_{n\to\infty}P_{Y^n\phi_n(X^n)}(\psi_n) = 0,\; \liminf_{n\to\infty}\frac{1}{n}\log\frac{1}{Q_{Y^n\phi_n(X^n)}(\psi_n)}\geq E.
\end{equation}
Hence $E$ is also achievable in our testing against generalized independence setting, cf. Definition \ref{def_1}. Therefore we have
  \begin{align}
    E_{\mathrm{mix},0}^{\star}(R_c)\geq E_{\mathrm{comp},0}^{\star}(R_c).
    \end{align}
\noindent Now for an arbitrarily given $\gamma>0$, assume that $(\phi_n,\psi_n)$ is a sequence of testing schemes such that
\begin{equation}
  \lim_{n\to\infty}\alpha_n = 0,\;\liminf_{n\to\infty}\frac{1}{n}\log\frac{1}{\beta_n}\geq E + \gamma\label{eq_40}
\end{equation}
hold.
Define for each $n$ the following decision region based on the likelihood ratio
\begin{align}
\mathcal{A}_n = \{(y^n,\phi_n(x^n))\mid & P_{Y^n\phi_n(X^n)}(y^n,\phi_n(x^n))\nonumber\\
&\geq e^{nE}Q_{Y^n\phi_n(X^n)}(y^n,\phi_n(x^n))\}.
\end{align}
By \cite[Lemma 4.1.2]{hanspectrum} and the definition of $P_{Y^nX^n}$, we have for all $n$
\begin{align}
  \alpha_n + e^{nE}\beta_n&\geq P_{Y^n\phi_n(X^n)}(\mathcal{A}_n^c)\nonumber\\
  &= \sum_{i=1}^m\nu_i P_{Y_i^n\phi_n(X_i^n)}(\mathcal{A}_n^c)
\end{align}
\noindent
Let $(\gamma_n)$ be a sequence such that $\gamma_n\to 0$ and $n\gamma_n\to\infty$ as $n\to\infty$. For each $i\in [1:m]$, we define a set 
\begin{align}
\mathcal{G}_{n}^{(i)} = \{(y^n,\phi_n(x^n))\mid & P_{Y^n\phi_n(X^n)}(y^n,\phi_n(x^n))\nonumber\\
&<e^{n\gamma_n}P_{Y_i^n\phi_n(X_i^n)}(y^n,\phi_n(x^n))\}.\label{g_set}
\end{align}
We then have 
\begin{align}
  P_{Y_i^n\phi_n(X_i^n)}[(\mathcal{G}_{n}^{(i)})^c]\leq e^{-n\gamma_n}P_{Y^n\phi_n(X^n)}[(\mathcal{G}_{n}^{(i)})^c]\leq e^{-n\gamma_n}.
\end{align}
$\mathcal{G}_{n}^{(i)}$ contains high probability pairs when we perform the change of measure from $P_{Y^n\phi_n(X^n)}$ to $P_{Y_i^n\phi_n(X_i^n)}$ in the numerator of the likelihood ratio test in the definition of $\mathcal{A}_n$. To make the derivation more compact, we further define two following sets
\begin{align}
\mathcal{C}_{n}^{(i)} &= \{(y^n,\phi_n(x^n))\mid  P_{Y_i^n\phi_n(X_i^n)}(y^n,\phi_n(x^n))\nonumber\\
&\hspace{2cm}<e^{n(E-\gamma_n+\log\gamma_q/n)} Q_{Y_{j^{\star}_i}^n}\times  Q_{\phi_n(X_{t_s^{\star}}^n)}(y^n,\phi_n(x^n))\},\nonumber\\
\mathcal{D}_{n}^{(i)} &= \{(y^n,\phi_n(x^n))\mid P_{Y_i^n\phi_n(X_i^n)}(y^n,\phi_n(x^n))\nonumber\\
&\hspace{2cm}<e^{n(E-\gamma_n)} Q_{Y^n\phi_n(X^n)}(y^n,\phi_n(x^n))\}.\label{c_i_n_def}
\end{align}
$\mathcal{D}_n^{(i)}$ is a rejection region in testing $P_{Y_i^n\phi_n(X_i^n)}$ against $Q_{Y^n\phi_n(X^n)}$. $\mathcal{C}_n^{(i)}$ is a rejection region of testing $P_{Y_i^n\phi_n(X_i^n)}$ against $Q_{Y_{j^{\star}_i}^n}\times  Q_{\phi_n(X_{t_s^{\star}}^n)}$, which is our first desired test.
\noindent From the definition of $Q_{Y^nX^n}$, we know that for all pairs $(y^n,\phi_n(x^n))$, and for all $i\in [1:m]$ and $i\in\mathfrak{F}_s$, the following inequality holds
\begin{align}
  Q_{Y^n\phi_n(X^n)}(y^n,\phi_n(x^n))&\geq \gamma_q Q_{Y_{j^{\star}_i}^n}\times Q_{\phi_n(X_{t_s^{\star}}^n)}(y^n,\phi_n(x^n)).\label{reverse_beta}
\end{align}
This implies that $\mathcal{C}_{n}^{(i)}\subseteq\mathcal{D}_{n}^{(i)}$ holds. Furthermore for $(y^n,\phi_n(x^n))\in\mathcal{D}_{n}^{(i)}\cap\mathcal{G}_{n}^{(i)}$ we have
\begin{align}
P_{Y^n\phi_n(X^n)}(y^n,\phi_n(x^n))&\stackrel{\eqref{g_set}}{<} e^{n\gamma_n}P_{Y_i^n\phi_n(X_i^n)}(y^n,\phi_n(x^n))\nonumber\\
&\stackrel{\eqref{c_i_n_def}}{<}e^{nE} Q_{Y^n\phi_n(X^n)}(y^n,\phi_n(x^n))\nonumber\\
\Rightarrow (y^n,&\phi_n(x^n))\in\mathcal{A}_n^c.
\end{align}
Using the above analysis we perform in the following two change of measure steps
\begin{align}
  &P_{Y_i^n\phi_n(X_i^n)}(\mathcal{A}_n^c) + e^{-n\gamma_n}\nonumber\\
  &\stackrel{(a)}{\geq} P_{Y_i^n\phi_n(X_i^n)}(\mathcal{D}_{n}^{(i)}\cap\mathcal{G}_{n}^{(i)}) + P_{Y_i^n\phi_n(X_i^n)}[(\mathcal{G}_{n}^{(i)})^c]\nonumber\\
  &\geq P_{Y_i^n\phi_n(X_i^n)}(\mathcal{D}_{n}^{(i)})\nonumber\\
  &\stackrel{(b)}\geq P_{Y_i^n\phi_n(X_i^n)}(\mathcal{C}_{n}^{(i)}),
\end{align}
where
\begin{itemize}
\item in $(a)$ we change the measures from $P_{Y^n\phi_n(X^n)}$ to $P_{Y_i^n\phi_n(X_i^n)}$,
\item in $(b)$ we change the measures from $Q_{Y^n\phi_n(X^n)}$ to $Q_{Y_{j^{\star}_i}^n}\times Q_{\phi_n(X_{t_s^{\star}}^n)}$.
\end{itemize}
\noindent In summary we have
\begin{align}
  & \alpha_n + e^{nE}\beta_n +e^{-n\gamma_n}\nonumber\\
  &\geq \sum_{i=1}^m \nu_i P_{Y_i^n\phi_n(X_i^n)}(\mathcal{C}_{n}^{(i)}).
\end{align}
In combination with \eqref{eq_40}, since $e^{nE}\beta_n\to 0$ as $n\to\infty$ holds, we obtain
\begin{align}
\lim_{n\to\infty} P_{Y_i^n\phi_n(X_i^n)}(\mathcal{C}_{n}^{(i)})= 0,\;\forall i\in [1:m].\label{before_transform}
\end{align}
In the next step, for each $i\in [1:m]$, and $i\in \mathfrak{F}_s$, we will perform change of measure from $Q_{Y_{j^{\star}_i}^n}\times Q_{\phi_n(X_{t_s^{\star}}^n)}$ to $P_{Y_i^n}\times P_{\bar{\phi}_n^{s}(X_i^n)}$. $\bar{\phi}_n^{s}$ is a compression mapping for each class $\mathfrak{F}_s$ that is constructed from $\phi_n$.

For a given $i\in [1:m]$, consider the problem of differentiating between $P_{Y_iX_i}^{\otimes n}$ and $Q_{Y_{j_i^{\star}}^n}\times Q_{X_{t_s^{\star}}^n}$, where $i\in\mathfrak{F}_s$, via the testing scheme $(\phi_n,\mathbf{1}_{(\mathcal{C}_{n}^{(i)})^c})$. The corresponding error probabilities are given by
\begin{align}
  P_{Y_i^n\phi_n(X_i^n)}(\mathcal{C}_{n}^{(i)}),\;Q_{Y_{j_i^{\star}}^n}\times Q_{\phi_n(X_{t_s^{\star}}^n)}[(\mathcal{C}_{n}^{(i)})^c].
\end{align}
Note that by the definition of $\mathcal{C}_n^{(i)}$ we have
\begin{align}
  Q_{Y_{j_i^{\star}}^n}\times Q_{\phi_n(X_{t_s^{\star}}^n)}[(\mathcal{C}_{n}^{(i)})^c]&\leq e^{-n(E-\gamma_n+\log\gamma_q/n)}P_{Y_i^n\phi_n(X_i^n)}[(\mathcal{C}_{n}^{(i)})^c]\nonumber\\
  &\leq e^{-n(E-\gamma_n+\log\gamma_q/n)}.
\end{align}
We want to transform the given testing scheme to obtain a new testing scheme $(\bar{\phi}_n^s,\bar{\psi}_n^{(i)})$, $i\in\mathfrak{F}_s$, for differentiating between $P_{Y_iX_i}^{\otimes n}$ and $P_{Y_i}^{\otimes n}\times P_{X_i}^{\otimes n}$. This can be done using similar arguments to those given in Appendix \ref{ap_common_seq} as follows. For the given positive $\gamma$ we define for the given $s\in[1:\vert \mathcal{P}_{\mathcal{X}}\vert ]$ a typical subset of $\mathcal{X}^n$
\begin{align}
  \mathcal{B}_{n,\gamma}^{s} = \{x^n\mid \vert  \iota_{P_{\mathcal{X},s}^{\otimes n}\Vert Q_{X_{t_s^{\star}}^n}}(x^n)/n - d_s^{\mathrm{x}}\vert  <\gamma\}.
\end{align}
Similarly for the given $i$ we define a typical subset of $\mathcal{Y}^n$
\begin{align}
  \mathcal{B}_{n,\gamma}^{(i)} = \{y^n\mid \vert  \iota_{P_{Y_i}^{\otimes n}\Vert Q_{Y_{j_i^{\star}}^n}}(y^n)/n - d_i^{\mathrm{y}}\vert  <\gamma\}.
\end{align}
Then the new compression mapping $\bar{\phi}_n^s$ is defined as
\begin{align}
  \bar{\phi}_n^s\colon\mathcal{X}^n&\to\mathcal{M}\cup\{e\}\nonumber\\
  \bar{\phi}_n^s(x^n)&\mapsto\begin{dcases}\phi_n(x^n),\;&\text{if}\; x^n\in\mathcal{B}_{n,\gamma}^{s},\nonumber\\
  e\;&\text{otherwise}\end{dcases}.
\end{align}
The decision mapping $\bar{\psi}_n^{(i)}$ is defined as
\begin{align}
  \bar{\psi}_n^{(i)}\colon \mathcal{Y}^n\times(\mathcal{M}\cup\{e\})&\to \{0,1\}\nonumber\\
  \bar{\psi}_n^{(i)}(y^n,\bar{u})&\mapsto\begin{dcases}\mathbf{1}_{(\mathcal{C}_{n}^{(i)})^c}(y^n,\bar{u}),\;&\text{if}\; y^n\in \mathcal{B}_{n,\gamma}^{(i)},\;\text{and}\; \bar{u}\neq e,\nonumber\\
  1&\text{otherwise}\end{dcases}.
\end{align}
Using this testing scheme $(\bar{\phi}_n^{s},\bar{\psi}_n^{(i)})$ we can bound the error probabilities in testing $P_{Y_iX_i}^{\otimes n}$ against $P_{Y_i}^{\otimes n}\times P_{X_i}^{\otimes n}$ as
\begin{align}
  P_{Y_i^{n}\bar{\phi}_n^s(X_i^n)}(1-\bar{\psi}_n^{(i)})&\leq P_{Y_i^n\phi_n(X_i^n)}(\mathcal{C}_{n}^{(i)}) + P_{Y_i}^{\otimes n}[(\mathcal{B}_{n,\gamma}^{(i)})^{c}] + P_{X_i}^{\otimes n}[(\mathcal{B}_{n,\gamma}^{s})^{c}]\nonumber\\
  P_{Y_i}^{\otimes n}\times P_{\bar{\phi}_n^s(X_i^n)}(\bar{\psi}_n^{(i)})&\leq Q_{Y_{j_i^{\star}}^n}\times Q_{\phi_n(X_{t_s^{\star}}^n)}[(\mathcal{C}_{n}^{(i)})^c] e^{n(d_s^{\mathrm{x}}+ d_i^{\mathrm{y}}+2\gamma)}\nonumber\\
  &\leq e^{-nE_i^{\prime}},\label{eq_56}
\end{align}
where $E_i^{\prime} = E-\gamma_n+\log\gamma_q/n - [d_{is}^{\mathrm{yx}}+2\gamma]$.
Then using \cite[Lemma 4.1.2]{hanspectrum} we obtain the following inequality 
\begin{align}
  P_{Y_i^{n}\bar{\phi}_n^s(X_i^n)}(1-\bar{\psi}_n^{(i)}) &+ e^{n(E_i^{\prime}-\gamma)}P_{Y_i}^{\otimes n}\times P_{\bar{\phi}_n^s(X_i^n)}(\bar{\psi}_{n}^{(i)})\nonumber\\
  &\geq P_{Y_i^{n}\bar{\phi}_n^s(X_i^n)}[(\bar{\mathcal{A}}_{n}^{(i)})^c],
\end{align}
holds where
\begin{align}
  \bar{\mathcal{A}}_{n}^{(i)} = \{(y^n,\bar{\phi}_n^s(x^n))\mid &P_{Y_i^{n}\bar{\phi}_n^s(X_i^n)}(y^n,\bar{\phi}_n^s(x^n))\nonumber\\
  &\geq e^{n(E_{i}^{\prime}-\gamma)}P_{Y_i}^{\otimes n}\times P_{\bar{\phi}_n^s(X_i^n)}(y^n,\bar{\phi}_n^s(x^n))\},
\end{align}
is our desired decision region using the likelihood ratio test in testing $P_{Y_i^{n}\bar{\phi}_n^s(X_i^n)}$ against $P_{Y_i}^{\otimes n}\times P_{\bar{\phi}_n^s(X_i^n)}$. Using the inequalities in \eqref{eq_56} we obtain
\begin{align}
P_{Y_i^{n}\bar{\phi}_n^s(X_i^n)}[(\bar{\mathcal{A}}_{n}^{(i)})^c]\leq & P_{Y_i^n\phi_n(X_i^n)}(\mathcal{C}_{n}^{(i)}) + P_{Y_i}^{\otimes n}[(\mathcal{B}_{n,\gamma}^{(i)})^{c}]\nonumber\\
& + P_{X_i}^{\otimes n}[(\mathcal{B}_{n,\gamma}^{s})^{c}]+ e^{-n\gamma}.\label{after_transform}
\end{align}
Under Assumption \ref{assump_1}, $P_{Y_i}^{\otimes n}[(\mathcal{B}_{n,\gamma}^{(i)})^{c}]\to 0$ and $P_{X_i}^{\otimes n}[(\mathcal{B}_{n,\gamma}^{s})^{c}]\to 0$ as $n\to\infty$ due to either the weak law of large numbers or Theorem 1 in \cite{barron1985strong}. 
Since both $\gamma_p>0$, and $\gamma>\gamma_n-\log \gamma_q/n$ as $n\to\infty$ hold, by combining \eqref{before_transform} and \eqref{after_transform} we have for all $i\in [1:m]$
\begin{align}
  \lim_{n\to\infty}\mathrm{Pr}\{\iota_{Y_i^n\bar{\phi}_n^s(X_i^n)}(Y_i^n;\bar{\phi}_n^s(X_i^n))<n(E-d_{is}^{\mathrm{yx}}-4\gamma)\} = 0,
\end{align}
where $(Y_i^n,X_i^n)\sim P_{Y_iX_i}^{\otimes n}$ holds. Hence, for all $i\in [1:m]$, $i\in\mathfrak{F}_s$, we have
\begin{align}
  E - d_{is}^{\mathrm{yx}}-4\gamma \leq \underline{I}(\mathbf{Y}_i;\bar{\phi}_s(\mathbf{X}_i)),
\end{align}
where $(\mathbf{Y}_i,\bar{\phi}_s(\mathbf{X}_i)) = \{(Y_i^n,\bar{\phi}_n^s(X_i^n))\}_{n=1}^{\infty}$ and $\underline{I}(\cdot;\cdot)$ is the spectral-inf mutual information rate, defined for a joint process $(\bar{\mathbf{U}},\bar{\mathbf{V}})= \{(U^n,V^n)_{n=1}^{\infty}\}$ as
\begin{align}
\underline{I}(\mathbf{U};\mathbf{V}) = \sup\big\{\beta\big\vert   \lim_{n\to \infty}\mathrm{Pr}\big[\iota_{P_{U^nV^n}}(U^n;V^n)<n\beta\big]=0\big\}.
\end{align}
Since the spectral-inf mutual information rate is less than or equal to the inf-mutual information rate by \cite[Theorem 3.5.2]{hanspectrum}
 $$\underline{I}(\mathbf{Y}_i;\bar{\phi}_s(\mathbf{X}_i))\leq \liminf_{n\to\infty}\frac{1}{n}I(Y_i^n;\bar{\phi}_n^s(X_i^n))$$ holds, $\forall i\in [1:m], i\in\mathfrak{F}_s$, we have 

\begin{align}
  E - d_{is}^{\mathrm{yx}}-4\gamma\leq \sup_{n_0}\inf_{n\geq n_0}\frac{1}{n}I(Y_i^n;\bar{\phi}_n^s(X_i^n)).
\end{align}
For each $i\in [1:m]$, let $n_i(\gamma)$ be such that 
\begin{align}
 \sup_{n_0}\inf_{n\geq n_0}\frac{1}{n}I(Y_i^n;\bar{\phi}_n^s(X_i^n))\leq \inf_{n\geq n_i(\gamma)}\frac{1}{n}I(Y_i^n;\bar{\phi}_n^s(X_i^n)) + \gamma.
\end{align}
Then for all $i\in [1:m]$, we have 
\begin{align}
 E - d_{is}^{\mathrm{yx}}-4\gamma\leq \inf_{n\geq n_i(\gamma)}\frac{1}{n}I(Y_i^n;\bar{\phi}_n^s(X_i^n)) + \gamma.
\end{align}

Let $T$ be a uniform random variable on $[1:n]$ and independent of everything else. For each $i\in [1:m]$ and $i\in\mathfrak{F}_s$, we define $U_{il} = (\bar{\phi}_n^s(X_i^n),X_i^{l-1})$ for all $l\in [1:n]$ and $U_i = (U_{iT},T)$.
Therefore for all $\forall i\in [1:m]$ and $i\in\mathfrak{F}_s$, as well as for all sufficiently large $n$, say $n\geq \max_{i}n_i(\gamma)$, we have
\begin{align}
  E -& d_{is}^{\mathrm{yx}}-5\gamma\leq \frac{1}{n}I(Y_i^n;\bar{\phi}_n^s(X_i^n))\nonumber\\
  & = \frac{1}{n}\sum_{l=1}^nI(Y_{il};\bar{\phi}_n^s(X_i^n)\vert Y_i^{l-1}) \stackrel{(a)}{=} \frac{1}{n}\sum_{l=1}^nI(Y_{il};\bar{\phi}_n^s(X_i^n),Y_i^{l-1})\nonumber\\
  &\leq \frac{1}{n}\sum_{l=1}^nI(Y_{il};\bar{\phi}_n^s(X_i^n),Y_i^{l-1},X_i^{l-1})\stackrel{(b)}{=} \frac{1}{n}\sum_{l=1}^nI(Y_{il};\bar{\phi}_n^s(X_i^n),X_i^{l-1})\nonumber\\
  & = I(Y_{iT};U_{iT},T) = I(Y_{iT};U_i),
\end{align}
where $(a)$ follows since $Y_i^n\sim P_{Y_i}^{\otimes n}$, and $(b)$ is valid since $Y_i^{l-1}-X_i^{l-1}-(Y_{il},\bar{\phi}_n^s(X_i^n))$ forms a Markov chain.
 Similarly we also have
\begin{align}
  R_c+\gamma&\geq \frac{1}{n}\log\vert  \bar{\phi}_n^s\vert  \geq \frac{1}{n}I(X^n_i;\bar{\phi}_n^s(X_i^n)) = \frac{1}{n}\sum_{l=1}^n I(X_{il}; \bar{\phi}_n^s(X_i^n), X_i^{l-1})\nonumber\\
  &= I(X_{iT};U_{iT},T) = I(X_{iT};U_i),\; \forall i\in [1:m].
\end{align}
Note that for $\tau,\eta \in \mathfrak{F}_s$, we have $P_{U_{\tau}\vert  X_{\tau T}} = P_{U_{\eta}\vert  X_{\eta T}}$. We define this common kernel for each $\mathfrak{F}_s$ as $P_{U_s\vert  \bar{X}_s}$. Therefore for all $s\in [1:\vert \mathcal{P}_{\mathcal{X}}\vert ]$ and all sufficiently large $n$ we have
\begin{align}
  E-5\gamma&\leq \min_{i\in\mathfrak{F}_s}[I(\bar{Y}_{i};U_s)+d_{is}^{\mathrm{yx}}],\nonumber\\
   R_c+\gamma&\geq I(\bar{X}_{s};U_s),
\end{align}
where for all $i\in\mathfrak{F}_s$ we have $(\bar{Y}_{i},\bar{X}_{s},U_s)\sim P_{Y_i\vert  X_i}\times P_{\mathcal{X},s}\times P_{U_s\vert  \bar{X}_s}$.
By standard cardinality bound arguments \cite{csiszar2011information} and taking $\gamma\to 0$ we have that for all $s\in [1:\vert \mathcal{P}_{\mathcal{X}}\vert ]$
\begin{align}
  E\leq \min_{i\in\mathfrak{F}_s}&[I(\bar{Y}_i;\bar{U}_s)+d_{is}^{\mathrm{yx}}],\; R_c\geq I(\bar{X}_s;\bar{U}_s),\nonumber\\
&  P_{\bar{Y}_i\bar{X}_s\bar{U}_s} = P_{\bar{Y}_i\bar{X}_s}\times P_{\bar{U}_s\vert\bar{X}_s}.
\end{align}
Hence $E\leq \theta_s(R_c)$ holds for all $s\in [1:\vert \mathcal{P}_{\mathcal{X}}\vert ]$, which leads to $E_{\mathrm{mix},0}^{\star}(R_c)\leq \min_s\theta_s(R_c)$.
\end{proof}

\section{$\epsilon$-Error Exponent in Mixture Setting}\label{sec_5}
In this section we characterize the maximum $\epsilon$-achievable error exponent in the testing against generalized independence setting in Section \ref{sec_4}. 
\subsection{Small $\epsilon$-optimality of $E_{\mathrm{mix},0}^{\star}(R_c)$}
The following partial result is an immediate consequence of the first part of Theorem \ref{thm_2}. It states that when $\epsilon$ is small enough then the maximum achievable error exponent $E_{\mathrm{mix},0}^{\star}(R_c)$ is also $\epsilon$-optimal.
\begin{corollary} \label{coroll_1}
  For a given $R_c$, if $\epsilon<\min_{i\in[1:m]}\nu_i\times \min\{\min_{s\in\mathcal{S}}1/\vert  \mathfrak{F}_s\vert  ,1\}$, where the inactive set $\mathcal{S}$ is defined as in the statement of Theorem \ref{thm_2}, then
  \begin{equation}
    E_{\mathrm{mix},\epsilon}^{\star}(R_c) = \min_s\theta_s(R_c).
  \end{equation}
\end{corollary}
\begin{proof}
  Assume that $(\phi_n,\psi_n)$ is a sequence of testing schemes such that the conditions in Definition \ref{def_1} are satisfied for the pair $(R_c,E)$
  \begin{align}
    \limsup_{n\to\infty}\frac{1}{n}\log\vert  \phi_n\vert  \leq R_c,\; &\limsup_{n\to\infty}\alpha_n\leq \epsilon,\nonumber\\
    \liminf_{n\to\infty}\frac{1}{n}&\log\frac{1}{\beta_n}\geq E.
  \end{align}
\noindent For an arbitrarily small $\gamma>0$ and for all $i\in [1:m]$, $i\in\mathfrak{F}_s$, we have $Q_{Y_{j_i^{\star}}^n}\times Q_{\phi_n(X_{t_s^{\star}}^n)}(\psi_n)\leq e^{-n(E-\gamma)}$ for all $n\geq n_0(\gamma)$. Furthermore we also have
  \begin{align}
    \epsilon\geq \limsup_{n\to\infty}\alpha_n\geq \nu_i\limsup_{n\to\infty}P_{Y_i^n\phi_n(X_i^n)}(1-\psi_n),\;\forall i\in [1:m],
  \end{align}
  which implies that $\forall i\in [1:m]$ we have $\limsup_{n\to\infty}P_{Y_i^n\phi_n(X_i^n)}(1-\psi_n)<\min\{\min_{s\in\mathcal{S}}1/\vert  \mathfrak{F}_s\vert  ,1\}$. By the \textit{proof} the first part of Theorem \ref{thm_2} we obtain then $E-\gamma<\min_s\theta_s(R_c)$. Since $\gamma$ is arbitrary the conclusion follows.
\end{proof}
\subsection{A sufficient condition for characterizing $E_{\mathrm{mix},\epsilon}^{\star}(R_c)$}
To obtain a full characterization of $E_{\mathrm{mix},\epsilon}^{\star}(R_c)$ we need to make an additional assumption. By \cite[Theorem 2.5]{witsenhausen1975conditional} the maximization in the expression of $\xi_{i}(R_c)$, $i\in\mathfrak{F}_s$, can be restriction to the set $\mathcal{W}_s(R_c) = \{P_{U\vert  \bar{X}_s}\mid I(\bar{X}_s;U) = R_c\}$. For each $i\in \mathfrak{F}_s$, we define $\mathcal{W}_{s}^{(i)}(R_c)$ to be the set of optimal solutions of $\xi_{i}(R_c)$ within $\mathcal{W}_s(R_c)$. We define a set $\mathcal{O}_s(R_c)$ as follows:
\begin{align}
&\text{If} \; \bigcap_{i\in\mathfrak{F}_s}\mathcal{W}_{s}^{(i)}(R_c)\neq \varnothing,\;\text{then we define}\; \mathcal{O}_s(R_c) = \bigcap_{i\in\mathfrak{F}_s}\mathcal{W}_{s}^{(i)}(R_c),\nonumber\\ &\text{otherwise we take}\; \mathcal{O}_s(R_c) = \bigcup_{i\in\mathfrak{F}_s}\mathcal{W}_s^{(i)}(R_c).
\end{align}
 For an arbitrarily given $P_{U\vert  X}\in\mathcal{O}_s(R_c)$, the set of $\{\xi_{i},\;i\in\mathfrak{F}_s\}$ can be arranged as
\begin{equation}
I(\bar{Y}_{i_1};U) + d_{i_1}^{\mathrm{y}}\leq I(\bar{Y}_{i_2};U) + d_{i_2}^{\mathrm{y}}\leq \dots\leq I(\bar{Y}_{i_{\vert  \mathfrak{F}_s\vert  }};U) + d_{\vert  \mathfrak{F}_s\vert  }^{\mathrm{y}}.\label{fixed_order}
\end{equation}
This corresponds to a permutation $\sigma_s(P_{U\vert  X})$ on $[1:\vert  \mathfrak{F}_s\vert  ]$. It can be seen that when $\bigcap_{i\in\mathfrak{F}_s}\mathcal{W}_{s}^{(i)}(R_c)$ is not empty then the permutation $\sigma_s(\cdot)$ does not depend on a specific $P_{U\vert  X}$ in $\mathcal{O}_s(R_c)$. To account for more general scenarios, we make the following \textit{separability} assumption.
\begin{assumption}\label{assump_2}
For given set of distributions $\{P_{X_iY_i}\}_{i=1}^m$ and set of sequences of distributions $\{(Q_{Y_j^n})\}_{j=1}^k$, class $\mathfrak{F}_s$ and $R_c$, any two kernels $P_{U\vert  X}^1$ and $P_{U\vert  X}^2$ in $\mathcal{O}_s(R_c)$ satisfy $\sigma_s(P_{U\vert  X}^1) = \sigma_s(P_{U\vert  X}^2)$, i.e., the order is invariant to the change of kernels inside $\mathcal{O}_s(R_c)$.
\end{assumption}
It can be seen that Assumption \ref{assump_2} is satisfied when $\vert  \mathfrak{F}_s\vert  =1$ for all $s$, i.e., when all $P_{Y_iX_i}$ have distinct marginal distributions. In the following we briefly discuss two non-trivial scenarios in which Assumption \ref{assump_2} can be fulfilled. 

\textbf{Example 1}: In this example we assume that $\{Q_{Y_j^n}\} = \{P_{Y_i}^{\otimes n}\}$ and $\{Q_{X_t^n}\} = \{P_{\mathcal{X},s}^{\otimes n}\}$ hold. This implies that for all $i\in [1:m]$ and for all $s\in [1:\vert \mathcal{P}_{\mathcal{X}}\vert ]$, we have $d_i^{\mathrm{y}} = 0$ and $d_s^{\mathrm{x}} = 0$. We then assume further that within each class $\mathfrak{F}_s$ the set of channels $P_{\bar{Y}_i\vert  \bar{X}_s}$ can be ordered according to the \textit{less noisy} relation\footnote{A channel $P_{Y\vert  X}$ is less noisy \cite{lessnoisy} than a channel $P_{Z\vert  X}$ if for every $P_{XU}$ we have $I(Y;U)\geq I(Z;U)$.}. Then for all $P_{U\vert  X}\in\mathcal{O}_s(R_c)$, we have
\begin{align}
I(\bar{Y}_{i_1};U)\leq I(\bar{Y}_{i_2};U)\leq \dots\leq I(\bar{Y}_{i_{\vert  \mathfrak{F}_s\vert  }};U),
\end{align}
for some fixed order $\{i_1,\dots,i_{\vert  \mathfrak{F}_s\vert  }\}$ which does not depend on whether $\bigcap_{i\in\mathfrak{F}_s}\mathcal{W}_{s}^{(i)}(R_c)$ is empty or not. Hence Assumption \ref{assump_2} is satisfied.

\textbf{Example 2}: We consider another example in which the set of distributions in the null hypothesis is $\{P_{\bar{Y}_i\vert  \bar{X}}\times P_{\bar{X}}\}_{i=1}^2$. We assume further that $P_{\bar{Y}_1\vert  \bar{X}}$ is an erasure channel with erasure probability $t$. We also assume that $\{Q_{Y_j^n}\} = \{P_{\bar{Y}_i}^{\otimes n}\}$ and $\{Q_{X_t^n}\} = \{P_{\bar{X}}^{\otimes n}\}$ hold in this example. Then $\xi_1(R_c) = (1-t)R_c$, which can be achieved by any kernel $P_{U_1\vert  \bar{X}}$ such that $I(\bar{X};U_1)=R_c$. This implies that $\mathcal{W}^{(1)}(R_c)\cap\mathcal{W}^{(2)}(R_c) = \mathcal{W}^{(2)}(R_c)$ holds, a non-empty set, and hence Assumption \ref{assump_2} is satisfied.

\noindent 
Assumption \ref{assump_2} implies that at a given $R_c$ we have
\begin{equation}
  \xi_{i_1}(R_c)\leq\xi_{i_2}(R_c)\leq\dots\leq \xi_{i_{\vert  \mathfrak{F}_s\vert  }}(R_c).\label{ordering_cond}
  \end{equation}
This can be seen as follows. If $\bigcap_{i\in\mathfrak{F}_s}\mathcal{W}_{s}^{(i)}(R_c)\neq \varnothing$, we can take any kernel inside $\mathcal{O}_s(R_c)$ to achieve $\xi_{i_l}(R_c)$ for all $l\in [1:\vert  \mathfrak{F}_s\vert  ]$. Hence \eqref{fixed_order} leads to \eqref{ordering_cond}. Otherwise, we take a kernel $P_{U\vert X}\in \mathcal{W}_{s}^{(i_1)}(R_c)$ to obtain $\xi_{i_1}(R_c)\leq I(\bar{Y}_{i_2};U) + d_{i_2}^{\mathrm{y}}\leq \xi_{i_2}(R_c)$ since $\mathcal{W}_{s}^{(i_1)}(R_c)\subseteq \mathcal{O}_s(R_c)$ holds and so on. Furthermore we observe that for all $l\in [1:\vert  \mathfrak{F}_s\vert  ]$, we have
\begin{align}
  \xi_{i_l}(R_c)&\geq \max_{P_{U\vert  \bar{X}_s}\colon I(\bar{X}_s;U)\leq R_c}\min_{i\in \{i_{\eta}\}_{\eta=l}^{\vert  \mathfrak{F}_s\vert  }}[I(\bar{Y}_{i};U)+d_{is}^{\mathrm{yx}}]\nonumber\\
  &\geq \max_{P_{U\vert  \bar{X}_s}\in\mathcal{O}_s(R_c)}\min_{i\in \{i_{\eta}\}_{\eta=l}^{\vert  \mathfrak{F}_s\vert  }}[I(\bar{Y}_{i};U)+d_{is}^{\mathrm{yx}}]\nonumber\\
                                                                                                                                                        &\stackrel{\eqref{fixed_order}}{=} \max_{P_{U\vert  \bar{X}_s}\in\mathcal{O}_s(R_c)} [I(\bar{Y}_{i_l};U)+d_{i_ls}^{\mathrm{yx}}]\nonumber\\
  &\stackrel{(*)}{\geq} \xi_{i_l}(R_c).
\end{align}
$(*)$ can be verified as follows. If $\bigcap_{i\in\mathfrak{F}_s}\mathcal{W}_{s}^{(i)}(R_c)\neq \varnothing$, we can take any kernel inside $\mathcal{O}_s(R_c)$ to achieve $\xi_{i_l}(R_c)$. Otherwise, $\mathcal{O}_s(R_c)\supseteq \mathcal{W}_{s}^{(i_l)}(R_c)$ holds, then $(*)$ follows. Therefore, we have the following relation
\begin{equation}
  \xi_{i_l}(R_c) = \max_{P_{U\vert  \bar{X}_s}\colon I(\bar{X}_s;U)\leq R_c}\min_{i\in \{i_{\eta}\}_{\eta=l}^{\vert  \mathfrak{F}_s\vert  }}[I(\bar{Y}_{i};U)+d_{is}^{\mathrm{yx}}].\label{arrange_relation}
\end{equation}
Assume that at a given $R_c$, Assumption \ref{assump_2} is valid. By relabeling elements in the set $\{P_{Y_iX_i}\}_{i=1}^m$ if necessary, we assume that  $(\xi_i(R_c))_{i=1}^m$ is an increasing sequence. An example of such ordering is given in Fig. \ref{fig_illustrate}.
\begin{figure}
\centering
\includegraphics{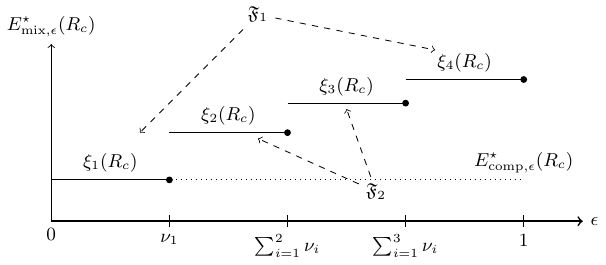}
\caption{Illustration of $E_{\mathrm{mix},\epsilon}^{\star}(R_c)$ under Assumption \ref{assump_2}. For comparison note that under Assumption \ref{assump_2} $E_{\mathrm{comp},\epsilon}^{\star}(R_c)$ does not depend on $\epsilon$. At $l=2$ the left-over sets are given by $\mathfrak{F}_1(2) = \{4\}$, $\mathfrak{F}_2(2) = \{2,3\}$.}\label{fig_illustrate}
\end{figure}

For each $l\in [1:m]$, for notation simplicity we define for each $s$ a left-over subset of $\mathfrak{F}_s$ as $\mathfrak{F}_s(l) = \mathfrak{F}_s\backslash [1:l-1]$. Since the ordering is unique, \eqref{ordering_cond} implies that $i_{1}\leq i_2\leq\dots\leq i_{\vert  \mathfrak{F}_s\vert  }$. Therefore, when $\mathfrak{F}_s(l)\neq \varnothing$, we have 
\begin{equation}
  \xi_{\min \mathfrak{F}_s(l)}(R_c) = \max_{P_{U\vert  \bar{X}_s}\colon I(\bar{X}_s;U)\leq R_c}\min_{i\in \mathfrak{F}_s(l)}[I(\bar{Y}_{i};U)+d_{is}^{\mathrm{yx}}].\label{before_propos_1}
\end{equation}
The above analysis leads to the following result.
\begin{proposition}\label{ln_conseq}
 Assume that at a given $R_c$ Assumption \ref{assump_2} is fulfilled and $(\xi_i(R_c))_{i=1}^m$ is an increasing sequence. Then for each $l\in [1:m]$ the following holds
\begin{align}
  \min_s\max_{P_{U\vert  \bar{X}_s}\colon I(\bar{X}_s;U)\leq R_c}\min_{i\in\mathfrak{F}_s(l)}[I(\bar{Y}_i;U) + d_{is}^{\mathrm{yx}}] = \xi_{l}(R_c). \label{ln_con_eq}
\end{align}
\end{proposition}
The left-hand side in \eqref{ln_con_eq} is the maximum achievable error exponent in testing $\{P_{Y_iX_i}^{\otimes n}\}_{i\in [l:m]}$ against $\{Q_{Y_j^n}\times Q_{X_t^n}\}$. This result will be used to establish the optimal $\epsilon$-error exponent in this section.
\begin{proof}
  Due to the above analysis the left-hand side equals to, cf. \eqref{before_propos_1},
  \begin{align}
    &\min_{s\colon \mathfrak{F}_s(l)\neq \varnothing}\xi_{\min \mathfrak{F}_s(l)}(R_c) = \min_{i\in \bigcup_s \mathfrak{F}_s(l)}\xi_i(R_c)\nonumber\\
    &=\min_{i\in [l:m]}\xi_i(R_c) = \xi_l(R_c).
    \end{align}
\end{proof}
\subsection{Characterization of $E_{\mathrm{mix},\epsilon}^{\star}(R_c)$ under Assumption \ref{assump_2}}
A complete characterization of $E_{\mathrm{mix},\epsilon}^{\star}(R_c)$ under Assumption \ref{assump_2} is provided in the following.
\begin{theorem}\label{thm_4}
  Assume that a given $R_c$, Assumption \ref{assump_2} is true and $(\xi_i(R_c))_{i=1}^m$ is an increasing sequence, then we have
\begin{align}
  E_{\mathrm{mix},\epsilon}^{\star}(R_c) = \sum_{i=1}^m \xi_i(R_c)\mathbf{1}_{[\sum_{j=1}^{i-1}\nu_j,\sum_{j=1}^{i}\nu_j)}(\epsilon).
\end{align}
\end{theorem}
The behavior of $E_{\mathrm{mix},\epsilon}^{\star}(R_c)$ is depicted in Fig. \ref{fig_illustrate}. The proof of Theorem \ref{thm_4} uses the result from the second part of Theorem \ref{thm_2}. It demonstrates an application of showing \textit{exponentially strong} converse. We briefly describe the proof idea in the following.

Fix an $\epsilon$ such that $\epsilon \in [\sum_{j=1}^{i-1}\nu_j, \sum_{j=1}^{i}\nu_j)$ holds. We build an $\epsilon$-achievable sequence of testing schemes based on a compound hypothesis testing problem in which we need to specify two sets of distributions in the null and alternative hypotheses. The set of distributions in the alternative hypothesis is
\begin{align}
 H_1^{\prime}:\{Q_{Y_j^n}\times Q_{X_t^n}\}_{j\in [1:k],t\in [1:r]}.
 \end{align}
  Since we are interested in showing $\epsilon$-error is achievable, we select the set of distributions in the null hypothesis as
\begin{align}
H_0^{\prime}:\{P_{Y_lX_l}^{\otimes n}\}_{l\in [i:m]}.
\end{align}
 We omit the other distributions $P_{Y_lX_l}^{\otimes n}$ for $l\in [1:i-1]$ in the null hypothesis of the above problem because in the mixture setting they contribute only to a total error probability up to $\sum_{j=1}^{i-1}\nu_j$ which is what we desire. Given a sequence of testing schemes such that $E$ is achievable $(\phi_n,\psi_n)$ in testing $H_0^{\prime}$ against $H_1^{\prime}$, we obtain as by-products the collection of intersected decision regions $\mathcal{I}_n^{(l)}(E)$ with vanishing complement probability, cf. Theorem \ref{thm_1}.\\
 We use $\phi_n$ to compress $x^n$. Next we select a decision region $\mathcal{A}_n$ for testing $P_{Y^nX^n}$ against $Q_{Y^nX^n}$  such that for all $l\in [i:m]$, $\mathcal{A}_n^c\subset (\mathcal{I}_n^{(l)}(E))^c$ holds. This ensures that $P_{Y_i^n\phi_n(X_i^n)}(\mathcal{A}_n^c)\to 0$ for all $l\in [i:m]$. Hence, asymptotically we have $P_{Y^n\phi_n(X^n)}[\mathcal{A}_n^c]\leq \sum_{j=1}^{i-1}\nu_jP_{Y_j^n\phi_n(X_j^n)}[\mathcal{A}_n^c]\leq \epsilon$.\\
As in the converse proof of Theorem \ref{thm_3}, assume that we have a likelihood based decision region $\mathcal{A}_n$ for testing $P_{Y^nX^n}$ against $Q_{Y^nX^n}$. Then we use change of measure steps to obtain for each $i\in [1:m]$ a likelihood based rejection region of testing $P_{Y_i^n\phi_n(X_i^n)}$ against $Q_{Y_{j_i^{\star}}^n}\times Q_{\phi_n(X_{t_s^{\star}}^n)}$, called $\mathcal{C}_{n}^{(i)}$. Then for all sufficiently large $n$ we have
\begin{align}
\epsilon>\sum_{l=1}^m \nu_l P_{Y_l^n\phi_n(X_l^n)}(\mathcal{C}_{n}^{(l)}).
\end{align}
We show that for all $l\in [1:i]$ we have $P_{Y_l^n\phi_n(X_l^n)}(\mathcal{C}_{n}^{(l)})$ converges to 1 when $E>\xi_i(R_c) + 3\gamma$ holds. Taking the limit superior on both sides on the above inequality, we obtain that $\epsilon>\sum_{j=1}^i\nu_j$, a contradiction. Hence we must have $E\leq\xi_i(R_c)$. \\
When no compression is involved, $\phi_n$ is the identity mapping, the convergence of $P_{Y_l^nX_l^n}(\mathcal{C}_{n}^{(l)})$ to 1 can be seen from the weak law of large numbers since $\mathcal{C}_{n}^{(l)}$ is characterized by the likelihood ratio. In our case this is not possible since the likelihood ratio can not be factorized as the sum of identical components. Using the strong converse arguments which guarantee the convergence in the limit superior sense, it can be seen that if $\epsilon\in [\max_{j\in [1:i-1]}{\nu_j}, \max_{j\in [1:i]}{\nu_j})$ holds, then we must have $E\leq\xi_{i}(R_c)$. This approach does not match the achievability result and does not yield conclusive information when $\epsilon>\max_{i\in[1:m]}\nu_i$. The required convergence of $P_{Y_l^n\phi_n(X_l^n)}(\mathcal{C}_{n}^{(l)})$ is guaranteed by the second part of Theorem \ref{thm_2} since $\xi_i(R_c)>\xi_{l}(R_c)$ holds for all $l\in [1:i]$ due to Assumption \ref{assump_2}. 

\begin{proof}
\textit{Achievability}: Given $i\in [1:m]$, consider a reduced compound problem of designing a testing scheme $(\phi_n,\psi_n)$ to differentiate between $\{P_{Y_lX_l}^{\otimes n}\}_{l\in [i:m]}$ and $\{Q_{Y_j^n}\times Q_{X_t^n}\}_{j\in [1:k],t\in [1:r]}$. By Theorem \ref{thm_1} and Proposition \ref{ln_conseq}, for each $\gamma>0$, there exists a sequence of testing schemes for this reduced compound setting such that $E+\gamma/2 = \xi_i(R_c)-\gamma/2$ is achievable, i.e.,
\begin{align}
  \lim_{n\to\infty}P_{Y_l^n\phi_n(X_l^n)}[(\mathcal{I}_{n}^{(l)}(E))^c] = 0,\;\forall l\in [i:m]. \label{eq_79}
\end{align}
In contrast to the achievability proof of Theorem \ref{thm_3}, although we can use the same sequence of compression mappings $(\phi_n)$ to compress $x^n$, we need to define a new sequence of decision mappings $\bar{\psi}_n$. For each $n$ and given the compression mapping $\phi_n$, define the following measure
\begin{align}
  \bar{P}_n = \sum_{l=1}^mP_{Y_l^n\phi_n(X_l^n)}.
\end{align}
Define an acceptance region for our testing $P_{Y^nX^n}$ against $Q_{Y^nX^n}$ problem as follows
\begin{align}
  \mathcal{A}_n = \big\{(y^n,\phi_n(x^n))\mid &\bar{P}_n(y^n,\phi_n(x^n))\nonumber\\
  &\geq \max_{j\in[1:k],t\in [1:r]}Q_{Y_j^n}\times Q_{\phi_n(X_t^n)}(y^n,\phi_n(x^n))e^{nE}\big\}.\label{accpt_region}
\end{align}
Then the probability of miss detection is given by
\begin{align}
  \beta_n = Q_{Y^n\phi_n(X^n)}(\mathcal{A}_n)&\leq kr\max_{j\in[1:k],t\in[1:r]}Q_{Y_j^n}\times Q_{\phi_n(X_t^n)}(\mathcal{A}_n)\nonumber\\
  &\leq kre^{-nE}\bar{P}_n(\mathcal{A}_n)\leq krme^{-nE}.
\end{align}
\noindent For the given $P_{Y^nX^n}$, the false alarm probability is given by
\begin{align}
  P_{Y^n\phi_n(X)^n}(\mathcal{A}_n^c) = \sum_{l=1}^m\nu_l P_{Y_l^n\phi_n(X_l^n)}(\mathcal{A}_n^c).
\end{align}
For each $l\in [1:m]$, if $(y^n,\phi_n(x^n))\in\mathcal{A}_n^c$ then we have
\begin{align}
P_{Y_l^n\phi_n(X_l^n)}(y^n,\phi_n(x^n))&\leq \bar{P}_n(y^n,\phi_n(x^n))\nonumber\\
&< e^{nE}\max_{j\in [1:k],t\in[1:r]}Q_{Y_j^n}\times Q_{\phi_n(X_t^n)}(y^n,\phi_n(x^n))
\end{align}
This means that $(y^n,\phi_n(x^n))\in (\mathcal{I}_{n}^{(l)}(E))^{c}$, cf. \eqref{ie_expl} for the definition. 
In summary we have
\begin{align}
  P_{Y^n\phi_n(X)^n}(\mathcal{A}_n^c)&=\sum_l\nu_lP_{Y_l^n\phi_n(X_l^n)}[\mathcal{A}_n^c]\nonumber\\
                                       &\leq \sum_l\nu_lP_{Y_l^n\phi_n(X_l^n)}[(\mathcal{I}_{n}^{(l)}(E))^c]\nonumber\\
                                     &\leq\sum_{l\in[1:i-1]}\nu_l + \sum_{l\in[i:m]}\nu_lP_{Y_l^n\phi_n(X_l^n)}[(\mathcal{I}_{n}^{(l)}(E))^c].
\end{align}
Therefore we have
\begin{equation}
  \limsup_{n\to\infty}\alpha_n \stackrel{\eqref{eq_79}}{\leq} \sum_{l\in[1:i-1]}\nu_l.
\end{equation}
This implies that $E_{\mathrm{mix},\epsilon}^{\star}(R_c)\geq \xi_i(R_c)$ if $\sum_{l=1}^{i-1}\nu_l\leq \epsilon <\sum_{l=1}^i\nu_l$.

\noindent\textit{Converse}: Given an $\epsilon \in [\sum_{j=1}^{i-1}\nu_j,\sum_{j=1}^i\nu_j)$ and an arbitrary $\gamma>0$, assume that $(E+\gamma)$ is $\epsilon$-achievable via a sequence of testing schemes $(\phi_n,\psi_n)$. In a similar fashion as in the converse proof of Theorem \ref{thm_3} we have,
\begin{align}
& \alpha_n + e^{nE}\beta_n +e^{-n\gamma_n}\nonumber\\
&\geq \sum_{i=1}^m \nu_i P_{Y_i^n\phi_n(X_i^n)}(\mathcal{C}_{n}^{(i)}),
\end{align}
where the rejection regions $\mathcal{C}_{n}^{(i)}$ are defined as in \eqref{c_i_n_def}.
Assume that $E = \xi_i(R_c) + 3\gamma$ holds. This implies that for all $l\in [1:i]$ we have $E>\xi_l(R_c)$. For each $l\in [1:i]$, we apply the second part of Theorem \ref{thm_2} to the problem of testing $P_{Y_lX_l}^{\otimes n}$ against $Q_{Y_{j_l^{\star}}^{n}}\times Q_{X_{t_s^{\star}}^{n}}$ where $l\in \mathfrak{F}_s$ via the sequence of testing schemes $(\phi_n,\mathbf{1}_{(\mathcal{C}_{n}^{(l)})^c})$ to obtain that
\begin{equation}
  \lim_{n\to\infty}P_{Y_l^n\phi_n(X_l^n)}(\mathcal{C}_{n}^{(l)}) = 1,\;\forall l\in[1:i].
\end{equation}
Since $e^{nE}\beta_n\to 0$ as $n\to\infty$, this implies that
\begin{align}
  \epsilon\geq \limsup_{n\to\infty}\alpha_n\geq \sum_{l=1}^i\nu_l,
\end{align}
a contradiction. Therefore we must have $E\leq \xi_i(R_c)$ and hence $E^{\star}_{\mathrm{mix},\epsilon}(R_c)\leq \xi_i(R_c)$.
\end{proof}
\section{A refined relation to the WAK problem}\label{sec_6}
In this section we use the techniques and results from the previous sections to establish new results for the WAK problem in which the joint distribution is a mixture of iid components.\\
We first recall the definition of the WAK problem. Assume that we have a joint source $(X^n,Y^n)\sim P_{X^nY^n}$ which takes values on an alphabet $\mathcal{X}_n\times\mathcal{Y}_n$ where $\mathcal{Y}_n$ is finite or countable infinite. A code for the WAK problem is a triple of mappings 
\begin{align}
\phi_{1n}\colon\mathcal{X}_n\to\mathcal{M}_1&,\; \phi_{2n}\colon\mathcal{Y}_n\to\mathcal{M}_2.\nonumber\\
\bar{\psi}_n\colon\mathcal{M}_1&\times\mathcal{M}_2\to\mathcal{Y}_n.
\end{align}
In the WAK setting we aim to control the error probability $\mathrm{Pr}\{\hat{Y}^n\neq Y^n\}$ where $\hat{Y}^n = \bar{\psi}_n(\phi_{1n}(X^n),\phi_{2n}(Y^n))$. The achievable region has been characterized using the information-spectrum formula in \cite{miyake1995coding}. We are interested in single-letter formulas for the ($\epsilon$-) achievable regions. To obtain these we relate the WAK problem to the testing against independence problem with general distributions.

The hypotheses are given by
\begin{align}
  H_0&\colon (y^n,x^n)\sim P_{Y^nX^n},\nonumber\\
  H_1&\colon (y^n,x^n)\sim Q_{Y^n}\times P_{X^n},
\end{align}
where $Q_{Y^n}$ is a distribution on $\mathcal{Y}_n$.
Similarly we use \eqref{testing_sch_a} and \eqref{testing_sch_b} for the definition of a testing scheme in this case. Additionally, Definition \ref{def_1} is taken as the definition of $\epsilon$-achievability.

When $\mathcal{Y}$ is a finite or countable finite alphabet, and the process $(Y_i)_{i=-\infty}^{\infty}$ in the WAK problem is a stationary and ergodic process which has a finite entropy rate, as well as $\mathcal{Y}_n = \mathcal{Y}^n$ and $Q_{Y^n}= P_{Y^n}$, a generalized relation between the WAK problem and the hypothesis testing against independence problem has been established in our recent work \cite{vu2021hypothesis}. The relation allows us to transfer results from the testing against independence problem to the WAK problem and vice versa. However it is not strong enough for our current interest. In the following we study a refined relation between the two problems.

In this section we similarly assume that $\mathcal{Y}_n = \mathcal{Y}^n$ where specifically $\mathcal{Y}$ is a finite alphabet. We further assume that $Q_{Y^n} = \chi^{\otimes n}$ where $\chi$ is the uniform distribution on $\mathcal{Y}$. For simplicity we call this setting U(niform)-HT.
Define a set
\begin{align}
\mathcal{A}_n(\phi_n,t) = \{(y^n,\phi_n(x^n))\mid P_{Y^n\vert  \phi_{n}(X^n)}(y^n\vert  \phi_{n}(x^n))\leq t\}.
\end{align}
For a given WAK-code $(\phi_{1n},\phi_{2n},\psi_n)$, and an arbitrary number $\eta$, we have
\begin{align}
&  \mathrm{Pr}\{\hat{Y}^n= Y^n\}\leq  P_{Y^n\phi_{1n}(X^n)}[(\mathcal{A}_n(\phi_{1n},e^{-\eta}/\vert  \mathcal{M}_2\vert  ))^c]\nonumber\\
                               &\hspace{1cm}+ \sum_{u_1}P_{\phi_{1n}(X^n)}(u_1)\sum_{\substack{y^n\colon P_{Y^n\vert  \phi_{1n}(X^n)}(y^n\vert  u_1)\leq e^{-\eta}/\vert  \mathcal{M}_2\vert  \\
  y^n=\bar{\psi}_n(u_1,\phi_{2n}(y^n))}}P_{Y^n\vert  \phi_{1n}(X^n)}(y^n\vert  u_1)\nonumber\\
  &\stackrel{(*)}{\leq}  P_{Y^n\phi_{1n}(X^n)}[(\mathcal{A}_n(\phi_{1n},e^{-\eta}/\vert  \mathcal{M}_2\vert  ))^c] + e^{-\eta},
\end{align}
where $(*)$ holds as for a given $u_1\in\mathcal{M}_1$ the number of $y^n$ satisfying $y^n=\bar{\psi}_n(u_1,\phi_{2n}(y^n))$ is upper bounded by $\vert  \mathcal{M}_2\vert  $.
Therefore, we obtain
\begin{align}
\mathrm{Pr}\{\hat{Y}^n\neq Y^n\} + e^{-\eta} \geq P_{Y^n\phi_{1n}(X^n)}(\mathcal{A}_n(\phi_{1n},e^{-\eta}/\vert  \mathcal{M}_2\vert  )).\label{wak_basic}
\end{align}
Additionally, given a testing scheme $(\phi_n,\psi_n)$ for the U-HT problem and an arbitrarily positive number $\gamma$ we also have by \cite[Lemma 4.1.2]{hanspectrum}
\begin{align}
  \alpha_n + \gamma\beta_n&\geq \mathrm{Pr}\big\{\frac{P_{Y^n\phi_n(X^n)}}{\chi^{\otimes n}\times P_{\phi_n}(X^n)}(Y^n,\phi_n(X^n))\leq \gamma\big\}\nonumber\\
  &= P_{Y^n\phi_n(X^n)}(\mathcal{A}_n(\phi_n,\gamma/\vert  \mathcal{Y}\vert  ^n)).\label{testing_abg}
\end{align}
Similar as in \cite[Theorem 2]{vu2021hypothesis} we have the following result.
\begin{theorem}\label{thm_5}
  Given positive numbers $\eta$ and $\gamma$.
  \begin{itemize}
  \item From a WAK-code $(\phi_{1n},\phi_{2n},\bar{\psi}_n)$, we can construct a testing scheme for a U-HT problem $(\phi_{1n},\psi_n)$ such that
    \begin{align}
      \alpha_n&\leq \mathrm{Pr}\{\hat{Y}^n\neq Y^n\} + e^{-\eta},\; \beta_n\leq \frac{e^{\eta}\vert  \mathcal{M}_2\vert  }{\vert  \mathcal{Y}\vert  ^{n}}.
    \end{align}
  \item For a given U-HT testing scheme $(\phi_n,\psi_n)$, there exists a WAK-code $(\phi_{1n},\phi_{2n},\bar{\psi}_n)$ such that
    \begin{align}
      \mathrm{Pr}\{\hat{Y}^n\neq Y^n\}\leq \alpha_n +\gamma\beta_n + \vert  \mathcal{Y}\vert  ^n/(\gamma\vert  \mathcal{M}_2\vert  ).
    \end{align}  
\end{itemize}
\end{theorem}
\begin{proof}
U-HT $\Leftarrow$ WAK: given a WAK-code $(\phi_{1n},\phi_{2n},\bar{\psi}_n)$ we design a U-HT testing scheme as follows. We use $\phi_{1n}$ to compress $x^n$ in the U-HT setting. A decision region for the U-HT setting is given by $(\mathcal{A}_n(\phi_{1n},e^{-\eta}/\vert  \mathcal{M}_2\vert  ))^{c}$. Then the false alarm probability is upper bounded as
\begin{align}
\alpha_n &= P_{Y^n\phi_{1n}(X^n)}(\mathcal{A}_n(\phi_{1n},e^{-\eta}/\vert  \mathcal{M}_2\vert  ))\nonumber\\
&\stackrel{\eqref{wak_basic}}{\leq} \mathrm{Pr}\{\hat{Y}^n\neq Y^n\} + e^{-\eta}.
\end{align}
The miss detection probability is upper bounded as
\begin{align}
\beta_n &= \chi^{\otimes n}\times P_{\phi_{1n}(X^n)}[(\mathcal{A}_n(\phi_{1n},e^{-\eta}/\vert  \mathcal{M}_2\vert  ))^c]\nonumber\\
&\stackrel{(**)}{\leq} \frac{e^{\eta}\vert  \mathcal{M}_2\vert  }{\vert  \mathcal{Y}\vert  ^{n}}P_{Y^n\phi_{1n}(X^n)}[(\mathcal{A}_n(\phi_{1n},e^{-\eta}/\vert  \mathcal{M}_2\vert  ))^c]\leq \frac{e^{\eta}\vert  \mathcal{M}_2\vert  }{\vert  \mathcal{Y}\vert  ^{n}},
\end{align}
where $(**)$ follows since for $(y^n,\phi_{1n}(x^n))\in (\mathcal{A}_n(\phi_{1n},e^{-\eta}/\vert  \mathcal{M}_2\vert  ))^c$ we have $e^{\eta}\vert  \mathcal{M}_2\vert  P_{Y^n\vert  \phi_{1n}(X^n)}(y^n\vert  \phi_{1n}(x^n))\geq 1$.\\
U-HT $\Rightarrow$ WAK: given a U-HT scheme $(\phi_n,\psi_n)$ we show the existence a WAK-code as follows. We use $\phi_n$ to compress $x^n$ in the WAK problem. We randomly assign $y^n$ to an index $m_2$ in an alphabet $\mathcal{M}_2$. For each $m_2$ we denote the corresponding (random) set of such $y^n$ by $\mathcal{B}(m_2)$. We declare that $\hat{y}^n$ is the original source sequence if it is the unique sequence satisfying $\hat{y}^n\in\mathcal{B}(m_2)$ and $(\hat{y}^n,\phi_n(x^n))\in (\mathcal{A}_n(\phi_n,\gamma/\vert  \mathcal{Y}\vert  ^n))^c$. For each $u\in \mathcal{M}$, the cardinality of the set of $y^n$ satisfying $(y^n,u)\in (\mathcal{A}_n(\phi_n,\gamma/\vert  \mathcal{Y}\vert  ^n))^c$ is upper bounded by $\vert  \mathcal{Y}\vert  ^n/\gamma$. There are two sources of errors:
\begin{itemize}
\item either $(y^n,\phi_n(x^n))\in \mathcal{A}_n(\phi_n,\gamma/\vert  \mathcal{Y}\vert  ^n)$ holds, 
\item or there exists another $\tilde{y}^n$ for which $\tilde{y}^n\in \mathcal{B}(m_2)$ and $(\tilde{y}^n,\phi_n(x^n))\in (\mathcal{A}_n(\phi_n,\gamma/\vert  \mathcal{Y}\vert  ^n))^c$ hold.
\end{itemize}
The probability of the first event is upper bounded as
  \begin{align}
    \mathrm{Pr}\{(Y^n,\phi_n(X^n))\in \mathcal{A}_n(\phi_n,\gamma/\vert  \mathcal{Y}\vert  ^n)\}\stackrel{\eqref{testing_abg}}{\leq} \alpha_n + \gamma\beta_n,
  \end{align}
The probability of the second event is upper bounded by $\vert  \mathcal{Y}\vert  ^n/(\gamma\vert  \mathcal{M}_2\vert  )$ because each sequence $\tilde{y}^n$ is assigned to a bin with probability $1/\vert  \mathcal{M}_2\vert  $ and the number of such sequences satisfying the second event is upper bounded by $\vert  \mathcal{Y}\vert  ^n/\gamma$. Hence, it can be seen that
\begin{align}
\mathrm{Pr}\{\hat{Y}^n\neq Y^n\}\leq \alpha_n +\gamma\beta_n + \vert  \mathcal{Y}\vert  ^n/(\gamma\vert  \mathcal{M}_2\vert  ).
\end{align}
\end{proof}
\noindent Fix an $\epsilon\in[0,1)$. Let $\mathcal{R}_{\mathrm{WAK},\epsilon}$ be the closure of all $(R_c,R_2)$ such that there exists a sequence WAK-codes $(\phi_{1n},\phi_{2n},\bar{\psi}_n)$ which satisfies
  \begin{align}
    \limsup_{n\to\infty}\frac{1}{n}\log \vert  \phi_{1n}\vert  \leq R_c,&\; \limsup_{n\to\infty}\frac{1}{n}\log\vert  \phi_{2n}\vert  \leq R_2,\nonumber\\
    \limsup_{n\to\infty}\mathrm{Pr}\{\hat{Y}^n\neq Y^n\}&\leq \epsilon.\label{ep_wak_reg}
  \end{align}
  Define $R_{2,\epsilon}^{\star}(R_c) = \inf\{R_2\mid (R_c,R_2)\in\mathcal{R}_{\mathrm{WAK},\epsilon}\}$. We observe that $R_{2,\epsilon}^{\star}(R_c)\leq \log\vert  \mathcal{Y}\vert  $ for all $\epsilon\in [0,1)$ and $R_c$. By setting $\gamma = \vert  \mathcal{Y}\vert  ^n$ in \eqref{testing_abg} we obtain that $E_{\epsilon}^{\star}(R_c)\leq \log \vert  \mathcal{Y}\vert  $. The following result summarizes the relation between the minimum encoding rate $R_{2,\epsilon}^{\star}(R_c)$ in the WAK problem and the maximum $\epsilon$-achievable error exponent in the U-HT problem. As in \cite[Theorem 3]{vu2021hypothesis} we have the following result.
  \begin{corollary}\label{coroll_1}
For any given $R_c>0$ and $\epsilon\in [0,1)$, we have    
\begin{equation}
R_{2,\epsilon}^{\star}(R_c) + E_{\epsilon}^{\star}(R_c) = \log\vert  \mathcal{Y}\vert  .
\end{equation}
\end{corollary}
\begin{proof}
  We consider the extreme case where $R_{2,\epsilon}^{\star}(R_c) = \log\vert  \mathcal{Y}\vert  $ holds. Assume that $E>0$ is an achievable error exponent in the U-HT problem with the corresponding sequence of testing schemes $(\phi_n,\psi_n)$. We take $\gamma = e^{n(E-\eta)}$ where $0<\eta<\frac{2}{3}E$. By plugging it into the second part of Theorem \ref{thm_5} and choosing $\vert  \mathcal{M}_2\vert   = (\vert  \mathcal{Y}\vert  ^n/\gamma) e^{n\eta/2}$, we obtain a sequence WAK-codes such that $\limsup_{n\to\infty}\mathrm{Pr}\{\hat{Y}^n\neq Y^n\}\leq \epsilon$. The corresponding compression rate is $\log\vert  \mathcal{Y}\vert  -E+\frac{3}{2}\eta<R_{2,\epsilon}^{\star}(R_c)$, a contradiction. Therefore in this case we must have $E_{\epsilon}^{\star}(R_c)=0$. The other cases can be worked out similarly as in the proof of \cite[Theorem 3]{vu2021hypothesis}.
\end{proof}
Assume now that the source distribution in the WAK problem (and the distribution in the U-HT problem) is given by $P_{X^nY^n} = \sum_{i=1}^m\nu_iP_{X_iY_i}^{\otimes n}$. Further we assume that $Q_{Y_j^n} = \chi^{\otimes n}$ for all $j\in [1:k]$, and $\{Q_{X_t^n}\} = \{P_{\mathcal{X},s}^{\otimes n}\}$ as well as $Q_{Y^nX^n} = \chi^{\otimes n}\times P_{X^n}$ hold. Therefore we can use results from Section III as follows. The quantities $\theta_s(R_c)$, $s\in [1:\vert \mathcal{P}_{\mathcal{X}}\vert ]$, in this case are given by
\begin{align}
\theta_s(R_c) &= \max_{P_{U\vert  X}\colon I(X_s;U)\leq R_c}\min_{i\in\mathfrak{F}_s}[I(Y_i;U) - H(Y_i) + \log\vert  \mathcal{Y}\vert  ]\nonumber\\ 
& = \log\vert  \mathcal{Y}\vert  -\min_{P_{U\vert  X}\colon I(X_s;U)\leq R_c} \max_{i\in\mathfrak{F}_s}H(Y_i\vert  U).
\end{align}
Therefore using Corollary \ref{coroll_1} we obtain
\begin{equation}
R_{2,0}^{\star}(R_c) = \max_s\min_{P_{U\vert  X}\colon I(X_s;U)\leq R_c} \max_{i\in\mathfrak{F}_s}H(Y_i\vert  U).
\end{equation}
Similarly if Assumption \ref{assump_2} hold and $(\xi_i(R_c))_{i=1}^m$ is an increasing sequence, the minimum $\epsilon$-achievable compression rate at $R_c$ is given by
\begin{align}
R_{2,\epsilon}^{\star}(R_c) = \log\vert  \mathcal{Y}\vert   - \xi_i(R_c),\;\text{when}\;\sum_{j=1}^{i-1}\nu_i\leq \epsilon <\sum_{j=1}^i\nu_i,
\end{align}
where in this case for all $i\in [1:m]$ we have
\begin{equation}
\xi_i(R_c) = \log\vert  \mathcal{Y}\vert  -\min_{P_{U\vert  X}\colon I(X_s;U)\leq R_c} H(Y_i\vert  U).
\end{equation}
\appendices

\section{Hypothesis testing with two-sided compression}\label{ap_common_seq}
In this section we study code transformations between hypothesis testing problems with compression at both terminals. Although the setting considered in the following is not directly related to our main problem, the arguments presented herein however are useful in simplifying proofs of main results in this work.

Assume that $x^n$ is available at Terminal 1. Further $y^n$ is available at Terminal 2. At a given number of samples $n$ a generic testing scheme involves a triple of mappings $(\phi_{1n},\phi_{2n},\psi_n)$ where
\begin{align}
  \phi_{1n}\colon\mathcal{X}^n&\to\mathcal{M}_1,\;  \phi_{2n}\colon\mathcal{Y}^n\to\mathcal{M}_2,\nonumber\\
  \psi_n\colon &\mathcal{M}_1\times \mathcal{M}_2\to\{0,1\}.
\end{align}
Similarly as in \eqref{error_probs_def} the generic error probabilities $\alpha_n$ and $\beta_n$ are defined as
\begin{align}
  \alpha_n = P_{\phi_{1n}(X^n)\phi_{2n}(Y^n)}(1-\psi_n),\; \beta_n = Q_{\phi_{1n}(X^n)\phi_{2n}(Y^n)}(\psi_n).
\end{align}
We study a relation between two hypothesis testing problems which have the same distribution in the null hypothesis. The first problem involves the following hypotheses
\begin{align}
  H_0&\colon (x^n,y^n)\sim P_{X^nY^n}\nonumber\\
  H_1&\colon (x^n,y^n)\sim Q_{Y^n}\times Q_{X^n}, \label{first_setup}
\end{align}
whereas the second problem considers 
\begin{align}
  \bar{H}_0&\colon (x^n,y^n)\sim P_{X^nY^n}\nonumber\\
  \bar{H}_1&\colon (x^n,y^n)\sim \bar{Q}_{Y^n}\times \bar{Q}_{X^n}. \label{second_setup}
\end{align}
We make the following technical assumptions
\begin{align}
  \forall n,\; & P_{X^n}\ll Q_{X^n},\;P_{Y^n}\ll Q_{Y^n},\nonumber\\
  &\; P_{X^n}\ll \bar{Q}_{X^n},\;P_{Y^n}\ll \bar{Q}_{Y^n}. \label{ab_conds}
\end{align}

 \begin{figure}[htb]
   \centering
   \includegraphics{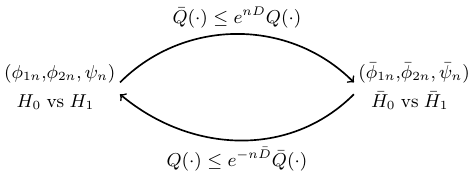}
   \caption{Constructive code transformations between two hypothesis testing problems. Each arrow indicates that given a testing scheme for the source node we can construct a testing scheme for the destination node. Each transformation allows us to change the measure of miss detection at the destination node to the measure of miss detection at the source node.}\label{code_transform_illu}
 \end{figure}
\noindent An overview of our derivation is depicted in Fig. \ref{code_transform_illu}. Let $(\phi_{1n},\phi_{2n},\psi_n)$ be a testing scheme for differentiating between $H_0$ and $H_1$ in \eqref{first_setup}. We construct a testing scheme $(\bar{\phi}_{1n},\bar{\phi}_{2n},\bar{\psi}_n)$ for testing $\bar{H}_0$ against $\bar{H}_1$ in \eqref{second_setup} as follows.
For an arbitrarily given $\gamma>0$, we define typical sets $\mathcal{B}_{n,\gamma}^1$, and $\mathcal{B}_{n,\gamma}^2$ as
\begin{align}
  \mathcal{B}_{n,\gamma}^1 &= \{x^n\mid \vert  [\log \bar{Q}_X^n(x^n) - \log Q_{X^n}(x^n)]/n-A_{X}\vert  <\gamma\},\nonumber\\
  \mathcal{B}_{n,\gamma}^2 &= \{y^n\mid \vert  [\log \bar{Q}_{Y^n}(y^n) - \log Q_{Y^n}(y^n)]/n-A_{Y}\vert  <\gamma\},\label{weak_law_sets}
\end{align}
where $A_{X}$ and $A_{Y}$ are finite numbers.\\
We define the compression mapping $\bar{\phi}_{1n}$ as follows
\begin{align}
  \bar{\phi}_{1n}\colon \mathcal{X}^n&\to \mathcal{M}_1\cup\{e_1\}\nonumber\\
  \bar{\phi}_{1n}(x^n)&\mapsto \begin{dcases}\phi_{1n}(x^n)\;&\text{if}\; x^n\in\mathcal{B}_{n,\gamma}^1,\\ 
  e_1\;&\text{otherwise}\end{dcases}.\label{transf_atx}
\end{align}
Similarly the compression mapping $\bar{\phi}_{2n}$ is defined as
\begin{align}
  \bar{\phi}_{2n}\colon \mathcal{Y}^n&\to \mathcal{M}_2\cup\{e_2\}\nonumber\\
  \bar{\phi}_{2n}(y^n)&\mapsto \begin{dcases}\phi_{2n}(y^n)\;&\text{if}\; y^n\in \mathcal{B}_{n,\gamma}^2,\\
  e_2\;&\text{otherwise}\end{dcases}.\label{transf_aty}
\end{align}
The corresponding decision mapping $\bar{\psi}_n$, is given as
\begin{align}
  \bar{\psi}_n\colon \prod_{k=1}^2 (\mathcal{M}_k\cup\{e_k\})&\to\{0,1\}\nonumber\\
  \bar{\psi}_n(\bar{u}_1,\bar{u}_2)&\mapsto \begin{dcases}\psi_n(\bar{u}_1,\bar{u}_2)\;&\text{if}\; \bar{u}_k\neq e_k,\; \forall k=1,2,\nonumber\\
  1&\;\text{otherwise}\end{dcases}.\label{transf_decs}
\end{align}
For notation simplicity we define for $k=1,2,$ and $u_k\in\mathcal{M}_k$, $\mathcal{W}_{u_k}^k =\phi_{kn}^{-1}(u_k)\cap \mathcal{B}_{n,\gamma}^k$.
The set of all $(x^n,y^n)$ for which $\bar{u}_1\neq e_1$ and $\bar{u}_2\neq e_2$ hold is
$
  \mathcal{B}_{n,\gamma}^1\times \mathcal{B}_{n,\gamma}^2,
$
which can be factorized further as
\begin{align}
  \bigcup_{(u_1,u_2)\in\mathcal{M}_1\times \mathcal{M}_2}(\mathcal{W}_{u_1}^1\times\mathcal{W}_{u_2}^2).
\end{align}
In differentiating between $P_{Y^nX^n}$ and $\bar{Q}_{Y^n}\times \bar{Q}_{X^n}$ the induced false alarm probability $\bar{\alpha}_n$  by the testing scheme $(\bar{\phi}_{1n},\bar{\phi}_{2n},\bar{\psi}_n)$  is upper bounded as

\begin{align}
  \bar{\alpha}_n &= 1-\sum_{(u_1,u_2)\in\mathcal{M}_1\times \mathcal{M}_2}P_{Y^nX^n}(\mathcal{W}_{u_1}^1\times\mathcal{W}_{u_2}^2)P_{H_0\vert  u_1,u_2}\nonumber\\
                    &\leq 1 - \sum_{(u_1,u_2)\in\mathcal{M}_1\times \mathcal{M}_2}P_{Y^nX^n}[\phi_{1n}^{-1}(u_1)\times \phi_{2n}^{-1}(u_2)]P_{H_0\vert  u_1,u_2}\nonumber\\
                    &+ P_{Y^nX^n}[([(\mathcal{B}_{n,\gamma}^1\times \mathcal{Y}^n)\cap(\mathcal{X}^n\times \mathcal{B}_{n,\gamma}^2)])^c]\nonumber\\
                    &\leq 1 -\sum_{(u_1,u_2)\in\mathcal{M}_1\times \mathcal{M}_2}P_{Y^nX^n}[\phi_{1n}^{-1}(u_1)\times \phi_{2n}^{-1}(u_2)]P_{H_0\vert  u_1,u_2} \nonumber\\
  &\quad+ P_{X^n}[(\mathcal{B}_{n,\gamma}^1)^c]+ P_{Y^n}[(\mathcal{B}_{n,\gamma}^2)^c]\nonumber\\
  & = \alpha_n +  P_{X^n}[(\mathcal{B}_{n,\gamma}^1)^c]+ P_{Y^n}[(\mathcal{B}_{n,\gamma}^2)^c].\label{false_alrm_bnd}
\end{align}
Similarly the probability of miss detection $\bar{\beta}_n$ is bounded by
\begin{align}
  \bar{\beta}_n &= \sum_{(u_1,u_2)\in\mathcal{M}_1\times \mathcal{M}_2}\bar{Q}_{X^n}(\mathcal{W}_{u_1}^1)\bar{Q}_{Y^n}(\mathcal{W}_{u_2}^2)P_{H_0\vert  u_1,u_2}\nonumber\\
                   &\leq e^{n( A_{X} + A_{Y} + 2\gamma)}\nonumber\\
  &\quad\times \sum_{(u_1,u_2)\in\mathcal{M}_1\times \mathcal{M}_2}Q_{X^n}(\mathcal{W}_{u_1}^1)Q_{Y^n}(\mathcal{W}_{u_2}^2)P_{H_0\vert  u_1,u_2}\nonumber\\
  &\leq e^{n( A_{X} + A_{Y} +2\gamma)}\beta_n.\label{miss_dect_bnd}
\end{align}
 
Given a testing scheme $(\bar{\phi}_{1n},\bar{\phi}_{2n},\bar{\psi}_n)$ for testing $P_{Y^nX^n}$ against $\bar{Q}_{Y^n}\times \bar{Q}_{X^n}$ we apply a similar procedure as from  \eqref{transf_atx} to \eqref{transf_decs}, by swapping the positions of $\bar{\phi}_{kn}$ and $\phi_{kn}$ for $k=1,2$, in \eqref{transf_atx} and \eqref{transf_aty} as well as switching the roles of $\bar{\psi}_n$ and $\psi_n$ in \eqref{transf_decs}, to obtain a testing scheme $(\phi_{1n},\phi_{2n},\psi_n)$ for testing $P_{Y^nX^n}$ against $Q_{Y^n}\times Q_{X^n}$.   

\noindent The induced false alarm probability $\alpha_n$ in testing  $H_0$ against $H_1$ is similarly upper bounded by
\begin{align}
  \alpha_n \leq \bar{\alpha}_n +  P_{X^n}[(\mathcal{B}_{n,\gamma}^1)^c]+ P_{Y^n}[(\mathcal{B}_{n,\gamma}^2)^c].\label{false_alrm_bnd_prime}
\end{align}
The induced miss detection probability $\beta_n$ in testing $H_0$ against $H_1$ is upper bounded as
\begin{align}
  \beta_n &= \sum_{(\bar{u}_1,\bar{u}_2)\in\bar{\mathcal{M}}_1\times \bar{\mathcal{M}}_2}Q_{X^n}(\mathcal{W}_{\bar{u}_1}^{1})Q_{Y^n}(\mathcal{W}_{\bar{u}_2}^{2})P_{\bar{H}_0\vert  \bar{u}_1,\bar{u}_2}\nonumber\\
                   &\leq e^{-n(A_{X} + A_{Y}- 2\gamma)}\nonumber\\
  &\quad\times \sum_{(\bar{u}_1,\bar{u}_2)\in\bar{\mathcal{M}}_1\times \bar{\mathcal{M}}_2}\bar{Q}_{X^n}(\mathcal{W}_{\bar{u}_1}^{1})\bar{Q}_{Y^n}(\mathcal{W}_{\bar{u}_2}^{2})P_{\bar{H}_0\vert  \bar{u}_1,\bar{u}_2}\nonumber\\
  &\leq e^{-n(A_{X} + A_{Y} -2\gamma)}\bar{\beta}_n.\label{miss_dect_bnd_prime}
\end{align}
For a given pair of compression rates $(R_1,R_2)$, let $\bar{E}_{\epsilon}^{\star}(R_1,R_2)$ be the maximum $\epsilon$-achievable error exponent for testing $\bar{H}_0$ against $\bar{H}_1$. Similarly let $E_{\epsilon}^{\star}(R_1,R_2)$ be the maximum $\epsilon$-achievable error exponent for testing $H_0$ against $H_1$. To establish a relation between $\bar{E}_{\epsilon}^{\star}(R_1,R_2)$ and $E_{\epsilon}^{\star}(R_1,R_2)$, we make the following assumption.
\begin{assumption}\label{assump_3}
  The sequences of joint distributions $(P_{X^nY^n})$, $(Q_{X^nY^n})$, $(\bar{Q}_{X^nY^n})$, and the quantities $A_{X}$, $A_{Y}$, satisfy \eqref{ab_conds} and the following conditions
  \begin{align}
    \forall \gamma>0,\; \lim_{n\to\infty} P_{X^n}(\mathcal{B}_{n,\gamma}^1)=1,\;    \lim_{n\to\infty} P_{Y^n}(\mathcal{B}_{n,\gamma}^2)=1. \label{aep_cond_gen}
  \end{align}
\end{assumption}
\noindent We give some examples in which Assumption \eqref{assump_3} is satisfied.
\begin{itemize}
\item Assume that
  \begin{align}
    P_{Y^n}= P_{Y}^{\otimes n},\; P_{X^n} &= P_{X}^{\otimes n},\; Q_{X^n} = Q_{X}^{\otimes n},\; Q_{Y^n} = Q_{Y}^{\otimes n},\nonumber\\
  \text{and}\;  \bar{Q}_{X^n} &= \bar{Q}_{X}^{\otimes n},\; \bar{Q}_{Y^n} = \bar{Q}_{Y}^{\otimes n},
  \end{align}
  hold such that the conditions in \eqref{ab_conds} are fulfilled. Furthermore we assume that
  \begin{align}
   A_{X} &= D(P_{X}\Vert Q_{X})- D(P_{X}\Vert \bar{Q}_{X}),\nonumber\\
    \text{and}\; A_{Y} &= D(P_{Y}\Vert Q_{Y}) - D(P_{Y}\Vert \bar{Q}_{Y})
    \end{align}
    are finite. By the weak law of large numbers, the conditions in \eqref{aep_cond_gen} are satisfied.
\item Assume that $(P_{Y^n})$, $(P_{X^n})$ are stationary and ergodic processes as well as $(Q_{Y^n})$, $(Q_{X^n})$, and $(\bar{Q}_{Y^n})$, $(\bar{Q}_{X^n})$ are finite order Markov processes with stationary transition probabilities such that the conditions in \eqref{ab_conds} are satisfied. Assume further that the relative divergence rates
  \begin{align}
    D_{X}&= \lim_{n\to\infty}[D(P_{X^{n+1}}\Vert Q_{X^{n+1}}) - D(P_{X^n}\Vert Q_{X^n})],\nonumber\\
    \bar{D}_{X}&= \lim_{n\to\infty}[D(P_{X^{n+1}}\Vert \bar{Q}_{X^{n+1}}) - D(P_{X^n}\Vert \bar{Q}_{X^n})],\nonumber\\
    D_{Y} &= \lim_{n\to\infty}[D(P_{Y^{n+1}}\Vert Q_{Y^{n+1}})-D(P_{Y^n}\Vert Q_{Y^n})],\nonumber\\
    \bar{D}_{Y} &= \lim_{n\to\infty}[D(P_{Y^{n+1}}\Vert \bar{Q}_{Y^{n+1}})-D(P_{Y^n}\Vert \bar{Q}_{Y^n})],
  \end{align}
  are finite. With $A_{X} = D_{X} - \bar{D}_{X}$, and $A_{Y} = D_{Y} - \bar{D}_{Y}$, the conditions in \eqref{aep_cond_gen} are fulfilled by \cite[Theorem 1]{barron1985strong}.
\end{itemize}
\noindent Assume that Assumption \ref{assump_3} is fulfilled, then \eqref{false_alrm_bnd} and \eqref{miss_dect_bnd} imply that
\begin{align}
 \bar{E}_{\epsilon}^{\star}(R_1,R_2) \geq E_{\epsilon}^{\star}(R_1,R_2) - [A_{X} + A_{Y}].
\end{align}
Conversely, \eqref{false_alrm_bnd_prime} and \eqref{miss_dect_bnd_prime} imply that
\begin{align}
  E_{\epsilon}^{\star}(R_1,R_2)\geq  \bar{E}_{\epsilon}^{\star}(R_1,R_2) + [A_{X} + A_{Y}].
\end{align}
In conclusion we have shown the following result.
\begin{theorem}\label{thm_codetrsfmis}
  Given $\epsilon\in [0,1)$ and $(R_1,R_2)\in\mathbb{R}_{+}^2$, under Assumption \ref{assump_3} we have
  \begin{align}
    E_{\epsilon}^{\star}(R_1,R_2) = \bar{E}_{\epsilon}^{\star}(R_1,R_2) + [A_{X} + A_{Y}].
  \end{align}  
\end{theorem}
As a corollary for this result by setting $\bar{Q}_{X^n} = P_{X}^{\otimes n}$ and $\bar{Q}_{Y^n} = P_{Y}^{\otimes n}$ we have the following result which states that Theorem 5 in \cite{ahlswede1986hypothesis} is tight for a special case.
\begin{theorem}\label{theorem_acs_extended}
  Assume that the sequences of distributions $(Q_{Y^n})$ and $(Q_{X^n})$ satisfy Assumption \ref{assump_1} with $m=1$, $r=1$, $S=1$, and $P_{Y_1X_1} = P_{YX}$. In testing $P_{YX}^{\otimes n}$ against $Q_{Y^n}\times Q_{X^n}$ using one-sided compression of the sequence $x^n$, the maximum $\epsilon$-achievable error exponent is given by
  \begin{align}
    E_{\epsilon}^{\star}(R_c) = \max_{P_{U\vert  X}\colon I(U;X)\leq R_c}I(Y;U) + d^{\mathrm{yx}},\; \forall \epsilon\in [0,1),
  \end{align}
where in the first category $d^{\mathrm{yx}} = D(P_X\Vert Q_X) + D(P_Y\Vert Q_Y)$, and in the second category $d^{\mathrm{yx}} = A_X^{(11)} + A_Y^{(11)}$.  
\end{theorem}

\section{Proof of Theorem \ref{thm_1}}\label{proof_thm_1}
\subsection{A support lemma}
For the achievability proof of Theorem \ref{thm_1} we need the following support lemma.
\begin{lemma}\label{lemma_mix_sup}
Let $(\{P_{V_{\eta}}^{\otimes n}\}, \{Q_{V_{\tau}^n}\})$ be either $(\{P_{\mathcal{X},s}^{\otimes n}\}, \{Q_{X_t^n}\})$ or $(\{P_{Y_i}^{\otimes n}\}, \{Q_{Y_j^n}\})$ in Assumption \ref{assump_1}. For each $\eta$, let $d_{\eta}^{\star}$ be the corresponding $d_s^{\mathrm{x}}$ or $d_i^{\mathrm{y}}$. For each $n$, define $P_{V^n} = \sum_{s}\nu_{\eta}P_{V_{\eta}}^{\otimes n}$ where for simplicity $(\nu_{\eta})$ are fixed positive numbers and $\sum_{\eta} \nu_{\eta} = 1$. Also for each $\eta$, we assume that $V_{\eta}^n\sim P_{V_{\eta}}^{\otimes n}$. Then for an arbitrarily given $\gamma>0$ we have 
\begin{align}
\lim_{n\to\infty}\mathrm{Pr}\{\min_{\tau}\iota_{P_{V^n}\Vert Q_{V_{\tau}^n}}(V_{\eta}^n)<n(d_{\eta}^{\star}-\gamma)\} = 0.
\end{align}
\end{lemma}
\begin{proof}
It can be seen that
\begin{align}
 & \mathrm{Pr}\{\min_{\tau}\iota_{P_{V^n}\Vert Q_{V_{\tau}^n}}(V_{\eta}^n)\leq  n(d_{\eta}^{\star}-\gamma)\}\nonumber\\
  &\leq \sum_{\tau}\mathrm{Pr}\{\iota_{P_{V^n}\Vert Q_{V_{\tau}^n}}(V_{\eta}^n)\leq  n(d_{\eta}^{\star}-\gamma)\}\nonumber\\
  &\leq\sum_{\tau}\mathrm{Pr}\{\iota_{P_{V_{\eta}}^{\otimes n}\Vert Q_{V_{\tau}^n}}(V_{\eta}^n)\leq  n(d_{\eta}^{\star}-\gamma-\log\nu_{\eta}/n)\}.
\end{align}
If $\{Q_{V_{\tau}^n}\}$ belongs to the second category, we always have $P_{V_{\eta}}^{\otimes n}\ll Q_{V_{\tau}^n}$ for all $(\eta,\tau)$. This leads to $P_{V^n}\ll Q_{V_{\tau}^n}$ for all $\tau$. The second inequality then follows since $P_{V^n}(v^n)\geq \nu_{\eta}P_{V_{\eta}}^{\otimes n}(v^n)$ for all $v^n$ holds.

Assume now that $Q_{V_{\tau}^n}$ lies in the first category. When $P_{V^n}\ll Q_{V_{\tau}^n}$ holds, i.e., for all $\eta$ we have $P_{V_{\bar{\eta}}}^{\otimes n}\ll Q_{V_{\tau}^n}$, then the set $\{v^n\mid Q_{V_{\tau}^n}(v^n)=0\}$ can be omitted since it has zero probability. If $P_{V^n}\nll Q_{V_{\tau}^n}$, for example when there exists a $\bar{\eta}$ satisfying $P_{V_{\bar{\eta}}}^{\otimes n}\nll Q_{V_{\tau}^n}$, then for all $v^n$ such that $Q_{V_{\tau}^n}(v^n)=0$ we have $\iota_{P_{V^n}\Vert Q_{V_{\tau}^n}}(v^n)=+\infty$ which in turn violates the inequality $\iota_{P_{V^n}\Vert Q_{V_{\tau}^n}}(v^n)\leq  n(d_{\eta}^{\star}-\gamma)$. Therefore in both cases we only need to consider the set $\{v^n\mid Q_{V_{\tau}^n}(v^n)>0\}$, in which the second inequality follows from the definitions of $\iota_{P_{V^n}\Vert Q_{V_{\tau}^n}}$, and $\iota_{P_{V_{\eta}}^{\otimes n}\Vert Q_{V_{\tau}^n}}$ as well as the fact that $P_{V^n}(v^n)\geq \nu_{\eta}P_{V_{\eta}}^{\otimes n}(v^n)$ for all $v^n$ holds. For notation simplicity, we define $E_{\eta,n} = d_{\eta}^{\star}-\gamma-\log\nu_{\eta}/n$. We examine two cases in Assumption \ref{assump_1} separately in the following.
\begin{itemize}
\item Consider the first case in which $\{Q_{V_{\tau}^n}\}$ is the set of product distributions on a common alphabet $\mathcal{V}$.
To avoid dealing with cumbersome extended real number operations, we define a set $\mathcal{A}_{\eta,n} = \{v^n\mid v^n\in\mathcal{V}^n,\; P_{V_{\eta}}^{\otimes n}(v^n)>0\}$. We further have\footnote{Without considering $\mathcal{A}_{\eta,n}$, we might encounter an indeterminate form $\dots+\infty+\dots-\infty+\dots$ in the third expression.}
\begin{align}
\mathrm{Pr}\big\{&\iota_{P_{V_{\eta}}^{\otimes n}\Vert Q_{V_{\tau}}^{\otimes n}}(V_{\eta}^n)<nE_{\eta,n}\big\}\nonumber\\
& = \mathrm{Pr}\big\{\iota_{P_{V_{\eta}}^{\otimes n}\Vert Q_{V_{\tau}}^{\otimes n}}(V_{\eta}^n)<nE_{\eta,n}, V_{\eta}^n\in\mathcal{A}_{\eta,n}\big\}\nonumber\\
& = \mathrm{Pr}\big\{\sum_{l}\iota_{P_{V_{\eta l}}\Vert Q_{V_{\tau l}}}(V_{\eta l})<nE_{\eta,n}, V_{\eta}^n\in\mathcal{A}_{\eta,n}\big\}.\label{last_lem_1}
\end{align}

If $D(P_{V_{\eta}}\Vert Q_{V_{\tau}})<\infty$ then by the weak law of large numbers \eqref{last_lem_1} goes to 0 as $d_{\eta}^{\star}\leq D(P_{V_{\eta}}\Vert Q_{V_{\tau}})$. When $D(P_{V_{\eta}}\Vert Q_{V_{\tau}})=+\infty$, let $\mathcal{B}_{\eta\tau}$ be the largest subset of $\mathcal{V}$ such that $P_{V_{\eta}}(v)>0$ and $Q_{V_{\tau}}(v)=0$ for all $v\in\mathcal{B}_{\eta\tau}$. By our assumption we have $P_{V_{\eta}}(\mathcal{B}_{\eta\tau})>0$. And if $v\in\mathcal{B}_{\eta\tau}$ then $\iota_{P_{V_{\eta}}\Vert Q_{V_{\tau}}}(v)=+\infty$ holds. Therefore we have
\begin{align}
&\{v^n\mid \sum_{l}\iota_{P_{V_{\eta l}}\Vert Q_{V_{\tau l}}}(v_l)<nE_{\eta,n},\;v^n\in\mathcal{A}_{\eta,n}\}\nonumber\\
&\subseteq \{v^n\mid v_l\in\mathcal{B}_{\eta\tau}^c,\;\forall l \in [1:n],\;v^n\in\mathcal{A}_{\eta,n}\}.
\end{align}
Therefore 
\begin{align}
\mathrm{Pr}&\big\{\sum_{l}\iota_{P_{V_{\eta l}}\Vert Q_{V_{\tau l}}}(V_{\eta l})<nE_{\eta,n},V_{\eta}^n\in \mathcal{A}_{\eta,n}\big\}\nonumber\\
&\leq (1-P_{V_{\eta}}(\mathcal{B}_{\eta\tau}))^n\to 0,\;\text{as}\; n\to\infty.
\end{align}
\item Next consider the case in which $\{Q_{V_{\eta}^n}\}$ is the set of finite order Markov processes satisfying Assumption \ref{assump_1}. By \cite[Theorem 1]{barron1985strong} we have with probability 1
\begin{align}
\frac{1}{n}\iota_{P_{V_{\eta}}^{\otimes n}\Vert Q_{V_{\tau}^n}}(V_{\eta}^n)\to A^{\eta\tau},\;\forall \eta,\tau,
\end{align}
where $A^{\eta\tau}$ is either $A_X^{(st)}$ or $A_Y^{(ij)}$. This holds even when $A^{\eta\tau}=+\infty$. Since, by definition $d_{\eta}^{\star} = \min_{\tau}A^{\eta\tau}<+\infty$, the conclusion also holds in this case.
\end{itemize}
\end{proof}
\subsection{Existence of a testing scheme}
Let $t_{\min} = \min_{s,\bar{s}\in [1:\vert \mathcal{P}_{\mathcal{X}}\vert ]} d_{TV}(P_{\mathcal{X},s},P_{\mathcal{X},{\bar{s}}})$ be the minimum (total variation) distance between any two probability distributions in the set of marginal distributions on $\mathcal{X}$, $\mathcal{P}_{\mathcal{X}}$. Denote by $P_{x^n}$ the type of a sequence $x^n$. We define a mapping $T\colon\mathcal{X}^n\to[1:\vert \mathcal{P}_{\mathcal{X}}\vert ]\cup \{e\}$ as follows. We look for a unique $s\in[1:\vert \mathcal{P}_{\mathcal{X}}\vert ]$ such that $d_{TV}(P_{x^n},P_{\mathcal{X},s})<t_{\min}/2$. If such an $s$ exists, we set $T(x^n)=s$. Otherwise we set $T(x^n)=e$.\\
For a given $R_c>0$ and for each $s\in [1:\vert \mathcal{P}_{\mathcal{X}}\vert ]$, let $P_{U_s\vert  \bar{X}_s}$ be a probability kernel which achieves $\theta_s(R_c)$. Define
\begin{align}
  P_{Y^nX^nU^n} &= \sum_{i=1}^m\nu_iP_{Y_iX_iU_i}^{\otimes n}, \;\text{where}\; \nu_i>0, \forall i\in [1:m],\;\sum_{i=1}^m\nu_i = 1,\nonumber\\
  &\text{and}\; P_{U_i\vert  X_i} = P_{U_s\vert  \bar{X}_s},\;\forall i\in\mathfrak{F}_s,\; s\in [1:\vert \mathcal{P}_{\mathcal{X}}\vert ].
\end{align}
Let $P_{Y^nU^n}$ be the corresponding marginal joint distribution on $\mathcal{Y}^n\times\mathcal{U}^n$. For notation simplicity we define the following score function\footnote{By our assumption $\min_{j\in [1:k]}\iota_{P_{Y^n}\Vert Q_{Y_j^n}}(y^n)< +\infty$, as if $P_{Y^n}(y^n)>0$ then there exists an $j$ such that $Q_{Y_j^n}(y^n)>0$, therefore $\zeta$ is well-defined.}
\begin{align}
\zeta(y^n,u^n) = \iota_{P_{Y^nU^n}}(y^n;u^n) + \min_{j\in[1:k]}\iota_{P_{Y^n}\Vert Q_{Y_j^n}}(y^n).
\end{align}
For each $i\in [1:m]$, and $n$, let $(Y_i^n,X_i^n)$ be a tuple of random variables such that $(Y_i^n,X_i^n)\sim P_{Y_iX_i}^{\otimes n}$. Let $\gamma>0$ be arbitrary but given. For each $s\in [1:\vert \mathcal{P}_{\mathcal{X}}\vert ]$ draw $M$ codewords $\{u_{s}^n(j)\}_{j=1}^{M}$, $M=e^{n(R_c+2\gamma)}$, from the marginal distribution $P_{U_s}^{\otimes n}$ of $(P_{U_s\vert  \bar{X}_s}\times P_{\mathcal{X},s})^{\otimes n}$. Given $x^n$, let $\hat{s}=T(x^n)$ be an estimate of the index of the marginal distribution of $x^n$. Assume that $\hat{s}\neq e$. We proceed with the compression process if
\begin{align}
  \min_{t}\iota_{P_{X^n}\Vert Q_{X_t^n}}(x^n)>  n(d_{\hat{s}}^{\mathrm{x}}-\gamma), \label{x_compress_cond}
\end{align}
where $P_{X}^n = \sum_{i}\nu_i P_{X_i}^{\otimes n}$ is the marginal on $\mathcal{X}^n$ of $P_{Y^nX^nU^n}$. Otherwise we send a special index $j_{\hat{s}}^{\star} = e^{\star}$.

When \eqref{x_compress_cond} is fulfilled, we then select a transmission index $j\in[1:M]$ to be $j^{\star}_{\hat{s}} = \arg\min_{j\in [1:M]}\pi_{{\hat{s}}}(x^n;u_{\hat{s}}^n(j))$ where for all $s\in [1:\vert \mathcal{P}_{\mathcal{X}}\vert ]$
\begin{equation}
  \pi_{s}(x^n;u^n) = \max_{i\in\mathfrak{F}_s}\mathrm{Pr}\{\zeta(Y_i^n,u^n)\leq n(E-d_s^{\mathrm{x}})\vert  X_i^n=x^n\}. \label{pi_funcs}
\end{equation}
If $\hat{s}=e$, we set $j^{\star}_{\hat{s}}=1$. The pair $(\hat{s},j^{\star}_{\hat{s}})$ is provided to the decision center.
If $\hat{s}\neq e$, and $j^{\star}_{\hat{s}}\neq e^{\star}$ hold we declare $H_0$ is the underlying distribution if
\begin{align}
  \zeta(y^n,u_{\hat{s}}^n(j^{\star}_{\hat{s}}))>n(E-d_{\hat{s}}^{\mathrm{x}}).\label{decision_thm_1}
\end{align}
Otherwise we declare that $H_1$ is true. 
Given a codebook realization, let $\hat{S}$ and $J^{\star}_{\hat{S}}$ be induced random variables under the null hypothesis.
 For each $i\in [1:m]$, $i\in\mathfrak{F}_s$, the corresponding false alarm probability is given by
\begin{align}
  \alpha_n^{(i)} &= \mathrm{Pr}\{[\zeta(Y_i^n,u_{\hat{S}}^n(J^{\star}_{\hat{S}}))\leq n(E-d_{\hat{S}}^{\mathrm{x}}),\; \text{and}\; \hat{S}\neq e,\; \text{and}\; J^{\star}_{\hat{S}} \neq e^{\star} ],\nonumber\\
  &\hspace{2cm}\text{or}\;[\hat{S}=e],\;\text{or}\; [J^{\star}_{\hat{S}} = e^{\star}]\}\nonumber\\
                 &\leq \underbrace{\mathrm{Pr}\{\zeta(Y_i^n,u_{s}^n(J^{\star}_s))\leq n(E-d_s^{\mathrm{x}})\}}_{\tau_i^1} + \mathrm{Pr}\{T(X_i^n)\neq s\}\nonumber\\
                 &\qquad +\underbrace{\mathrm{Pr}\{\min_{t}\iota_{P_{X^n}\Vert Q_{X_t^n}}(X_i^n)<n(d_s^{\mathrm{x}}-\gamma)\}}_{\tau_s^3}.
\end{align}
For notation simplicity we define for each $s$, $E_s = E-d_s^{\mathrm{x}}$. The first term can be upper bounded further as
\begin{align}
  \tau_i^1&= \int \mathrm{Pr}\{\zeta(Y_i^n,u^n))\leq nE_s\vert  X_i^n=x^n\}dP_{X_i^nu_{s}^n(J^{\star}_s)}(x^n,u^n)\nonumber\\
  &\leq \int \max_{i^{\prime}\in\mathfrak{F}_s}\mathrm{Pr}\{\zeta(Y_{i^{\prime}}^n,u^n))\leq nE_s\vert  X_{i^{\prime}}^n=x^n\}dP_{X_i^nu_{s}^n(J^{\star}_s)}(x^n,u^n)\nonumber\\
  &\stackrel{\eqref{pi_funcs}}{=} \mathrm{E}[\pi_s(X_i^n;u_{s}^n(J^{\star}_s))]= \int_{[0,1]}\mathrm{Pr}\{\pi_s(X_i^n;u_{s}^n(J^{\star}_s))>t\}dt.
\end{align}
Let $(\bar{X}_s^n,\bar{U}_s^n)$ be another tuple of generic random variables that is independent of the random codebook $(U_s^n(j))$ such that $(\bar{X}_s^n,\bar{U}_s^n)\sim (P_{\mathcal{X},s}\times P_{U_s\vert  \bar{X}_s})^{\otimes n}$. Furthermore for the given codebook realization let $\bar{J}^{\star}_s$ be the random message induced by $\bar{X}_s^n$ through the encoding process.
Then we have
\begin{align}
  \alpha_n = \max_{i\in [1:m]}\alpha_n^{(i)}\leq &\sum_{s}\int_{[0,1]}\mathrm{Pr}\{\pi_s(\bar{X}_s^n;u_{s}^n(\bar{J}^{\star}_s))>t\}dt\nonumber\\
  & + \sum_{s}\underbrace{\mathrm{Pr}\{T(\bar{X}_s^n)\neq s\}}_{\tau_s^2} + \sum_s\tau_s^3.
\end{align}  
Averaging over all codebooks we obtain by Fubini's theorem
\begin{align}
  \mathbb{E}[\alpha_n]\leq \sum_{s}\int_{[0,1]}\mathrm{Pr}\{\pi_s(\bar{X}_s^n;U_{s}^n(\bar{J}^{\star}_s))>t\}dt + \sum_{s}(\tau_s^2 + \tau_s^3).
\end{align}
 By the non-asymptotic covering lemma in \cite[Lemma 5]{verdu2012non} we have
\begin{align}
  \mathrm{Pr}\{\pi_s(\bar{X}_s^n;U_{s}^n(\bar{J}^{\star}_s))>t\} &\leq \mathrm{Pr}\{\pi_s(\bar{X}_s^n;\bar{U}_s^n)>t\} + e^{-\exp(n\gamma)}\nonumber\\
  &+ \underbrace{\mathrm{Pr}\{\iota_{P_{\bar{X}_s^n\bar{U}_s^n}}(\bar{X}_s^n;\bar{U}_s^n)>n(R_c+\gamma)\}}_{\tau_s^4}.
\end{align}
This result implies the following chain of expressions
\begin{align}
  \mathbb{E}[\alpha_n]&\leq \sum_{s}\int_{[0,1]}\mathrm{Pr}\{\pi_s(\bar{X}_s^n;\bar{U}_s^n)>t\}dt + S e^{-\exp(n\gamma)}+ \sum_{s}(\tau_s^2 + \tau_s^3+\tau_s^4)\nonumber\\
  &= \sum_{s}\mathbb{E}[\pi_s(\bar{X}_s^n;\bar{U}_s^n)] + S e^{-\exp(n\gamma)}+ \sum_{s}(\tau_s^2 + \tau_s^3+\tau_s^4).
\end{align}
For each $i\in\mathfrak{F}_s$ let $\bar{Y}_i^n$ be a tuple of random variable such that $\bar{Y}_i^n-\bar{X}_s^n-\bar{U}_s^n$ and $P_{\bar{Y}_i^n\vert  \bar{X}_s^n} = P_{Y_i\vert  X_i}^{\otimes n}$ hold. From the definition of $\pi_s(\cdot;\cdot)$ we also have
\begin{align}
  \mathbb{E}[\pi_s(\bar{X}_s^n;\bar{U}_s^n)]\leq\sum_{i\in\mathfrak{F}_s} \mathrm{Pr}\{\zeta(\bar{Y}_i^n,\bar{U}_s^n)\leq nE_s\}.
\end{align}
In summary there exists a codebook realization, hence a mapping $\phi_n\colon \mathcal{X}^n\to \mathcal{M}\triangleq([1:\vert \mathcal{P}_{\mathcal{X}}\vert ]\cup\{e\})\times([1:M]\cup\{e^{\star}\})$, such that
\begin{align}
  \alpha_n\leq&\sum_s\sum_{i\in\mathfrak{F}_s} \mathrm{Pr}\{\zeta(\bar{Y}_i^n,\bar{U}_s^n)\leq nE_s\} \nonumber\\
  &+ Se^{-\exp(n\gamma)}+ \sum_{s}(\tau_s^2+\tau_s^3+\tau_s^4),
\end{align}
holds.
\subsection{Bounding error probabilities}
 Due to Lemma \ref{lemma_mix_sup}, $\tau_s^3$ goes to 0 as $n\to\infty$. By the weak law of large numbers the terms $\tau_s^2$ and $\tau_s^4$ go to $0$ as $n\to\infty$. We focus now on the first term. We observe that
\begin{align}
  \mathrm{Pr}\{\zeta(\bar{Y}_i^n,\bar{U}_s^n)\leq nE_s\}&\leq \mathrm{Pr}\{\iota_{P_{Y^nU^n}}(\bar{Y}_i^n;\bar{U}_s^n)\leq n(E_s-d_i^{\mathrm{y}}+\gamma)\}\nonumber\\
  &+ \mathrm{Pr}\{\min_{j\in[1:k]}\iota_{P_{Y^n}\Vert Q_{Y_j^n}}(\bar{Y}_i^n)<n(d_i^{\mathrm{y}}-\gamma)\}.
\end{align}
Similarly, the last term goes to 0 due to Lemma \ref{lemma_mix_sup}. In the next step we perform several change of measure steps from the general distribution $P_{Y^nU^n}$ to our distributions of interest $P_{Y_iU_i}^{\otimes n}$, $i\in [1:m]$.
For that purpose, for each $i\in \mathfrak{F}_s$, we define the following sets 
\begin{align}
 \mathcal{A}_{is} &= \{(y^n,u^n)\mid P_{\bar{Y}_i^n\bar{U}_s^n}(y^n,u^n)>0\},\nonumber\\
 \mathcal{B}_i &= \{y^n\mid P_{Y^n}(y^n)\leq e^{n\gamma_n}P_{\bar{Y}_i}^{\otimes n}(y^n)\},\nonumber\\
 \mathcal{C}_s & = \{u^n\mid P_{U^n}(u^n)\leq e^{n\gamma_n}P_{\bar{U}_s}^{\otimes n}(u^n)\},
\end{align}
 where $\gamma_n\to 0$ and $n\gamma_n\to\infty$ as $n\to\infty$. We have
\begin{align}
P_{\bar{Y}_i\bar{U}_s}^{\otimes n}[(\mathcal{B}_i\times \mathcal{U}^n)^c\cap\mathcal{A}_{is}]\leq P_{\bar{Y}_i}^{\otimes n}[(\mathcal{B}_i)^c]\leq e^{-n\gamma_n},\nonumber\\
P_{\bar{Y}_i\bar{U}_s}^{\otimes n}[(\mathcal{Y}^n\times \mathcal{C}_s)^c\cap\mathcal{A}_{is}]\leq P_{\bar{U}_s}^{\otimes n}[(\mathcal{C}_s)^c]\leq e^{-n\gamma_n}.
\end{align}
For $(y^n,u^n)\in\mathcal{A}_{is}\cap(\mathcal{B}_i\times \mathcal{U}^n)\cap(\mathcal{Y}^n\times \mathcal{C}_s)$ as $P_{Y^nU^n}(y^n,u^n)\geq \nu_iP_{\bar{Y}_i\bar{U}_s}^{\otimes n}(y^n,u^n)$ holds by the definition of $P_{Y^nU^n}$, we further have
\begin{equation}
\iota_{P_{Y^nU^n}}(y^n;u^n)\geq \log\nu_i + \iota_{P_{\bar{Y}_i^n\bar{U}_s^n}}(y^n;u^n)-2n\gamma_n.
\end{equation}
This leads to
\begin{align}
\mathrm{Pr}&\{\iota_{P_{Y^nU^n}}(\bar{Y}_i^n;\bar{U}_s^n)\leq n(E_s-d_i^{\mathrm{y}}+\gamma)\}\nonumber\\
& = \mathrm{Pr}\{\iota_{P_{Y^nU^n}}(\bar{Y}_i^n;\bar{U}_s^n)\leq n(E-d_{is}^{\mathrm{yx}}+\gamma),\;(\bar{Y}_i^n,\bar{U}_s^n)\in\mathcal{A}_{is}\} \nonumber\\
&\leq \mathrm{Pr}\{\iota_{P_{Y^nU^n}}(\bar{Y}_i^n;\bar{U}_s^n)\leq n(E-d_{is}^{\mathrm{yx}}+\gamma),\;\nonumber\\
&\hspace{2cm}(\bar{Y}_i^n,\bar{U}_s^n)\in\mathcal{A}_{is}\cap(\mathcal{B}_i\times \mathcal{U}^n)\cap(\mathcal{Y}^n\times \mathcal{C}_s)\} + 2e^{-n\gamma_n}\nonumber\\
&\leq \mathrm{Pr}\big\{\iota_{P_{\bar{Y}_i^n\bar{U}_s^n}}(\bar{Y}_i^n;\bar{U}_s^n)\leq n(E-d_{is}^{\mathrm{yx}}-\log\nu_i/n+2\gamma_n+\gamma)\big\} + 2e^{-n\gamma_n}.\label{two_terms}
\end{align}
For an arbitrary $\gamma>0$ select $E = \min_s\theta_s(R_c)-2\gamma$ which implies that $E-d_{is}^{\mathrm{yx}}+\gamma<I(\bar{Y}_i;\bar{U}_s)$ holds. By the weak law of large numbers we obtain
\begin{align}
\lim_{n\to\infty}  \mathrm{Pr}\big\{\iota_{P_{\bar{Y}_i^n\bar{U}_s^n}}(\bar{Y}_i^n;\bar{U}_s^n)&\leq n(E-d_{is}^{\mathrm{yx}}-\log\nu_i/n+2\gamma_n+\gamma)\big\} = 0,\nonumber\\
&\forall i\in [1:m], i\in\mathfrak{F}_s.
\end{align}
\noindent Therefore we have
\begin{align}
  \mathrm{Pr}\{\zeta(\bar{Y}_i^n,\bar{U}_s^n)\leq nE_s\}&\to 0,\;\forall s, \;i\in \mathfrak{F}_s,\; \text{as}\; n\to\infty,\nonumber\\
  \Rightarrow \alpha_n&\to 0,\;\text{as}\; n\to\infty.
\end{align}
Define the following sets
\begin{align}
  \mathcal{G} = \{(y^n,\phi_n(x^n))\mid &\zeta(y^n,u^n(\phi_n(x^n)))> n(E-d_{\hat{s}}^{\mathrm{x}}),\nonumber\\
  &\phi_n(x^n)\notin \{(e,1),(1,e^{\star}),\dots,(S,e^{\star})\}\}
\end{align}
and
\begin{align}
  \mathcal{H}_j = \{y^n\mid Q_{Y_j^n}(y^n)>0\},\; j\in [1:k].
\end{align}
$\mathcal{G}$ is our decision region described in \eqref{decision_thm_1}. Using $\mathcal{G}$ and $\{\mathcal{H}_j\}$ we perform in the following change of measure steps from $Q_{Y_j^n}$ to $P_{Y^n\vert  U^n}$ and from $Q_{\phi_n(X_t^n)}$ to $P_{\phi_n(X^n)}$ in the calculation of the miss detection probability. For each $j\in [1:k]$ and $(y^n,\phi_n(x^n))\in\mathcal{G}\cap (\mathcal{H}_j\times \mathcal{M})$ we have
\begin{align}
&\zeta(y^n,u^n(\phi_n(x^n)))> n(E-d_{\hat{s}}^{\mathrm{x}})\nonumber\\
\Rightarrow &\iota_{P_{Y^nU^n}}(y^n,u^n(\phi_n(x^n)) + \iota_{P_{Y^n}\Vert Q_{Y_j^n}}(y^n)>n(E-d_{\hat{s}}^{\mathrm{x}})\nonumber\\
  \Rightarrow &\log\frac{P_{Y^n\vert  U^n}(y^n\vert  u^n(\phi_n(x^n)))}{Q_{Y_j^n}(y^n)}>n(E-d_{\hat{s}}^{\mathrm{x}}),\nonumber\\
  \Rightarrow &Q_{Y_j^n}(y^n)\leq e^{-n(E-d_{\hat{s}}^{\mathrm{x}})}P_{Y^n\vert  U^n}(y^n\vert  u^n(\phi_n(x^n))),
\end{align}
where the second last expression follows since $(y^n,\phi_n(x^n))\in\mathcal{G}\cap(\mathcal{H}_j\times \mathcal{M})$ implies that $P_{Y^n}(y^n)>0$, otherwise in the previous line the first term is $0$ while the second term is $-\infty$, and due to our coding arguments $P_{U^n}(u^n(\phi_n(x^n)))>0$. Furthermore for $(y^n,\phi_n(x^n))\in\mathcal{G}$, as $\phi_n(x^n)\notin \{(1,e^{\star}),\dots,(\vert \mathcal{P}_{\mathcal{X}}\vert,e^{\star})\}$ holds, we also have
\begin{align}
  Q_{\phi_n(X_t^n)}(\phi_n(x^n))\leq P_{\phi_n(X^n)}(\phi_n(x^n))e^{-n(d_{\hat{s}}^{\mathrm{x}}-\gamma)},\;\forall t\in [1:r].
\end{align}
Therefore, for each $j\in [1:k]$ and $t\in [1:r]$ the probability of miss detection is bounded by
\begin{align}
  \beta_n^{(jt)} &= Q_{Y_j^n}\times Q_{\phi_n(X_t^n)}(\mathcal{G})= Q_{Y_j^n}\times Q_{\phi_n(X_t^n)}(\mathcal{G}\cap(\mathcal{H}_j\times \mathcal{M}))\nonumber\\
                 &\leq e^{-n(d_{\hat{s}}^{\mathrm{x}}-\gamma)}e^{-n(E-d_{\hat{s}}^{\mathrm{x}})}\nonumber\\
                 &\times\sum_{(y^n,\phi_n(x^n))\in \mathcal{G}\cap(\mathcal{H}_j\times \mathcal{M})}P_{\phi_n(X^n)}(\phi_n(x^n))P_{Y^n\vert  U^n}(y^n\vert  u^n(\phi_n(x^n)))\nonumber\\
  &\leq e^{-n(E-\gamma)}.
\end{align}
This implies that we have
\begin{align}
  \beta_n = \max_{j\in[1:k],t\in[1:r]}\beta_n^{(jt)}\leq e^{-n(E-\gamma)}.
\end{align}
 Therefore, the chosen sequence of $\phi_n$ satisfies
\begin{align}
  \lim_{n\to\infty}\alpha_n = 0,\;\liminf_{n\to\infty}\frac{1}{n}\log\frac{1}{\beta_n}\geq E-\gamma.
\end{align}
Combining with the definition of intersected sets $\mathcal{I}_n^{(i)}(E)$ this further implies that for all $i\in [1:m]$ we have
\begin{align}
  \lim_{n\to\infty}P_{Y_i^n\phi_n(X_i^n)}[\mathcal{I}_{n}^{(i)}(E-2\gamma)^c]\leq \lim_{n\to\infty}(\alpha_n^{(i)} + e^{n(E-2\gamma)}\sum_{j,t}\beta^{(jt)}_n) = 0.
\end{align}

\section{Proof of Theorem \ref{thm_2}}\label{proof_thm_2}
\noindent For an arbitrarily given $\gamma>0$, define the following typical sets
\begin{align}
  \mathcal{B}_{n,\gamma}^{(i)} &= \{y^n\mid \vert  \iota_{P_{Y_i}^{\otimes n}\Vert Q_{Y_{j_i^{\star}}^n}}(y^n)/n - d_i^{\mathrm{y}}\vert  <\gamma\},\; i\in [1:m],\nonumber\\
  \mathcal{B}_{n,\gamma}^s &=  \{x^n\mid \vert  \iota_{P_{\mathcal{X},s}^{\otimes n}\Vert Q_{X_{t_s^{\star}}}^n}(x^n)/n - d_s^{\mathrm{x}}\vert  <\gamma\},\;s\in [1:\vert \mathcal{P}_{\mathcal{X}}\vert ].\label{bigB_def}
\end{align}
We have
\begin{align}
  \lim_{n\to\infty}&P_{\mathcal{X},s}^{\otimes n}(\mathcal{B}_{n,\gamma}^s) = 1,\;\forall s\in [1:\vert \mathcal{P}_{\mathcal{X}}\vert ],\nonumber\\
  \lim_{n\to\infty}&P_{Y_i}^{\otimes n}(\mathcal{B}_{n,\gamma}^{(i)}) = 1,\;\forall i\in [1:m].
\end{align}
either due to the weak law of large numbers or due to \cite[Theorem 1]{barron1985strong}.
For simplicity we first consider the case that there is a single marginal distribution in the set $\mathcal{P}_{\mathcal{X}}$. In this case we simply write $P_{\mathcal{X},1}$ as $P_X$, $t_1^{\star}$ as $t^{\star}$, $\mathcal{B}_{n,\gamma}^1$ as $\mathcal{B}_{n,\gamma}$, and $\theta_s(R_c)$ as $\theta(R_c)$. Furthermore, for notation compactness we define the following quantities in this case
\begin{align}
  \bar{d}_i^{\star} = d_i^{\mathrm{y}} + d_1^{\mathrm{x}},\;\forall i\in [1:m].
\end{align}
To support to analysis we define $\bar{P} = \prod P_{Y_i\vert  X_i}\times P_{X}$. Given an arbitrary joint distribution $P_{(Y_i)_{i=1}^mX}$, consider the following region
\begin{align}
\mathcal{R} = \{(R_c,E)\mid &E\leq \min_{i\in [1:m]}[I(Y_i;U) + \bar{d}_i^{\star}],\nonumber\\
&R_c\geq I(X;U),\; U-X-(Y_i)_{i=1}^m\}.\label{region_R}
\end{align}
It can be seen that $\mathcal{R}$ only depends on marginal distributions $\{P_{Y_iX}\}_{i\in [1:m]}$ but not on the joint distribution $P_{(Y_i)_{i=1}^mX}$. Without the loss of generality we assume that in the evaluation of $\mathcal{R}$, $P_{(Y_i)_{i=1}^mX} = \bar{P}$. In the following we use the hyperplane characterization of $\mathcal{R}$. For that purpose, we first show the following result.
\begin{lemma}\label{lemma_1}
$\mathcal{R}$ is a closed, convex set. Furthermore, $\theta(R_c)$ is a concave function.
\end{lemma}
\begin{proof}
  Assume that $\{(R_{c,i},E_i)\}_{i=1}^2$ are two points in $\mathcal{R}$ with corresponding kernels $P_{U_i\vert  X}$. Assume that $\alpha$ is a random variable taking values on $\{1,2\}$ with $P_{\alpha}(1) = \nu$, $\nu\in [0,1]$, and independent of everything else. Define $U = (U_{\alpha},\alpha)$.
Then we have
\begin{align}
I(X;U) &= \nu I(X;U_1) + (1-\nu) I(X;U_2)\nonumber\\
&\leq \nu R_{c,1} + (1-\nu) R_{c,2}.\nonumber\\
I(Y_i;U) &= \nu I(Y_i;U_1) + (1-\nu) I(Y_i;U_2)\nonumber\\
& \geq \nu E_{1} + (1-\nu) E_2 - \bar{d}_i^{\star},\; i \in [1:m].
\end{align}
Therefore, the convex combination $(\nu R_{c,1}+ (1-\nu) R_{c,2}, \nu E_{1} + (1-\nu)E_2)$ lies also in $\mathcal{R}$, or $\mathcal{R}$ is a convex set. It can also be seen that $\mathcal{R}$ is a closed set since all alphabets are finite. Let $\mathcal{R}^{-} = \{(R_c,E)\mid (R_c,-E)\in\mathcal{R}\}$ be the reflection of the set $\mathcal{R}$ via the $R_c$-axis. Then $\mathcal{R}^{-}$ is a convex set. Furthermore the function 
\begin{align}
\theta^{-}(R_c) &= \min\{E\mid (R_c,E)\in\mathcal{R}^{-}\}\nonumber\\
& = \min_{P_{U\vert  X}\colon I(X;U)\leq R_c} -\min_{i\in [1:m]}[I(Y_i;U) + \bar{d}_i^{\star}],\nonumber\\
&=-\max_{P_{U\vert  X}\colon I(X;U)\leq R_c}\min_{i\in [1:m]}[I(Y_i;U) + \bar{d}_i^{\star}],
\end{align}
is a convex function. This implies that $\theta(R_c) = -\theta^{-}(R_c)$ is a concave function.
\end{proof}

\noindent For any $\mu>0$ define
\begin{align}
R_{\mathrm{HT}}^{\mu}(\bar{P}) = \min_{P_{U\vert  X}}(I(X;U)-\mu \min_{i\in[1:m]}[I(Y_i;U) + \bar{d}_i^{\star}]).\label{r_mu}
\end{align}
Since $\mathcal{R}$ is a closed convex set, the line $R_c - \mu E = R_{\mathrm{HT}}^{\mu}(\bar{P})$ supports $\mathcal{R}$ or\footnote{Assume that $(x,y)\in\mathbb{R}_{+}^2$ and $(x,y)\notin\mathcal{R}$ hold. Then there exists a pair $(a,b)\in\mathbb{R}^2$ such that $ax+by<aR+bE$ for all $(R,E)\in\mathcal{R}$ as $\mathcal{R}$ is a closed convex set. Since $(0,\min_{i}\bar{d}_i^{\star})\in\mathcal{R}$, we must have $y>\min_i \bar{d}_i^{\star}$. This implies that either $a$ or $b$ is negative by plugging $(R,E) = (0,\min_i\bar{d}_i^{\star})$ in. Similarly plugging $(R,\min_i\bar{d}_i^{\star})$ with sufficiently large $R$ in we see that $a$ is positive and $b$ is negative. Setting $\mu=-b/a$, we see that if $(x,y)\notin\mathcal{R}$ then $(x,y)\notin\cap_{\mu>0}\{(R_c,E)\mid R_c-\mu E\geq R_{\mathrm{HT}}^{\mu}(\bar{P})\}$. Hence $\mathcal{R}\supseteq \cap_{\mu>0}\{(R_c,E)\mid R_c-\mu E\geq R_{\mathrm{HT}}^{\mu}(\bar{P})\}$ holds. The other direction is straightforward.} $\mathcal{R} = \cap_{\mu>0}\{(R_c,E)\mid R_c-\mu E\geq R_{\mathrm{HT}}^{\mu}(\bar{P})\}$. For our proof we also use an additional characterization of $R_{\mathrm{HT}}^{\mu}(\bar{P})$ which is stated in the following. For any pair of positive numbers $(\mu,\alpha)$ we define 
\begin{align}
    R_{\mathrm{HT}}^{\mu,\alpha}(\bar{P}) =& \min_{P_{\tilde{U}\tilde{X}(\tilde{Y}_i)_{i=1}^m}}\bigg(I(\tilde{X},(\tilde{Y}_i)_{i=1}^m;\tilde{U}) -\mu \min_{i\in [1:m]}[I(\tilde{Y}_i;\tilde{U}) + \bar{d}_i^{\star}]\nonumber\\
    &\qquad +\alpha I(\tilde{U};(\tilde{Y}_i)_{i=1}^m\vert  \tilde{X}) + (\alpha+1) D(P_{\tilde{X}(\tilde{Y}_i)_{i=1}^m}\Vert \bar{P})\bigg),\label{r_mu_alpha}
\end{align}
where $P_{\tilde{U}\tilde{X}(\tilde{Y}_i)_{i=1}^m}$ is a joint probability measure on $\tilde{\mathcal{U}}\times\mathcal{X}\times\mathcal{Y}^m$ satisfying $P_{\tilde{X}(\tilde{Y}_i)_{i=1}^m}\ll \bar{P}$ and $(\tilde{U},\tilde{X},(\tilde{Y}_i)_{i=1}^m)\sim P_{\tilde{U}\tilde{X}(\tilde{Y}_i)_{i=1}^m}$. By the support lemma in \cite{csiszar2011information} we can upper bound the cardinality of $\vert  \tilde{\mathcal{U}}\vert  $ by a constant. An alternative characterization of $R_{\mathrm{HT}}^{\mu}(\bar{P})$ is given in the following.
\begin{lemma}\label{ht_reduced_reg}
\begin{align}
    \sup_{\alpha>0}R_{\mathrm{HT}}^{\mu,\alpha}(\bar{P}) = R_{\mathrm{HT}}^{\mu}(\bar{P}).
\end{align}
\end{lemma}
\noindent The proof of Lemma \ref{ht_reduced_reg} is deferred to the end of Subsection \ref{subsec_A}. 
\subsection{Strong converse proof for $\epsilon<\min\{\min_{s\in\mathcal{S}}\frac{1}{\vert  \mathfrak{F}_s\vert}, 1\}$}\label{subsec_A}
We now present the main part of the proof of Theorem \ref{thm_2}. In the proof, we will use a recent technique by Tyagi and Watanabe \cite {tyagi2019strong}. When the inactive set is empty, $\mathcal{S}=\varnothing$, showing that $E_{\mathrm{comp},\epsilon}^{\star}(R_c)\leq \min_{i\in[1:m]}\xi_i(R_c)$ for all $\epsilon\in [0,1)$ can be deduced from the strong converse result of testing $P_{Y_iX_i}^{\otimes n}$ against $Q_{Y_{j_i^{\star}}^n}\times Q_{X_{t^{\star}}^n}$ for all $i\in [1:m]$, cf. Theorem \ref{theorem_acs_extended}. Without loss of generality, we assume in the following that the inactive set is not empty $\mathcal{S}\neq \varnothing$. 

Assume first that the set of marginal distributions on $\mathcal{X}$, $\mathcal{P}_{\mathcal{X}}$, is a singleton. This implies that $\vert  \mathfrak{F}_1\vert   = m$ holds. 
Given a sequence of testing schemes $(\phi_n,\psi_n)$ such that the $\epsilon$-achievability conditions are fulfilled
\begin{align}
  \limsup_{n\to\infty}\frac{1}{n}\log\vert  \phi_n\vert  &\leq R_c,\;\limsup_{n\to\infty}\alpha_n\leq \epsilon,\nonumber\\
  \;\liminf_{n\to\infty}&\frac{1}{n}\log\frac{1}{\beta_n}\geq E,
\end{align}
we define for each $i\in [1:m]$ the following likelihood based decision region
\begin{align}
  \mathcal{A}_{n,\gamma}^{i} = \{(y^n,\phi_n(x^n))\mid & P_{Y_i^n\phi_n(X_i^n)}(y^n,\phi_n(x^n))\nonumber\\
  &\geq e^{n(E-\gamma)}Q_{Y_{j_i^{\star}}^n}\times Q_{\phi_n(X_{t^{\star}}^n)}(y^n,\phi_n(x^n))\}.\label{decision_i}
\end{align}
By \cite[Lemma 4.1.2]{hanspectrum}, cf. also \cite[Lemma 12.2]{polyanskiy2014lecture}, we obtain
\begin{align}
  \alpha_n+e^{n(E-\gamma)}\beta_n&\geq P_{Y_i^n\phi_n(X_i^n)}[(\mathcal{A}_{n,\gamma}^{i})^c],\nonumber\\
  Q_{Y_{j_i^{\star}}^n}\times Q_{\phi_n(X_{t^{\star}}^n)}(\mathcal{A}^i_{n,\gamma})&\leq e^{-n(E-\gamma)}. 
\end{align}
For each $n$, $(\phi_n,\mathbf{1}_{\mathcal{A}^i_{n,\gamma}})$ can be seen as a testing scheme for differentiating between $P_{Y_iX_i}^{\otimes n}$ and $Q_{Y_{j_i^{\star}}^n}\times Q_{X_{t^{\star}}^n}$. For a given $i\in [1:m]$ we now a construct a testing scheme $(\bar{\phi}_n,\bar{\psi}_{in})$ to differentiate between $P_{Y_iX_i}^{\otimes n}$ and $P_{Y_i}^{\otimes n}\times P_{X}^{\otimes n}$.  Given $\phi_n$, the compression mapping $\bar{\phi}_n$ does not depend on $i$. Our arguments are similar to the code transformations given in Appendix \ref{ap_common_seq}. We present the procedure in the following for completeness. Given $\phi_n$,  $\bar{\phi}_n$ is defined as
\begin{align}
  \bar{\phi}_n\colon\mathcal{X}^n&\to\mathcal{M}\cup\{e\}\nonumber\\
  \bar{\phi}_n(x^n)&\mapsto\begin{dcases}\phi_n(x^n),\;&\text{if}\; x^n\in\mathcal{B}_{n,\gamma},\nonumber\\
  e\;&\text{otherwise}\end{dcases}.
\end{align}
For each $i\in\mathfrak{F}_1$, the decision mapping $\bar{\psi}_{in}$ is defined as
\begin{align}
  \bar{\psi}_{in}\colon \mathcal{Y}^n\times(\mathcal{M}\cup\{e\})&\to \{0,1\}\nonumber\\
  \bar{\psi}_{in}(y^n,\bar{u})&\mapsto\begin{dcases}\mathbf{1}_{\mathcal{A}^i_{n,\gamma}}(y^n,\bar{u}),\;&\text{if}\; y^n\in \mathcal{B}_{n,\gamma}^{(i)},\;\text{and}\; \bar{u}\neq e,\nonumber\\
  1&\text{otherwise}\end{dcases}.
\end{align}
In the above definitions the typical sets $\mathcal{B}_{n,\gamma}$ (or $\mathcal{B}_{n,\gamma}^s$) and $\mathcal{B}_{n,\gamma}^{(i)}$ are defined in \eqref{bigB_def}. Then we also have
\begin{align}
  P_{Y_i^n\bar{\phi}_n(X_i^n)}(1-\bar{\psi}_{in})&\leq P_{Y_i^n\phi_n(X_i^n)}[(\mathcal{A}_{n,\gamma}^{i})^c] + P_{X}^{\otimes n}[(\mathcal{B}_{n,\gamma})^c] + P_{Y_i}^{\otimes n}[(\mathcal{B}_{n,\gamma}^{(i)})^c]\nonumber\\
  P_{Y_i^n}\times P_{\bar{\phi}_n(X_i^n)}(\bar{\psi}_{in})&\leq e^{n(d_i^{\mathrm{y}}+ d_s^{\mathrm{x}} +2\gamma)}Q_{Y_{j_i^{\star}}^n}\times Q_{\phi_n(X_{t^{\star}}^n)}(\mathcal{A}^i_{n,\gamma})\nonumber\\
  &\leq e^{-n(E-\bar{d}_i^{\star}-3\gamma)}.\label{first_s}
\end{align}
Let $\mathcal{A}_{n}^{(i)}$ be the corresponding acceptance region of $\bar{\psi}_{in}$. We take $n_0$ to be sufficiently large such that the following conditions hold for all $n\geq n_0$
\begin{align}
\log\vert  \bar{\phi}_n\vert &\leq n(R_c+\gamma),\nonumber\\
P_{Y_i^n\bar{\phi}_n(X_i^n)}(\mathcal{A}_{n}^{(i)})&\geq 1-(\epsilon+\gamma),\;\forall i\in [1:m].
\end{align}
Next, we can further decompose $\mathcal{A}_n^{(i)}$ as
\begin{equation}
  \mathcal{A}_n^{(i)}=\bigcup_{u\in \mathcal{M}}\mathcal{A}_{n,u}^{(i)}\times \{u\}.\label{decompose_Ai}
\end{equation}
For simplicity define $\delta = 1-(\epsilon+\gamma)$. Consider the set
\begin{align}
\mathcal{V}_i = \{x^n\mid P_{Y_i\vert  X_i}^{\otimes n}(\mathcal{A}_{n,\bar{\phi}_n(x^n)}^{(i)}\vert  x^n)>\eta\},
\end{align}
where $\delta>\eta>0$. Then we have
\begin{align}
                   \delta&\leq P_{Y_i^n\phi_n(X_i^n)}(\mathcal{A}_n^{(i)})\nonumber\\
                   & = \mathrm{Pr}[(Y_i^n,\phi_n(X_i^n))\in \mathcal{A}_n^{(i)}, X_i^n \in \mathcal{V}_i] + \mathrm{Pr}[(Y_i^n,\phi_n(X_i^n))\in \mathcal{A}_n^{(i)}, X_i^n \notin \mathcal{V}_i]\nonumber\\
                   &\leq P_{X}^{\otimes n}(\mathcal{V}_i) + \eta P_{X}^{\otimes n}(\mathcal{V}_i^c),
\end{align}
which implies that $P_{X}^{\otimes n}(\mathcal{V}_i)\geq (\delta-\eta)/(1-\eta)$. Using this inequality we further obtain 
\begin{align}
P_{X}^{\otimes n}(\cap_{i=1}^m\mathcal{V}_i)&=1-P_{X}^{\otimes n}(\cup_{i=1}^m\mathcal{V}_i^c)\nonumber\\
& \geq 1-m(1-\delta)/(1-\eta).
\end{align}
For $P_{X}^{\otimes n}(\cap_{i=1}^m\mathcal{V}_i)>0$, we require that
\begin{equation}
  \eta< 1-m(1-\delta) = 1-m(\epsilon+\gamma),
\end{equation}
must hold. This in turn implies that $\epsilon<1/m = 1/\vert  \mathfrak{F}_1\vert  $. Taking $\eta = (1-m(\epsilon+\gamma))/2$, we obtain $P_{X}^{\otimes n}(\cap_{i=1}^m\mathcal{V}_i)>(1-m(\epsilon+\gamma))/(1+m(\epsilon+\gamma))$.
Define $\tilde{\mathcal{V}}_n = \cap_{i=1}^m \mathcal{V}_i$, $\tilde{\epsilon} =(1-m(\epsilon+\gamma))/(1+m(\epsilon+\gamma))$, and the following distribution on $\mathcal{X}^n$
\begin{align}
\tilde{P}_{\mathcal{X}^n}(x^n) = P_{X}^{\otimes n}(x^n)/P_{X}^{\otimes n}(\tilde{\mathcal{V}}_n)\mathbf{1}\{x^n\in\tilde{\mathcal{V}}_n\}.
\end{align}
For a given $x^n\in \tilde{\mathcal{V}}_n$ we also define the following joint conditional distribution 
on $\mathcal{Y}^{nm}$
\begin{align}
\tilde{P}_{\mathcal{Y}^{nm},x^n}((y_i^n)_{i=1}^m\vert  x^n) = \prod_{i=1}^m\frac{P_{Y_i\vert  X_i}^{\otimes n}(y^n\vert  x^n)}{P_{Y_i\vert  X_i}^{\otimes n}(\mathcal{A}_{n,\bar{\phi}_n(x^n)}^{(i)}\vert  x^n)}\mathbf{1}\{y^n\in\mathcal{A}_{n,\bar{\phi}_n(x^n)}^{(i)}\}.
\end{align}
Additionally, we define $\tilde{P}_{\mathcal{Y}^{nm},x^n}((y_i^n)_{i=1}^m\vert  x^n)=0$ for all $y^n$ if $x^n\notin \tilde{\mathcal{V}}_n$.
Let $(\tilde{X}^n,(\tilde{Y}_i^n)_{i=1}^m)$ be a tuple of general sources such that 
\begin{equation}
(\tilde{X}^n,(\tilde{Y}_i^n)_{i=1}^m) \sim \tilde{P}\triangleq\tilde{P}_{\mathcal{Y}^{nm},x^n}\times \tilde{P}_{\mathcal{X}^n}.
\end{equation}
Then for each $x^n\in \tilde{\mathcal{V}}_n$ we have the following inequality
\begin{align}
  P_{X}^{\otimes n}(x^n) = P_{\tilde{X}^n}(x^n)P_{X}^{\otimes n}(\tilde{\mathcal{V}}_n)\geq \tilde{\epsilon}P_{\tilde{X}^n}(x^n).\label{x_side}
\end{align}
 \begin{figure}[htb]
   \centering
   \includegraphics{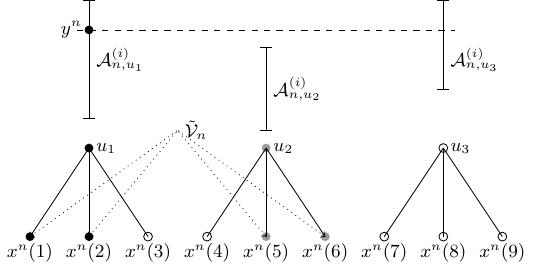}
   \caption{Different situations for $y^n\in \bigcup_{u\colon\bar{\phi}_n^{-1}(u)\cap\tilde{\mathcal{V}}_n\neq \varnothing}\mathcal{A}_{n,u}^{(i)}$.}
   \label{partition_illu}
\end{figure}

\noindent Consider an arbitrarily fixed sequence $y^n\in \bigcup_{u\colon\bar{\phi}_n^{-1}(u)\cap\tilde{\mathcal{V}}_n\neq \varnothing}\mathcal{A}_{n,u}^{(i)}$. For the following analysis, the illustration given in Fig. \ref{partition_illu} is perhaps helpful for readers.
Let $u$ be a message index such that $\bar{\phi}_n^{-1}(u)\cap\tilde{\mathcal{V}}_n\neq \varnothing$ and $y^n\in \mathcal{A}_{n,u}^{(i)}$ hold. For all $x^n\in\tilde{\mathcal{V}}_n\cap \bar{\phi}_n^{-1}(u)$, we then have
\begin{align}
  P_{Y_iX_i}^{\otimes n}(y^n,x^n) &= P_{Y_i\vert  X_i}^{\otimes n}(\mathcal{A}_{n,\bar{\phi}_n(x^n)}^{(i)}\vert  x^n)P_{X}^{\otimes n}(\tilde{\mathcal{V}}_n)P_{\tilde{Y}_i^n\tilde{X}^n}(y^n,x^n)\nonumber\\
                                &\geq \eta\tilde{\epsilon}P_{\tilde{Y}_i^n\tilde{X}^n}(y^n,x^n).\label{complicated_yn}
\end{align}
If $u$ is another message index such that $\bar{\phi}_n^{-1}(u)\cap\tilde{\mathcal{V}}_n\neq \varnothing$ and that $y^n\notin \mathcal{A}_{n,u}^{(i)}$ hold, then for all $x^n\in\tilde{\mathcal{V}}\cap\bar{\phi}_n^{-1}(u)$, we have $P_{\tilde{Y}_i^n\tilde{X}^n}(y^n,x^n)=0$.
When $x^n\notin\tilde{\mathcal{V}}_n$, we also have $P_{\tilde{Y}_i^n\tilde{X}^n}(y^n,x^n)=0$. Hence for all $y^n\in \bigcup_{u\colon\bar{\phi}_n^{-1}(u)\cap\tilde{\mathcal{V}}_n\neq \varnothing}\mathcal{A}_{n,u}^{(i)}$, we obtain
\begin{align}
&  P_{Y_i}^{\otimes n}(y^n)=\sum_{x^n\in\mathcal{X}^n}P_{Y_iX_i}^{\otimes n}(y^n,x^n)\nonumber\\
  &=\sum_{x^n\notin \tilde{\mathcal{V}}_n}P_{Y_iX_i}^{\otimes n}(y^n,x^n) + \sum_{\substack{x^n\in \tilde{\mathcal{V}}_n\\y^n\in \mathcal{A}_{n,\bar{\phi}_n(x^n)}^{(i)}}}P_{Y_iX_i}^{\otimes n}(y^n,x^n) \nonumber\\
  &\hspace{1cm}+ \sum_{\substack{x^n\in \tilde{\mathcal{V}}_n\\y^n\notin \mathcal{A}_{n,\bar{\phi}_n(x^n)}^{(i)}}} P_{Y_iX_i}^{\otimes n}(y^n,x^n)\nonumber\\
                          &\geq \eta\tilde{\epsilon}\big[\sum_{x^n\notin \tilde{\mathcal{V}}_n}P_{\tilde{Y}_i^n\tilde{X}^n}(y^n,x^n) + \sum_{\substack{x^n\in \tilde{\mathcal{V}}_n\\y^n\in \mathcal{A}_{n,\bar{\phi}_n(x^n)}^{(i)}}}P_{\tilde{Y}_i^n\tilde{X}^n}(y^n,x^n)\nonumber\\
                          &\hspace{1cm} + \sum_{\substack{x^n\in \tilde{\mathcal{V}}_n\\y^n\notin \mathcal{A}_{n,\bar{\phi}_n(x^n)}^{(i)}}} P_{\tilde{Y}_i^n\tilde{X}^n}(y^n,x^n)\big]\nonumber\\
  &= \eta\tilde{\epsilon}P_{\tilde{Y}_i^n}(y^n),  \label{y_side}
\end{align}
We observe that not only $\tilde{P}\ll \bar{P}^{\otimes n}$ holds but also we have the following bound
\begin{align}
& D(\tilde{P}\Vert \bar{P}^{\otimes n})\nonumber\\
 &= \sum P_{\tilde{X}^n}(x^n)\log\frac{P_{\tilde{X}^n}(x^n)}{P_{X}^{\otimes n}}\nonumber\\
& + \sum P_{\tilde{X}^n}(x^n)\sum P_{(\tilde{Y}_i^n)_{i=1}^m\vert  \tilde{X}^n}((y_i^n)_{i=1}^m\vert  x^n)\log\frac{P_{(\tilde{Y}_i^n)_{i=1}^m\vert  \tilde{X}^n}((y_i^n)_{i=1}^m\vert  x^n)}{\prod_{i=1}^m P_{Y_i\vert  X_i}^{\otimes n}(y_i^n\vert  x^n)}\nonumber\\
&\leq \log\frac{1}{\tilde{\epsilon}} + m\log\frac{1}{\eta}\triangleq \tilde{\eta}.\label{divergence_bound}
\end{align}
We now derive bounds on the compression rate $R_c$ and the $\epsilon$-achievable error exponent $E$ using expressions involving $(\tilde{X}^n,(\tilde{Y}_i^n)_{i=1}^m)$. We first have
\begin{align}
    n(R_c+\gamma)\geq \log\vert  \bar{\phi}_n\vert  \geq I(\tilde{X}^n;\tilde{M}),\label{rate_ineq}
\end{align}
where $\tilde{M} =\bar{\phi}_n(\tilde{X}^n)$. Note further that the following Markov chain holds
\begin{equation}
(\tilde{Y}_i^n)_{i=1}^m - \tilde{X}^n - \tilde{M}.\label{markov_thm2}
\end{equation}
For each $i\in [1:m]$ the support set of the joint distribution $P_{\tilde{Y}_i^n\bar{\phi}_n(\tilde{X}^n)}$ is a subset of the following set, cf. Fig. \ref{partition_illu} for a visual illustration, $$\tilde{\mathcal{A}}_n^{(i)} = \bigcup_{u\colon\bar{\phi}_n^{-1}(u)\cap\tilde{\mathcal{V}}_n\neq\varnothing}\mathcal{A}_{n,u}^{(i)}\times \{u\}.$$
\noindent Then compared with \eqref{decompose_Ai} we have $\tilde{\mathcal{A}}_n^{(i)}\subseteq \mathcal{A}_n^{(i)}$ due to the restriction on $u$. Therefore, on one hand we have
\begin{align}
  P_{Y_i^n}\times P_{\bar{\phi}_n(X_i^n)}(\tilde{\mathcal{A}}_n^{(i)})\leq P_{Y_i^n}\times P_{\bar{\phi}_n(X_i^n)}(\mathcal{A}_n^{(i)}).\label{second_s}
\end{align}
On the other hand from \eqref{x_side} and \eqref{y_side} we have
\begin{align}
&  P_{Y_i^n}\times P_{\bar{\phi}_n(X_i^n)}(\tilde{\mathcal{A}}_n^{(i)}) = \sum_{u\colon\bar{\phi}_n^{-1}(u)\cap\tilde{\mathcal{V}}_n\neq\varnothing}P_{\bar{\phi}_n(X_i^n)}(u)\sum_{y^n\in\mathcal{A}_{n,u}^{(i)}}P_{Y_i}^{\otimes n}(y^n)\nonumber\\
                                                               &\stackrel{(*)}{\geq} \sum_{u\colon\bar{\phi}_n^{-1}(u)\cap\tilde{\mathcal{V}}_n\neq\varnothing}\sum_{x^n\in\bar{\phi}_n^{-1}(u)\cap\tilde{\mathcal{V}}_n}P_{X}^{\otimes n}(x^n)\sum_{y^n\in\mathcal{A}_{n,u}^{(i)}}P_{Y_i}^{\otimes n}(y^n)\nonumber\\
                                                               &\geq  \sum_{u\colon\bar{\phi}_n^{-1}(u)\cap\tilde{\mathcal{V}}_n\neq\varnothing}\sum_{x^n\in\bar{\phi}_n^{-1}(u)\cap\tilde{\mathcal{V}}_n} \tilde{\epsilon}P_{\tilde{X}^n}(x^n)\sum_{y^n\in\mathcal{A}_{n,u}^{(i)}}\tilde{\epsilon}\eta P_{\tilde{Y}_i^n}(y^n)\nonumber\\
  & = \tilde{\epsilon}^2\eta\sum_{u\colon\bar{\phi}_n^{-1}(u)\cap\tilde{\mathcal{V}}_n\neq\varnothing}\sum_{y^n\in\mathcal{A}_{n,u}^{(i)}}P_{\bar{\phi}_n(\tilde{X}^n)}(u)P_{\tilde{Y}_i^n}(y^n).
\end{align}
where $(*)$ holds because $P_{\bar{\phi}_n(X_i^n)}(u) = \sum_{x^n\in\bar{\phi}_n^{-1}(u)}P_{X}^{\otimes n}(x^n)$. The last equality is valid because for $x^n\in\bar{\phi}_n^{-1}(u)\cap(\tilde{\mathcal{V}}_n)^c$ we have $P_{\tilde{X}^n}(x^n)=0$. This leads to the following
\begin{align}
&\log\frac{1}{P_{Y_i^n}\times P_{\bar{\phi}_n(X_i^n)}(\tilde{A}_n^{(i)})}\nonumber\\
&\stackrel{(**)}{\leq} \sum_{u,y^n}P_{\bar{Y}_i^n\bar{\phi}_n(\bar{X}^n)}(y^n,u)\log\frac{P_{\bar{Y}_i^n\bar{\phi}_n(\bar{X}^n)}(y^n,u)}{\tilde{\epsilon}^2\eta P_{\phi_n(\tilde{X}^n)}(u)P_{\tilde{Y}_i^n}(y^n)}\nonumber\\
&\leq I(\tilde{Y}_i^n;\bar{\phi}_n(\tilde{X}^n)) + \log\frac{1}{\eta\tilde{\epsilon}^2},\label{third_s}
\end{align}
where we have use log-sum inequality in $(**)$ and the summation therein is taken over $u\colon\bar{\phi}_n^{-1}(u)\cap\tilde{\mathcal{V}}_n\neq\varnothing$ and $y^n\in\mathcal{A}_{n,u}^{(i)}$. The last inequality holds since for $(y^n,u) \notin \tilde{\mathcal{A}}_{n}^{(i)}$ we have $P_{\tilde{Y}_i^n\bar{\phi}_n(\tilde{X}^n)}(y^n,u) = 0$. Combining \eqref{first_s}, \eqref{second_s}, and \eqref{third_s} we obtain
\begin{align}
n(E-3\gamma)\leq \min_{i\in[1:m]}[I(\tilde{Y}_i^n;\tilde{M}) + n\bar{d}_i^{\star}] + \log\frac{1}{\eta\tilde{\epsilon}^2}.\label{exponent_ineq}
\end{align}
By combining those two previous inequalities in \eqref{rate_ineq} and \eqref{exponent_ineq}, we obtain for arbitrarily given positive numbers $\mu$ and $\alpha$
\begin{align}
&    n(R_c+\gamma - \mu (E-3\gamma))\nonumber\\
    &\geq I(\tilde{X}^n;\tilde{M})-\mu \min_{i\in[1:m]}[I(\tilde{Y}_i^n;\tilde{M}) + n\bar{d}_i^{\star}]-\mu\log\frac{1}{\eta\tilde{\epsilon}^2}\nonumber\\
    &\stackrel{\eqref{divergence_bound},\eqref{markov_thm2}}{\geq} I(\tilde{X}^n;\tilde{M})-\mu \min_{i\in[1:m]}[I(\tilde{Y}_i^n;\tilde{M}) + n\bar{d}_i^{\star}]\nonumber\\
    & + (\alpha+1)I(\tilde{M};(\tilde{Y}_i^n)_{i=1}^m\vert  \tilde{X}^n) + (\alpha+1) D(\tilde{P}\Vert \bar{P}^{\otimes n}) - (\alpha+1)\tilde{\eta}-\mu\log\frac{1}{\eta\tilde{\epsilon}^2}\nonumber\\
    & = I(\tilde{M};\tilde{X}^n,(\tilde{Y}_i^n)_{i=1}^m) - \mu \min_{i\in [1:m]}[I(\tilde{Y}_i^n;\tilde{M}) + n\bar{d}_i^{\star}] + \alpha I(\tilde{M};(\tilde{Y}_i^n)_{i=1}^m\vert  \tilde{X}^n)\nonumber\\
    & + (\alpha+1)D(\tilde{P}\Vert \bar{P}^{\otimes n}) - (\alpha+1)\tilde{\eta}-\mu\log\frac{1}{\eta\tilde{\epsilon}^2}\nonumber\\
    & = A_1+ A_2- (\alpha+1)\tilde{\eta}-\mu\log\frac{1}{\eta\tilde{\epsilon}^2},\label{eq_202}
\end{align}
where
\begin{align}
    A_1 &= H(\tilde{X}^n,(\tilde{Y}_i^n)_{i=1}^m) + \alpha H((\tilde{Y}_i^n)_{i=1}^m\vert  \tilde{X}^n) + (\alpha+1)D(\tilde{P}\Vert \bar{P}^{\otimes n}),\nonumber\\
    A_2 &= -H(\tilde{X}^n,(\tilde{Y}_i^n)_{i=1}^m\vert  \tilde{M}) - \alpha H((\tilde{Y}_i^n)_{i=1}^m\vert  \tilde{X}^n,\tilde{M})\nonumber\\
    &\hspace{2cm} - \mu \min_{i\in [1:m]}[I(\tilde{Y}_i^n;\tilde{M}) + n\bar{d}_i^{\star}].
\end{align}
Let $T$ be a uniform random on $[1:n]$. Furthermore, define $\tilde{U}_l = (\tilde{M},(\tilde{Y}_i^{l-1})_{i=1}^m)$ for all $l=[1:n]$, and $\tilde{U} = (\tilde{U}_T,T)$. We show at the end of this subsection that
\begin{align}
A_1&\geq n(H(\tilde{X}_T,(\tilde{Y}_{iT})_{i=1}^m) +\alpha H((\tilde{Y}_{iT})_{i=1}^m\vert  \tilde{X}_T)\nonumber\\
&\hspace{1cm}+ (\alpha+1)D(P_{\tilde{X}_T(\tilde{Y}_{iT})_{i=1}^m}\Vert \bar{P})),\label{a1_ineq}\\
A_2&\geq n(-H(\tilde{X}_T,(\tilde{Y}_{iT})_{i=1}^m\vert  \tilde{U}) - \alpha H((\tilde{Y}_{iT})_{i=1}^m\vert  \tilde{X}_T,\tilde{U})\nonumber\\
&\hspace{1cm}-\mu \min_{i\in [1:m]}[I(\tilde{Y}_{iT};\tilde{U}) + \bar{d}_i^{\star}])\label{a2_ineq}.
\end{align}
In summary we obtain for all positive $\mu$ and $\alpha$ that
\begin{align}
&    (R_c+\gamma) - \mu (E-3\gamma)\nonumber\\
    &\geq I(\tilde{X}_T,(\tilde{Y}_{iT})_{i=1}^m;\tilde{U}) -\mu \min_{i\in[1:m]}[I(\tilde{Y}_{iT};\tilde{U}) + \bar{d}_i^{\star}] +\alpha I(\tilde{U};(\tilde{Y}_{iT})_{i=1}^m\vert  \tilde{X}_T)\nonumber\\
    &\qquad + (\alpha+1)D(P_{\tilde{X}_T(\tilde{Y}_{iT})_{i=1}^m}\Vert \bar{P})-\frac{ (\alpha+1)\tilde{\eta}+\mu\log\frac{1}{\eta\tilde{\epsilon}^2}}{n}\nonumber\\
    &\geq R_{\text{HT}}^{\mu,\alpha}(\bar{P}) -\frac{ (\alpha+1)\tilde{\eta}+\mu\log\frac{1}{\eta\tilde{\epsilon}^2}}{n}. \label{eq_206}
\end{align}
Taking $n\to\infty$ and supremum over $\alpha>0$ and using Lemma \ref{ht_reduced_reg}, we have shown that $(R_c+\gamma,E-3\gamma)\in\mathcal{R}$. Taking $\gamma\to 0$ we obtain the conclusion that $E\leq \theta(R_c)$.\\
If the set of marginal distributions on $\mathcal{X}$, $\mathcal{P}_{\mathcal{X}}$, has more than one element then for each inactive $s\in \mathcal{S}$ we have $E\leq \theta_s(R_c)$ provided that $\epsilon<1/\vert  \mathfrak{F}_s\vert  $ holds. For each active $s\notin\mathcal{S}$, then by the strong converse result for the simple hypothesis testing problem we obtain $E\leq\min_{i\in\mathfrak{F}_s}\xi_i(R_c) = \theta_s(R_c)$, for all $\epsilon\in [0,1)$. Thereby we obtain the conclusion.

\noindent \textit{Proof of Lemma \ref{ht_reduced_reg}}: 
Select for any $P_{U\vert  X}$ in the optimization domain of $R_{\text{HT}}^{\mu}(\bar{P})$ a $P_{\tilde{U}\tilde{X}(\tilde{Y}_i)_{i=1}^m}= P_{U\vert  X}\times \bar{P}$. Then we can see that $$\sup_{\alpha>0}R_{\text{HT}}^{\mu,\alpha}(\bar{P})\leq R_{\text{HT}}^{\mu}(\bar{P}).$$
Given an $\alpha>0$, let $P_{\tilde{U}\tilde{X}(\tilde{Y}_i)_{i=1}^m}^{\alpha}$ be an optimal solution for $R_{\text{HT}}^{\mu,\alpha}(\bar{P})$. As $R_{\text{HT}}^{\mu}(\bar{P})\leq \log\vert  \mathcal{X}\vert  $, and $I(\tilde{X};\tilde{U})-\mu \min_{i\in[1:m]}[I(\tilde{Y}_i;\tilde{U}) + \bar{d}_i^{\star}]\geq -\mu(\log\vert  \mathcal{Y}\vert   + \min_{i}\bar{d}_i^{\star})$, we have
\begin{align}
    D(P_{\tilde{U}\tilde{X}(\tilde{Y}_i)_{i=1}^m}^{\alpha}\Vert P_{\tilde{U}\vert  \tilde{X}}^{\alpha}\times \bar{P}) &=  I(\tilde{U}; (\tilde{Y}_i)_{i=1}^m\vert  \tilde{X})+ D(P_{\tilde{X}(\tilde{Y}_i)_{i=1}^m}^{\alpha}\Vert \bar{P})\nonumber\\
    &\leq \frac{\log\vert  \mathcal{X}\vert  +\mu(\log\vert  \mathcal{Y}\vert   + \min_{i\in [1:m]}\bar{d}_i^{\star})}{\alpha}.
\end{align}
We then have 
\begin{align}
&    R_{\text{HT}}^{\mu,\alpha}(\bar{P})\nonumber\\
    &\geq I(\tilde{X};\tilde{U}) - \mu \min_{i\in [1:m]}[I(\tilde{Y}_i;\tilde{U}) + \bar{d}_i^{\star}] &&\text{eval. with}\; P_{\tilde{U}\tilde{X}(\tilde{Y}_i)_{i=1}^m}^{\alpha}\nonumber\\
    &\geq I(X;\bar{U}) - \mu \min_{i\in [1:m]}[I(Y_i;\bar{U}) + \bar{d}_i^{\star}]\nonumber\\
    & - \Delta\bigg(\frac{\log\vert  \mathcal{X}\vert  +\mu(\log\vert  \mathcal{Y}\vert  +\min_{i\in [1:m]}\bar{d}_i^{\star})}{\alpha}\bigg) &&\text{eval. with}\; P_{\tilde{U}\vert  \tilde{X}}^{\alpha}\times \bar{P}\nonumber\\
    &\geq  R_{\text{HT}}^{\mu}(\bar{P})- \Delta\bigg(\frac{\log\vert  \mathcal{X}\vert  +\mu(\log\vert  \mathcal{Y}\vert  +\min_{i\in [1:m]}\bar{d}_i^{\star})}{\alpha}\bigg),
\end{align}
where $\Delta(t)\to 0$ as $t\to 0$. Taking the supremum over $\alpha$ we obtain the conclusion. 

\noindent\textit{Proof of \eqref{a1_ineq} and \eqref{a2_ineq}}: First, we have
\begin{align}
    A_1 =& H(\tilde{X}^n,(\tilde{Y}_i^n)_{i=1}^m)+D(\tilde{P}\Vert \bar{P}^{\otimes n}) \nonumber\\
    &+ \alpha [H((\tilde{Y}_i^n)_{i=1}^m\vert  \tilde{X}^n) + D(\tilde{P}\Vert \bar{P}^{\otimes n})].
\end{align}
Note that since $\tilde{P}\ll \bar{P}^{\otimes n}$, whenever $\bar{P}(x,(y_i)_{i=1}^m)=0$ we must have $P_{\tilde{X}_l(\tilde{Y}_{il})_{i=1}^m}(x,(y_i)_{i=1}^m)=0$ for all $l\in [1:n]$. This implies that $P_{\tilde{X}_T(\tilde{Y}_{iT})_{i=1}^m} = 1/n\sum P_{\tilde{X}_l(\tilde{Y}_{il})_{i=1}^m} \ll \bar{P}$. The absolute continuities ensure the validity of the following derivations
\begin{align}
&    H(\tilde{X}^n,(\tilde{Y}_i^n)_{i=1}^m)+D(\tilde{P}\Vert \bar{P}^{\otimes n})\nonumber\\
     &= \sum P_{\tilde{X}^n(\tilde{Y}_i^n)_{i=1}^m}(x^n,(y_i^n)_{i=1}^m)\log\frac{1}{\bar{P}^{\otimes n}(x^n,(y_i^n)_{i=1}^m)}\nonumber\\
    &=\sum_{l=1}^{n}\sum P_{\tilde{X}_l(\tilde{Y}_{il})_{i=1}^m}(x,(y_i)_{i=1}^m)\log\frac{1}{\bar{P}(x,(y_i)_{i=1}^m)}\nonumber\\
    &=\sum P_{\tilde{X}_T(\tilde{Y}_{iT})_{i=1}^m}(x,(y_i)_{i=1}^m)\log\frac{1}{\bar{P}(x,(y_i)_{i=1}^m)}\nonumber\\
    &= n(H(\tilde{X}_T,(\tilde{Y}_{iT})_{i=1}^m) + D(P_{\tilde{X}_T(\tilde{Y}_{iT})_{i=1}^m}\Vert \bar{P})), 
\end{align}
and
\begin{align}
&    H((\tilde{Y}_i^n)_{i=1}^m\vert  \tilde{X}^n) + D(\tilde{P}\Vert \bar{P}^{\otimes n}) = -H(\tilde{X}^n) + H(\tilde{X}^n,(\tilde{Y}_i^n)_{i=1}^m)+D(\tilde{P}\Vert \bar{P}^{\otimes n})\nonumber\\
    & = -H(\tilde{X}^n) + n H(\tilde{X}_T) +n H((\tilde{Y}_{iT})_{i=1}^m\vert  \tilde{X}_T) + n D(P_{\tilde{X}_T(\tilde{Y}_{iT})_{i=1}^m}\Vert \bar{P})\nonumber\\
    &\geq n H((\tilde{Y}_{iT})_{i=1}^m\vert  \tilde{X}_T) + nD(P_{\tilde{X}_T(\tilde{Y}_{iT})_{i=1}^m}\Vert \bar{P}).
\end{align}
In the last step we have use the inequality $H(\tilde{X}_T)\geq H(\tilde{X}_T\vert  T)$, as in this case $\tilde{X}_T$ and $T$ might not be independent of each other. This implies \eqref{a1_ineq}. Furthermore, since conditioning reduces entropy, we can lower bound the term $A_2$ as follows
\begin{align}
    A_2 &\geq \sum_{l=1}^n -H(\tilde{X}_l,(\tilde{Y}_{il})_{i=1}^m\vert  \tilde{M},\tilde{X}^{l-1},(\tilde{Y}_i^{l-1})_{i=1}^m)\nonumber\\
    & -\alpha H((\tilde{Y}_{il})_{i=1}^m\vert  \tilde{M},(\tilde{Y}_i^{l-1})_{i=1}^m,\tilde{X}^n)\nonumber\\
    &\qquad - \mu \min_{i\in [1:m]}\bigg[\sum_{l=1}^n I(\tilde{Y}_{il};\tilde{M},\tilde{Y}_i^{l-1}) + n\bar{d}_i^{\star}\bigg]\nonumber\\
    &\geq \sum_{l=1}^n -H(\tilde{X}_l,(\tilde{Y}_{il})_{i=1}^m\vert  \tilde{M},(\tilde{Y}_i^{l-1})_{i=1}^m) - \alpha H((\tilde{Y}_{il})_{i=1}^m\vert  \tilde{X}_l,\tilde{M},(\tilde{Y}_i^{l-1})_{i=1}^m)\nonumber\\
     &\qquad - \mu \min_{i\in [1:m]}\bigg[\sum_{l=1}^n I(\tilde{Y}_{il};\tilde{M},(\tilde{Y}_{\eta}^{l-1})_{\eta=1}^m) + n\bar{d}_i^{\star}\bigg]\nonumber\\
    & =  \sum_{l=1}^n -H(\tilde{X}_l,(\tilde{Y}_{il})_{i=1}^m\vert  \tilde{U}_l) - \alpha H((\tilde{Y}_{il})_{i=1}^m\vert  \tilde{X}_l,\tilde{U}_l) \nonumber\\
    &\hspace{1cm}- \mu \bigg[
    \min_{i\in [1:m]}\sum_{l=1}^nI(\tilde{Y}_{il};\tilde{U}_l) + n\bar{d}_i^{\star}\bigg]\nonumber\\
    & \geq n(-H(\tilde{X}_T,(\tilde{Y}_{iT})_{i=1}^m\vert  \tilde{U}_T,T) - \alpha H((\tilde{Y}_{iT})_{i=1}^m\vert  \tilde{X}_T,\tilde{U}_T,T)\nonumber\\
    &\hspace{1cm}-\mu \min_{i\in [1:m]}[I(\tilde{Y}_{iT};\tilde{U}_T,T) + \bar{d}_i^{\star}]).
\end{align}
\subsection{The case that $\epsilon>\max\{\max_{s\in\mathcal{S}}\frac{\vert  \mathfrak{F}_s\vert  -1}{\vert  \mathfrak{F}_s\vert },0\}$}
Fix a compression rate $R_c$ and an arbitrary $\gamma>0$. If the inactive set is empty, $\mathcal{S} = \varnothing$, then we have $\min_s\theta_s(R_c) = \min_{i\in [1:m]}\xi_i(R_c)$ and the threshold $\max\{\max_{s\in\mathcal{S}}(\vert  \mathfrak{F}_s\vert  -1)/\vert  \mathfrak{F}_s\vert  ,0\}$ becomes $0$. We can use the achievability of Theorem \ref{thm_1} and the strong converse in the subsection \ref{subsec_A} to verify the statement. 
Therefore in the following we assume that the inactive set is non-empty, $\mathcal{S}\neq \varnothing$. Our coding scheme is influenced by the one given in \cite{tian2008successive}.
\subsubsection{Construction of a testing scheme}
Recall that $\mathfrak{F}_s$ represents the distributions which have the same marginal distribution $P_{\mathcal{X},s}$ on $\mathcal{X}$, and $\vert  \mathfrak{F}_s\vert  $ is the number of these. For each inactive $s\in\mathcal{S}$ we partition the set $\mathcal{X}^n$ into $\vert  \mathfrak{F}_s\vert  $ sets $\{\mathcal{C}_n^{(ls)}\}_{l=1}^{\vert  \mathfrak{F}_s\vert  }$ such that for all sufficiently large $n$, we have $P_{\mathcal{X},s}^{\otimes n}(\mathcal{C}_n^{(ls)})>1-\epsilon$ for all $l\in [1:\vert  \mathfrak{F}_s\vert  ]$. This is possible\footnote{Let $0<\gamma<[1-\vert  \mathfrak{F}_s\vert  (1-\epsilon)]/2$ be an arbitrarily given number. Given a type class $T_P^n$ which is a subset of the strongly typical set $\mathcal{T}_{\gamma}^{n}(P_{\mathcal{X},s})$, we divide it into $\vert  \mathfrak{F}_s\vert  $ subsets such that each subset has cardinality of $\floor{\vert  T_{P}^n\vert  /\vert  \mathfrak{F}_s\vert  }$ and omit the remaining sequences. Enumerating over all type classes inside the typical set, the number of omitted typical sequences is upper bounded by $(n+1)^{\vert  \mathcal{X}\vert  }\vert  \mathfrak{F}_s\vert  $. Hence each of the constructed subsets has probability of at least $(1-2\gamma)/\vert  \mathfrak{F}_s\vert  >1-\epsilon$ for all sufficiently large $n$. The atypical sequences and the omitted typical sequences can be then assigned into these sets randomly.} because we have $\epsilon\in (\max_{s\in\mathcal{S}}(\vert  \mathfrak{F}_s\vert  -1)/\vert  \mathfrak{F}_s\vert  ,1)$. Furthermore for each $i\in \mathfrak{F}_s$ where $s\in\mathcal{S}$, let $l_i$ be its position inside the set $\mathfrak{F}_s$ according to the natural ordering.

For each of these inactive states $s$ and $i\in\mathfrak{F}_s$, let $(\phi_n^{(l_is)},\psi_n^{(l_is)})$ be a sequence of testing schemes to differentiate between $P_{Y_iX_i}^{\otimes n}$ and $Q_{Y_{j_i^{\star}}^n}\times Q_{X_{t_s^{\star}}^n}$ such that $\xi_i(R_c)-\gamma/3$ is achievable. Then similarly as in the conclusion of Theorem \ref{thm_1}, the false alarm probability of the likelihood ratio test also goes to zero
  \begin{align}
    \lim_{n\to\infty}\mathrm{Pr}\big\{&P_{Y_i^n\phi_n^{(l_is)}(X_i^n)}(Y_i^n,\phi_n^{(l_is)}(X_i^n))\nonumber\\
    &\leq e^{n(\xi_i(R_c)-\gamma/3)}Q_{Y_{j_i^{\star}}^n}\times Q_{\phi_n^{(l_is)}(X_{t_s^{\star}}^n)}(Y_i^n,\phi_n^{(l_is)}(X_i^n))\big\} = 0.
  \end{align}
  Note also that the following expressions
  \begin{align}
    \mathrm{Pr}\{P_{X_i}^{\otimes n}(X_i^n)\leq e^{n(d_s^{\mathrm{x}}+\gamma/3)}Q_{X_{t_s^{\star}}^n}(X_i^n)\}&\to 1,\nonumber\\
    \mathrm{Pr}\{P_{Y_i}^{\otimes n}(Y_i^n)\leq e^{n(d_i^{\mathrm{y}}+\gamma/3)}Q_{Y_{j_i^{\star}}^n}(Y_i^n)\}&\to 1,\label{reason_1}
    \end{align}
    hold due to either the weak law of large numbers or Theorem 1 in \cite{barron1985strong}.
  Rewriting the above expression we therefore obtain
  \begin{align}
  \lim_{n\to\infty}\mathrm{Pr}\big\{\iota_{P_{Y_i^n\phi_n^{(l_is)}(X_i^n)}}&(Y_i^n;\phi_n^{(l_is)}(X_i^n))\leq n(\xi_i(R_c)-d_{is}^{\mathrm{yx}}-\gamma)\big\} = 0.\label{eq_208}
\end{align}

\noindent For each active state $s$, $s\notin\mathcal{S}$, let $(\phi_n^s,\psi_n^s)$ be a sequence of testing schemes to differentiate between $\{P_{Y_iX_i}^{\otimes n}\}_{i\in\mathfrak{F}_s}$  and  $\{Q_{Y_j^n}\times Q_{X_t^n}\}$ such that $\theta_s(R_c)-\gamma/3 = \min_{\bar{i}\in\mathfrak{F}_s}\xi_{\bar{i}}(R_c)-\gamma/3$ is achievable. The existence of $(\phi_n^s,\psi_n^s)$ follows from Theorem \ref{thm_1}. Similarly from the conclusion of Theorem \ref{thm_1}, for each intersected set $\mathcal{I}^{(i)}(\theta_s(R_c)-\gamma/3)$, $i\in\mathfrak{F}_s$, we have
\begin{align}
\lim_{n\to\infty} P_{Y_i^n\phi_n^s(X_i^n)}[(\mathcal{I}^{(i)}(\theta_s(R_c)-\gamma/3))^c]=0.
\end{align}
From the definition of $\mathcal{I}^{(i)}(\theta_s(R_c)-\gamma/3)$ in \eqref{ie_expl} we further have
  \begin{align}
    \lim_{n\to\infty}\mathrm{Pr}\big\{&P_{Y_i^n\phi_n^s(X_i^n)}(Y_i^n,\phi_n^s(X_i^n))\nonumber\\
    &\leq e^{n(\theta_s(R_c)-\gamma/3)}Q_{Y_{j_i^{\star}}^n}\times Q_{\phi_n^s(X_{t_s^{\star}}^n)}(Y_i^n,\phi_n^s(X_i^n))\big\} = 0.
  \end{align}
Using \eqref{reason_1} again we obtain 
\begin{align}
\lim_{n\to\infty}\mathrm{Pr}\big\{&\iota_{P_{Y_i^n\phi_n^{s}(X_i^n)}}(Y_i^n;\phi_n^{s}(X_i^n))\nonumber\\
&\leq n(\min_{\bar{i}\in\mathfrak{F}_s}\xi_{\bar{i}}(R_c)-d_{is}^{\mathrm{yx}}-\gamma)\big\} = 0,\;\forall i\in\mathfrak{F}_s. \label{eq_210}
\end{align}
 Define the following auxiliary distribution
\begin{align}
  P_{Y^nX^n} = \sum_{i=1}^m\nu_iP_{Y_iX_i}^{\otimes n},\;\text{where}\; \nu_i>0,\;\text{and}\;\sum \nu_i = 1.
\end{align}
Let $P_{X^n}$ and $P_{Y^n}$ be the marginal distributions of $P_{Y^nX^n}$. Furthermore, let $P_{Y^nU}$ be the push-forward distribution resulting from applying $(\phi_n^{(l_is)})$ and $\phi_n^s$ to $\{P_{X_i^n}\}$, i.e.,
\begin{align}
  P_{Y^nU} = \sum_{s\in\mathcal{S}}\sum_{i\in\mathfrak{F}_s}\nu_iP_{Y_i^n\phi_n^{(l_is)}(X_i^n)} + \sum_{s\notin\mathcal{S}}\sum_{i\in\mathfrak{F}_s}\nu_iP_{Y_i^n\phi_n^{s}(X_i^n)}.
\end{align}
We use the same mapping $T(\cdot)$ to estimate $s$ before encoding the as in the proof of Theorem \ref{thm_1}.  If $\hat{s}=e$ we send $(e,1)$ to the decision center. If $\hat{s}\neq e$, we check if $x^n\in\mathcal{B}_{n,\gamma}^{\hat{s}}$ where 
\begin{align}
   \mathcal{B}_{n,\gamma}^{\hat{s}} = \{x^n\mid\min_{t\in [1:r]}\iota_{P_{X^n}\Vert Q_{X_t^n}}(x^n)>n(d_{\hat{s}}^{\mathrm{x}}-\gamma)\}.
\end{align}
If the condition is not fulfilled we send a special symbol $u=e^{\star}$. Assume that the condition $x^n\in\mathcal{B}_{n,\gamma}^{\hat{s}}$ holds. If $\hat{s}\in\mathcal{S}$ we  send the following message to the decision center $u = \bar{\phi}_n^{\hat{s}}(x^n) = \sum_{l=1}^{\vert  \mathfrak{F}_{\hat{s}}\vert  }\phi_n^{(l\hat{s})}(x^n)\mathbf{1}\{x^n\in\mathcal{C}_n^{(l\hat{s})}\cap\mathcal{B}_{n,\gamma}^{\hat{s}}\}$. If $\hat{s}\notin\mathcal{S}$ we send $u=\phi_n^{\hat{s}}(x^n)\mathbf{1}\{x^n\in\mathcal{B}_{n,\gamma}^{\hat{s}}\}$.
For $E = \min_{i\in [1:m]}\xi_i(R_c)-2\gamma$, we decide that the null hypothesis $H_0$ is true if $\hat{s}\neq e$, $u\neq e^{\star}$ and
\begin{align}
  \iota_{P_{Y^nU}}(y^n;u) + \min_{j\in[1:k]}\iota_{P_{Y^n}\Vert  Q_{Y_j^n}}(y^n)>n(E-d_{\hat{s}}^{\mathrm{x}}).
\end{align}
\subsubsection{Bounding error probabilities}
For each $i\in \mathfrak{F}_s$ where $s$ is active, $s\notin\mathcal{S}$, the false alarm probability is given by
\begin{align}
  \alpha_n^{(i)} &\leq \mathrm{Pr}\{\iota_{P_{Y^nU}}(Y_i^n;\phi_n^s(X_i^n)) + \min_{j\in[1:k]}\iota_{P_{Y^n}\Vert  Q_{Y_j^n}}(Y_i^n)\leq n(E-d_s^{\mathrm{x}})\}\nonumber\\
  &\quad+ P_{\mathcal{X},s}^{\otimes n}[(B_{n,\gamma}^{s})^c]+\mathrm{Pr}\{T(X_i^n)\neq s\}.\nonumber\\
                 &\leq \mathrm{Pr}\{\iota_{P_{Y^nU}}(Y_i^n;\phi_n^s(X_i^n))\leq n(E-d_{is}^{\mathrm{yx}}+\gamma)\}\nonumber\\
  &\quad+\mathrm{Pr}\{\min_{j\in[1:k]}\iota_{P_{Y^n}\Vert Q_{Y_j^n}}(\bar{Y}_i^n)<n(d_i^{\mathrm{y}}-\gamma)\}\nonumber\\
  &\quad + P_{\mathcal{X},s}^{\otimes n}[(B_{n,\gamma}^{s})^c]+\mathrm{Pr}\{T(X_i^n)\neq s\}.
\end{align}
Similarly, for each $i\in\mathfrak{F}_s$ where $s$ is not active, $s\in\mathcal{S}$, we have
\begin{align}
  \alpha_n^{(i)}&\leq  \mathrm{Pr}\{\iota_{P_{Y^nU}}(Y_i^n;\bar{\phi}_n^s(X_i^n)) + \min_{j\in[1:k]}\iota_{P_{Y^n}\Vert  Q_{Y_j^n}}(Y_i^n)\nonumber\\
  &\hspace{1cm}\leq n(E-d_s^{\mathrm{x}}),\;X_i^n\in \mathcal{B}_{n,\gamma}^{s}\}\nonumber\\
                &\quad+ P_{\mathcal{X},s}^{\otimes n}[(B_{n,\gamma}^{s})^c]+\mathrm{Pr}\{T(X_i^n)\neq s\}\nonumber\\
                &\leq \mathrm{Pr}\{\iota_{P_{Y^nU}}(Y_i^n;\phi_n^{(l_is)}(X_i^n)) \leq n(E-d_{is}^{\mathrm{yx}}+\gamma),\;X_i^n\in \mathcal{C}_{n}^{(l_is)}\cap\mathcal{B}_{n,\gamma}^{s}\}\nonumber\\
  &\quad+ \mathrm{Pr}\{\iota_{P_{Y^nU}}(Y_i^n;\bar{\phi}_n^{s}(X_i^n)) \leq n(E-d_{is}^{\mathrm{yx}}+\gamma),\;X_i^n\notin \mathcal{C}_{n}^{(l_is)}\}\nonumber\\
                &\quad+ \mathrm{Pr}\{\min_{j\in[1:k]}\iota_{P_{Y^n}\Vert Q_{Y_j^n}}(\bar{Y}_i^n)<n(d_i^{\mathrm{y}}-\gamma)\} \nonumber\\
                &\quad + P_{\mathcal{X},s}^{\otimes n}[(B_{n,\gamma}^{s})^c]+\mathrm{Pr}\{T(X_i^n)\neq s\}\nonumber\\
                &\leq \mathrm{Pr}\{\iota_{P_{Y^nU}}(Y_i^n;\phi_n^{(l_is)}(X_i^n)) \leq n(E-d_{is}^{\mathrm{yx}}+\gamma)\}\nonumber\\
                &\quad + \mathrm{Pr}\{\min_{j\in[1:k]}\iota_{P_{Y^n}\Vert Q_{Y_j^n}}(\bar{Y}_i^n)<n(d_i^{\mathrm{y}}-\gamma)\}\nonumber\\
  &\quad + P_{\mathcal{X},s}^{\otimes n}[(B_{n,\gamma}^{s})^c]+\mathrm{Pr}\{T(X_i^n)\neq s\} + \epsilon,
\end{align}
for all sufficiently large $n$. The last inequality holds since for all sufficiently large $n$ we have  $ P_{\mathcal{X},s}^{\otimes n}[(\mathcal{C}_{n}^{(l_is)})^c]\leq \epsilon$ by the definition of $\mathcal{C}_{n}^{(l_is)}$. 
Let $\gamma_n$ be a sequence such that $\gamma_n\to 0$ and $n\gamma_n\to\infty$ as $n\to \infty$. 
  Using the change of measure steps as in the proof of Theorem \ref{thm_1} we have
  \begin{align}
    &\mathrm{Pr}\{\iota_{P_{Y^nU}}(Y_i^n;\phi_n^s(X_i^n))\leq n(E-d_{is}^{\mathrm{yx}}+\gamma)\}\nonumber\\
    &\leq \mathrm{Pr}\{\iota_{P_{Y_i^n\phi_n^s(X_i^n)}}(Y_i^n;\phi_n^s(X_i^n))\leq n(E-d_{is}^{\mathrm{yx}}+\gamma-1/n\log{\nu_i}+2\gamma_n)\}\nonumber\\
    &\quad + 2e^{-n\gamma_n}, \nonumber\\
&    \mathrm{Pr}\{\iota_{P_{Y^nU}}(Y_i^n;\phi_n^{(l_is)}(X_i^n)) \leq n(E-d_{is}^{\mathrm{yx}}+\gamma)\}\nonumber\\
    &\leq \mathrm{Pr}\{\iota_{P_{Y_i^n}\phi_n^{(l_is)}(X_i^n)}(Y_i^n;\phi_n^{(l_is)}(X_i^n))\leq n(E-d_{is}^{\mathrm{yx}}+\gamma-1/n\log \nu_i+2\gamma_n)\}\nonumber\\
    &\quad + 2e^{-n\gamma_n}. 
  \end{align}  

Using \eqref{eq_208}, \eqref{eq_210} and Lemma \ref{lemma_mix_sup}, we obtain
  \begin{align}
    \lim_{n\to\infty}\alpha_n^{(i)} = 0,\;\forall i\in\mathfrak{F}_s,\;s\notin \mathcal{S},\nonumber\\
    \limsup_{n\to\infty}\alpha_n^{(i)}\leq \epsilon,\;\forall i\in\mathfrak{F}_s,\;s \in \mathcal{S}.
  \end{align}
This implies that
\begin{equation}
 \limsup_{n\to\infty}\alpha_n= \limsup_{n\to\infty}\max_{i\in [1:m]}\alpha_n^{(i)}= \max_{i\in [1:m]}\limsup_{n\to\infty}\alpha_n^{(i)}\leq \epsilon.
\end{equation}
Similarly as in the last part of the proof of Theorem \ref{thm_1}, by change of measure steps, we also have
\begin{align}
  \beta_n^{(jt)}\leq e^{-nE},\;\forall j\in [1:k],\; t\in [1:r].
\end{align}
Therefore for $\epsilon>\max_{s\in\mathcal{S}}(\vert  \mathfrak{F}_s\vert  -1)/\vert  \mathfrak{F}_s\vert  $ we have $E_{\epsilon}^{\star}(R_c)= \min_{i\in[1:m]}\xi_i(R_c)$ as the converse direction is straightforward. 
\subsection{Convergence of $\alpha_n^{(i^{\star})}$}
We assume that the optimality achieving index $s^{\star}$ is active, $s^{\star}\not\subset \mathcal{S}$, otherwise there is nothing to prove. Let $\mathcal{R}$, $R_{\mathrm{HT}}^{\mu}(P_{X_i^{\star}Y_{i^{\star}}})$ and  $R_{\mathrm{HT}}^{\mu,\alpha}(P_{X_{i^{\star}}Y_{i^{\star}}})$ be defined as in \eqref{region_R}, \eqref{r_mu} and \eqref{r_mu_alpha} with $m=1$ and $P_{X_{i^{\star}}Y_{i^{\star}}}$ in place of $\bar{P}$ and  $d_{i^{\star}s^{\star}}^{\mathrm{yx}}$ in place of $\bar{d}_i^{\star}$.

Let $(\phi_n,\psi_n)$ be an arbitrary sequence of testing schemes such that $\liminf_{n\to\infty}\frac{1}{n}\log\frac{1}{\beta_n}\geq E$ where $E = \xi_{i^{\star}}(R_c) + \tau$ for an arbitrary $\tau>0$ holds. Select $\gamma\in (0,\tau/5)$ small enough such that $(R+\gamma,E-3\gamma)\notin \mathcal{R}$.\\
Similarly as in the previous proofs we transform $(\phi_n,\psi_n)$ to a testing scheme $(\bar{\phi}_n,\bar{\psi}_n)$, where $\bar{\psi}_n$ is a deterministic mapping, for differentiating between $P_{Y_{i^{\star}}X_{i^{\star}}}^{\otimes n}$ and $P_{Y_{i^{\star}}}^{\otimes n}\times P_{X_{i^{\star}}}^{\otimes n}$. This can be done by using the typical sets $\mathcal{B}_{n,\gamma}^{s^{\star}}$ and $\mathcal{B}_{n,\gamma}^{(i^{\star})}$ defined in \eqref{bigB_def}. The resulting error probabilities are similarly bounded by
\begin{align}
  P_{Y_{i^{\star}}^n\bar{\phi}_n(X_{i^{\star}}^n)}(1-\bar{\psi}_n)&\leq \alpha_n^{(i^{\star})} + e^{n(E-\gamma)}\beta_n^{(j_{i^{\star}}^{\star}t_{s^{\star}}^{\star})}\nonumber\\
  &\quad + P_{X_{i^{\star}}}^{\otimes n}[(\mathcal{B}_{n,\gamma}^{s^{\star}})^c] + P_{Y_{i^{\star}}}^{\otimes n}[(\mathcal{B}_{n,\gamma}^{(i^{\star})})^c],\nonumber\\
  P_{Y_{i^{\star}}^n}\times P_{\bar{\phi}_n(X_{i^{\star}}^n)}(\bar{\psi}_n)&\leq e^{-n(E-d_{i^{\star}s^{\star}}^{\mathrm{yx}}-3\gamma)}.\label{istarerror}
\end{align}
Let $\bar{\mathcal{A}}_n$ be the acceptance region of $\bar{\psi}_n$.
We argue that there exists a $\lambda>0$ such that for all $n\geq n_0(\gamma)$ we have
\begin{align}
  P_{Y_{i^{\star}}^n\bar{\phi}_n(X_{i^{\star}}^n)}(\bar{\mathcal{A}}_n)\leq e^{-\lambda n}.\label{expo_clause}
\end{align}
Assume otherwise that for all $\lambda>0$ there exists an $n\geq n_0(\gamma)$ such that
\begin{align}
   P_{Y_{i^{\star}}^n\bar{\phi}_n(X_{i^{\star}}^n)}(\bar{\mathcal{A}}_n)> e^{-\lambda n}.
\end{align}
Similarly for notation simplicity we define $\bar{B}_n = \bigcup_{u}\{\bar{\phi}_n^{-1}(u)\}\times\bar{\mathcal{A}}_{n,u}$ as well as
\begin{align}
  P_{\tilde{X}^n\tilde{Y}^n} = \frac{P_{X_{i^{\star}}Y_{i^{\star}}}^{\otimes n}(x^n,y^n)}{P_{X_{i^{\star}}Y_{i^{\star}}}^{\otimes n}(\bar{B}_n)}\mathbf{1}\{(x^n,y^n)\in\bar{B}_n\}.
\end{align}
We then have
\begin{align}
  D(P_{\tilde{X}^n\tilde{Y}^n}\Vert  P_{X_{i^{\star}}Y_{i^{\star}}}^{\otimes n})\leq \lambda n.
\end{align}
For $(y^n,u)\in\bar{\mathcal{A}}_n$ it can be seen that
 \begin{align*}
 P_{Y_{i^{\star}}^n\bar{\phi}_n(X_{i^{\star}}^n)}(y^n,u) &= P_{X_{i^{\star}}Y_{i^{\star}}}^{\otimes n}(\bar{B}_n)P_{\tilde{Y}^n\bar{\phi}_n(\tilde{X}^n)}(y^n,u)\nonumber\\
 &\geq e^{-\lambda n}P_{\tilde{Y}^n\bar{\phi}_n(\tilde{X}^n)}(y^n,u)
 \end{align*}
  holds whereas for $(y^n,u)\notin \bar{\mathcal{A}}_n$ we have $P_{\tilde{Y}^n\bar{\phi}_n(\tilde{X}^n)}(y^n,u) = 0$. Therefore for all $(y^n,u)\in \bar{\mathcal{A}}_n$ the following inequalities hold
\begin{align}
P_{Y_{i^{\star}}}^{\otimes n}(y^n)\geq e^{-\lambda n} P_{\tilde{Y}^n}(y^n),\; P_{\bar{\phi}_n(X_{i^{\star}}^n)}(u)\geq e^{-\lambda n} P_{\bar{\phi}_n(\tilde{X}^n)}(u). 
\end{align}
This implies further that we also have
\begin{align}
&  \log\frac{1}{P_{Y_{i^{\star}}^n}\times P_{\bar{\phi}_n(X_{i^{\star}}^n)}(\bar{\mathcal{A}}_n)}\nonumber\\
  &\leq \sum_{(y^n,u)\in\bar{\mathcal{A}}_n}P_{\tilde{Y}^n\bar{\phi}_n(\tilde{X}^n)}(y^n,u)\log\frac{P_{\tilde{Y}^n\bar{\phi}_n(\tilde{X}^n)}(y^n,u)}{P_{Y_{i^{\star}}^n}(y^n)P_{\bar{\phi}_n(X_{i^{\star}}^n)}(u)}\nonumber\\
                                                                                                &\leq \sum_{(y^n,u)\in\bar{\mathcal{A}}_n}P_{\tilde{Y}^n\bar{\phi}_n(\tilde{X}^n)}(y^n,u)\log\frac{P_{\tilde{Y}^n\bar{\phi}_n(\tilde{X}^n)}(y^n,u)}{e^{-2\lambda n} P_{\tilde{Y}^n}(y^n)P_{\bar{\phi}_n(\tilde{X}^n)}(u)}\nonumber\\
  &= I(\tilde{Y}^n;\bar{\phi}_n(\tilde{X}^n)) +  2\lambda n.
\end{align}
Hence we obtain
\begin{align}
  n(E-3\gamma)\leq I(\tilde{Y}^n;\bar{\phi}_n(\tilde{X}^n)) + d_{i^{\star}s^{\star}}^{\mathrm{yx}}+ 2\lambda n,
\end{align}
where $(\tilde{Y}^n,\tilde{X}^n)\sim P_{\tilde{X}^n\tilde{Y}^n}$.
By using the same lines of arguments as from \eqref{eq_202} to \eqref{eq_206} we have for given positive $\alpha$ and $\mu$
\begin{align}
  (R_c+\gamma) - \mu(E-3\gamma)\geq R_{\mathrm{HT}}^{\mu,\alpha}(P_{X_{i^{\star}}Y_{i^{\star}}}) - ((\alpha+1)+2\mu)\lambda.
\end{align}
As $\mathcal{R}$ is a closed convex set and $(R+\gamma,E-3\gamma)\notin\mathcal{R}$ holds, by hyperplane separation theorem there exist positive numbers $\mu$ and $\nu$ such that $(R_c+\gamma) - \mu(E-3\gamma)< R_{\mathrm{HT}}^{\mu}(P_{X_{i^{\star}}Y_{i^{\star}}}) -2\nu$ holds. Furthermore there also exists an $\alpha$ such that $R_{\mathrm{HT}}^{\mu}(P_{X_{i^{\star}}Y_{i^{\star}}})<R_{\mathrm{HT}}^{\mu,\alpha}(P_{X_{i^{\star}}Y_{i^{\star}}})+\nu$. Then we obtain for such $\alpha,\mu$
\begin{align}
  R_{\mathrm{HT}}^{\mu,\alpha}(P_{X_{i^{\star}}Y_{i^{\star}}})-\nu&\geq R_{\mathrm{HT}}^{\mu}(P_{X_{i^{\star}}Y_{i^{\star}}})-2\nu\nonumber\\
  &> (R_c+\gamma) - \mu(E-3\gamma)\nonumber\\
  &\geq  R_{\mathrm{HT}}^{\mu,\alpha}(P_{X_{i^{\star}}Y_{i^{\star}}}) - ((\alpha+1)+2\mu)\lambda.
\end{align}
This inequality is violated for $\lambda<\nu/((\alpha+1)+2\mu)$.
Then \eqref{expo_clause} and \eqref{istarerror} imply that  
\begin{align}
\lim_{n\to\infty}\alpha_n^{(i^{\star})} = 1.
\end{align}

\bibliographystyle{IEEEtran}
\bibliography{IEEEabrv,references}

\begin{thebibliography}{10}
\providecommand{\url}[1]{#1}
\csname url@samestyle\endcsname
\providecommand{\newblock}{\relax}
\providecommand{\bibinfo}[2]{#2}
\providecommand{\BIBentrySTDinterwordspacing}{\spaceskip=0pt\relax}
\providecommand{\BIBentryALTinterwordstretchfactor}{4}
\providecommand{\BIBentryALTinterwordspacing}{\spaceskip=\fontdimen2\font plus
\BIBentryALTinterwordstretchfactor\fontdimen3\font minus
  \fontdimen4\font\relax}
\providecommand{\BIBforeignlanguage}[2]{{%
\expandafter\ifx\csname l@#1\endcsname\relax
\typeout{** WARNING: IEEEtran.bst: No hyphenation pattern has been}%
\typeout{** loaded for the language `#1'. Using the pattern for}%
\typeout{** the default language instead.}%
\else
\language=\csname l@#1\endcsname
\fi
#2}}
\providecommand{\BIBdecl}{\relax}
\BIBdecl

\bibitem{ahlswede1986hypothesis}
R.~Ahlswede and I.~Csisz{\'a}r, ``Hypothesis testing with communication
  constraints,'' \emph{IEEE Transactions on Information Theory}, vol.~32,
  no.~4, pp. 533--542, 1986.

\bibitem{berger1979hypothesis}
T.~Berger, ``Decentralized estimation and decision theory,'' in \emph{IEEE 7th
  Spring Workshop on Inf. Theory}, Mt. Kisco, NY, September 1979.

\bibitem{han1987hypothesis}
T.~Han, ``Hypothesis testing with multiterminal data compression,'' \emph{IEEE
  Transactions on Information Theory}, vol.~33, no.~6, pp. 759--772, 1987.

\bibitem{shimokawa1994error}
H.~Shimokawa, T.~S. Han, and S.~Amari, ``Error bound of hypothesis testing with
  data compression,'' in \emph{Proceedings of 1994 IEEE International Symposium
  on Information Theory}.\hskip 1em plus 0.5em minus 0.4em\relax IEEE, 1994, p.
  114.

\bibitem{shalaby1992multiterminal}
H.~M. Shalaby and A.~Papamarcou, ``Multiterminal detection with zero-rate data
  compression,'' \emph{IEEE Transactions on Information Theory}, vol.~38,
  no.~2, pp. 254--267, 1992.

\bibitem{rahman2012optimality}
M.~S. Rahman and A.~B. Wagner, ``On the optimality of binning for distributed
  hypothesis testing,'' \emph{IEEE Transactions on Information Theory},
  vol.~58, no.~10, pp. 6282--6303, 2012.

\bibitem{tian2008successive}
C.~Tian and J.~Chen, ``Successive refinement for hypothesis testing and
  lossless one-helper problem,'' \emph{IEEE Transactions on Information
  Theory}, vol.~54, no.~10, pp. 4666--4681, 2008.

\bibitem{watanabe2017neyman}
S.~Watanabe, ``Neyman--pearson test for zero-rate multiterminal hypothesis
  testing,'' \emph{IEEE Transactions on Information Theory}, vol.~64, no.~7,
  pp. 4923--4939, 2017.

\bibitem{mclachlanpeel}
G.~J. McLachlan and D.~Peel, \emph{Finite mixture models}.\hskip 1em plus 0.5em
  minus 0.4em\relax New York: John Wiley \& Sons, 2000.

\bibitem{chen1996general}
P.-N. Chen, ``General formulas for the neyman-pearson type-ii error exponent
  subject to fixed and exponential type-i error bounds,'' \emph{IEEE
  Transactions on Information Theory}, vol.~42, no.~1, pp. 316--323, 1996.

\bibitem{hanspectrum}
T.~S. Han, \emph{Information-Spectrum Methods in Information Theory}.\hskip 1em
  plus 0.5em minus 0.4em\relax Berlin Heidelberg: Springer-Verlag Berlin
  Heidelberg, 2003.

\bibitem{han2018first}
T.~S. Han, R.~Nomura \emph{et~al.}, ``First-and second-order hypothesis testing
  for mixed memoryless sources,'' \emph{Entropy}, vol.~20, no.~3, p. 174, 2018.

\bibitem{ritchie2020consistent}
A.~Ritchie, R.~A. Vandermeulen, and C.~Scott, ``Consistent estimation of
  identifiable nonparametric mixture models from grouped observations,''
  \emph{Advances in Neural Information Processing Systems}, vol.~33, pp.
  11\,676--11\,686, 2020.

\bibitem{elmore2004estimating}
R.~T. Elmore, T.~P. Hettmansperger, and H.~Thomas, ``Estimating component
  cumulative distribution functions in finite mixture models,''
  \emph{Communications in Statistics-Theory and Methods}, vol.~33, no.~9, pp.
  2075--2086, 2004.

\bibitem{cruz2004semiparametric}
I.~Cruz-Medina, T.~Hettmansperger, and H.~Thomas, ``Semiparametric mixture
  models and repeated measures: the multinomial cut point model,''
  \emph{Journal of the Royal Statistical Society: Series C (Applied
  Statistics)}, vol.~53, no.~3, pp. 463--474, 2004.

\bibitem{wei2020convergence}
Y.~Wei and X.~Nguyen, ``Convergence of de finetti's mixing measure in latent
  structure models for observed exchangeable sequences,'' \emph{arXiv preprint
  arXiv:2004.05542}, 2020.

\bibitem{pal2002noise}
C.~Pal, B.~Frey, and T.~Kristjansson, ``Noise robust speech recognition using
  gaussian basis functions for non-linear likelihood function approximation,''
  in \emph{2002 IEEE International Conference on Acoustics, Speech, and Signal
  Processing}, vol.~1.\hskip 1em plus 0.5em minus 0.4em\relax IEEE, 2002, pp.
  I--405.

\bibitem{anandkumar2012method}
A.~Anandkumar, D.~Hsu, and S.~M. Kakade, ``A method of moments for mixture
  models and hidden markov models,'' in \emph{Conference on Learning
  Theory}.\hskip 1em plus 0.5em minus 0.4em\relax JMLR Workshop and Conference
  Proceedings, 2012, pp. 33--1.

\bibitem{jordan2010lec}
M.~I. Jordan, ``{Stat260: Bayesian Modeling and Inference, Lecture 1: History
  and De Finetti’s Theorem},'' 2010.

\bibitem{kirsch2019elementary}
W.~Kirsch, ``An elementary proof of de finetti’s theorem,'' \emph{Statistics
  \& Probability Letters}, vol. 151, pp. 84--88, 2019.

\bibitem{vu2021hypothesis}
M.~T. Vu, T.~J. Oechtering, and M.~Skoglund, ``Hypothesis testing and
  identification systems,'' \emph{IEEE Transactions on Information Theory},
  vol.~67, no.~6, pp. 3765--3780, 2021.

\bibitem{wyner1975source}
A.~Wyner, ``On source coding with side information at the decoder,'' \emph{IEEE
  Transactions on Information Theory}, vol.~21, no.~3, pp. 294--300, 1975.

\bibitem{ahlswede1975source}
R.~Ahlswede and J.~K{\"o}rner, ``Source coding with side information and a
  converse for degraded broadcast channels,'' \emph{IEEE Transactions on
  Information Theory}, vol.~21, no.~6, pp. 629--637, 1975.

\bibitem{barron1985strong}
A.~R. Barron, ``The {S}trong {E}rgodic {T}heorem for {D}ensities: {G}eneralized
  {S}hannon-{M}c{M}illan-{B}reiman {T}heorem,'' \emph{Annals of Probability},
  vol.~13, no.~4, pp. 1292--1303, 1985.

\bibitem{verdu2012non}
S.~Verd{\'u}, ``Non-asymptotic achievability bounds in multiuser information
  theory,'' in \emph{Communication, Control, and Computing (Allerton), 2012
  50th Annual Allerton Conference on}.\hskip 1em plus 0.5em minus 0.4em\relax
  IEEE, 2012, pp. 1--8.

\bibitem{tyagi2019strong}
H.~Tyagi and S.~Watanabe, ``Strong converse using change of measure
  arguments,'' \emph{IEEE Transactions on Information Theory}, vol.~66, no.~2,
  pp. 689--703, 2019.

\bibitem{csiszar2011information}
I.~Csiszar and J.~K{\"o}rner, \emph{Information theory: coding theorems for
  discrete memoryless systems}.\hskip 1em plus 0.5em minus 0.4em\relax
  Cambridge: Cambridge University Press, 2011.

\bibitem{witsenhausen1975conditional}
H.~Witsenhausen and A.~Wyner, ``A conditional entropy bound for a pair of
  discrete random variables,'' \emph{IEEE Transactions on Information Theory},
  vol.~21, no.~5, pp. 493--501, 1975.

\bibitem{lessnoisy}
J.~K{\"o}rner and K.~Marton, ``Comparison of two noisy channels,'' in
  \emph{Topics in Information Theory}, vol.~16.\hskip 1em plus 0.5em minus
  0.4em\relax Amsterdam: North Holland, 1977, pp. 411--424.

\bibitem{miyake1995coding}
S.~Miyake and F.~Kanaya, ``Coding theorems on correlated general sources,''
  \emph{IEICE Transactions on Fundamentals of Electronics, Communications and
  Computer Sciences}, vol.~78, no.~9, pp. 1063--1070, 1995.

\bibitem{polyanskiy2014lecture}
Y.~Polyanskiy and Y.~Wu, ``Lecture notes on information theory,'' \emph{MIT
  (6.441), UIUC (ECE 563)}, 2017.

\end{thebibliography}
\end{document}